\documentclass[12pt]{article}
\usepackage{amssymb,amsmath,amsthm,amsxtra,amsfonts}
\usepackage{mathabx}
\usepackage[margin={1in}]{geometry}
\usepackage{aliascnt}
\usepackage{color}
\usepackage{url}
\usepackage[usenames,dvipsnames,table]{xcolor}
\usepackage[hypertexnames=false,colorlinks=true,citecolor=Blue,urlcolor=Blue,linkcolor=Blue]{hyperref}
\usepackage{subfigure}
\usepackage{bbm}
\usepackage{dsfont}
\usepackage{graphicx}
\usepackage{mathrsfs}
\usepackage{tikz}
\usetikzlibrary{plotmarks,arrows,calc,patterns}
\usepackage{pgfplots}
\usepackage[longnamesfirst,round]{natbib}
\usepackage{threeparttable}
\usepackage{booktabs}
\usepackage{setspace}
\usepackage{pdflscape}
\usepackage{easyReview}
\usepackage{enumitem}
\usepackage{multirow}
\usepackage{algorithm}
\usepackage{algpseudocode}
\usepackage{stmaryrd}

\newtheorem{theorem}{Theorem}[section]
\newtheorem{proposition}{Proposition}[section]
\newtheorem{lemma}{Lemma}[section]

\newtheorem{remark}{Remark}[section]
\newtheorem{assumption}{Assumption}[section]
\newtheorem*{assumption*}{Assumption}
\newtheorem{definition}{Definition}[section]

\usepackage{xr}

\newcommand{\R}{\ensuremath{\mathbb{R}}}

\newcommand{\Nn}{\ensuremath{\mathbb{N}}}

\newcommand{\E}{\ensuremath{\mathds{E}}}

\newcommand{\dx}{\ensuremath{\mathrm{d}}}
\newcommand{\Var}{\ensuremath{\mathrm{Var}}}

\newcommand{\one}{\ensuremath{\mathds{1}}}

\graphicspath{{threewayfactors/}}
\begin{document}

\title{\bf Tensor PCA for Factor Models\footnote{We would like to thank Alex Belloni, Peter Bossaerts, George Kapetanios, Oliver Linton, Alexei Onatski, Mirco Rubin, Giovanni Urga, as well as seminar participants at the Bayes Business School, University of Cambridge, and Triangle Econometrics Conference for helpful comments.}}
\author{
	Andrii Babii\footnote{Department of Economics, University of North Carolina at Chapel Hill - Gardner Hall, CB 3305 Chapel Hill, NC
		27599-3305. Email: \href{mailto:babii.andrii@gmail.com}{babii.andrii@gmail.com}.} \\
	\textit{\normalsize UNC Chapel Hill}  \and Eric Ghysels\footnote{Department of Economics, University of North Carolina at Chapel Hill - Gardner Hall, CB 3305 Chapel Hill, NC
		27599-3305 and Department of Finance, Kenan-Flagler Business School, Email: \href{mailto:eghysels@gmail.com}{eghysels@gmail.com}.} \\
	\textit{\normalsize UNC Chapel Hill}  \and Junsu Pan\footnote{Department of Economics, University of North Carolina at Chapel Hill - Gardner Hall, CB 3305 Chapel Hill, NC
		27599-3305. Email: \href{mailto:junsupan@live.unc.edu}{junsupan@live.unc.edu}.} \\
	\textit{\normalsize UNC Chapel Hill} 
}
\maketitle

\begin{abstract}
Modern empirical analysis often relies on high-dimensional panel datasets with non-negligible cross-sectional and time-series correlations. Factor models are natural for capturing such dependencies. A tensor factor model describes the $d$-dimensional panel as a sum of a reduced rank component and an idiosyncratic noise, generalizing traditional factor models for two-dimensional panels. We consider a tensor factor model corresponding to the notion of a reduced multilinear rank of a tensor. We show that for a strong factor model, a simple tensor principal component analysis algorithm is optimal for estimating factors and loadings. When the factors are weak, the convergence rate of simple TPCA can be improved with alternating least-squares iterations. We also provide inferential results for factors and loadings and propose the first test to select the number of factors. The new tools are applied to the problem of imputing missing values in a multidimensional panel of firm characteristics.
\end{abstract}

\begin{keywords}
Multidimensional panel data, tensors, cross-sectional dependence, dynamic networks, spatial data, factor models, principal component analysis, asset pricing, imputation, firm characteristics.
\end{keywords}

%\begin{jel}
%	\; ?, ?.
%\end{jel}
\thispagestyle{empty}

\setcounter{page}{0}
\newpage

\section{Introduction}
Modern empirical analysis often relies on high-dimensional panel datasets with non-negligible dependencies across one or several dimensions. Various economic reasons exist for such dependencies, including common shocks in macroeconomics and finance, interactions in a network, or spatial distribution of economic activity. Since their original introduction in the psychology literature by \cite{spearmen1904general}, factor models have become a ubiquitous tool in economics for modeling dependencies in two-dimensional panel datasets. 

\smallskip

Among the many applications of factor analysis in macroeconomics and finance, we may cite asset pricing, see \cite{ross1976arbitrage} and \cite{chamberlain1982arbitrage}; business cycle analysis \cite{sargent1977business}; industrial production analysis, see \cite{foerster2011sectoral} and \cite{andreou2019inference}; forecasting with big data; see \cite{stock2002forecasting}. Factor models are also widely used for causal inference in panel data; see \cite{pesaran2006estimation}, \cite{bai2009panel}, \cite{abadie2010synthetic}, \cite{gobillon2016regional}, \cite{athey2021matrix}, or more generally to model unobserved heterogeneity in microeconometrics; see \cite{cunha2010estimating} and \cite{bonhomme2010generalized}. Lastly, structural economic models often naturally lead to a factor structure as seen in \cite{liu2024dynamic} for the input-output production networks or \cite{lewbel1991rank} for the demand systems.

\smallskip

Traditional factor models are designed for two-dimensional panel datasets consisting of cross-sectional units tracked over time, where both dimensions can potentially be large; see \cite{stock2002} and \cite{bai2003inferential}.\footnote{We will use the term `traditional' for $2$-dimensional factor models applicable to 2-dimensional panel datasets.} However, many economic data feature more than two dimensions. For example, the input-output production or trade network usually evolves over time, which is not captured by the traditional factor model.\footnote{See, e.g.\ \cite{graham2020network} for a recent review of static networks.} A panel of macroeconomic indicators often involves regional aggregation of state- or county-level observations, introducing a geographical dimension in addition to the traditional cross-section and time series. Asset pricing models for the cross-section of equities typically involve characteristic-based portfolio sorts, introducing a third dimension. The sorting into deciles is common, but only the lowest and highest deciles are used and combined in a high minus low return spread. Going beyond national borders introduces an international dimension to macro and financial datasets, resulting in a four-dimensional data structure. Each of these examples illustrates that to obtain matrix representations of the data, we often aggregate the multidimensional panel datasets, suppressing more granular information.

\smallskip

Principal component analysis (PCA) is a commonly used method for identifying and estimating traditional factor models for two-dimensional panel data sets; see \cite{pearson1901liii}. PCA extracts latent factors and their loadings using either the singular value decomposition (SVD) of the original panel dataset collected in a matrix or, equivalently, the eigendecomposition of the associated sample covariance matrices; see \cite{jolliffe2002principal} for a review of PCA and factor analysis.

\smallskip

A $d$-way tensor is a $d$-dimensional array generalizing vectors and matrices introduced in \cite{ricci1900methodes}.\footnote{Tensor tools have already been used by economists, e.g., for the identification of finite mixture models; see \cite{bonhomme2016estimating}.} In this paper, we consider an extension of traditional factor models to multidimensional datasets, called tensor factor models. Similarly to their $2$-way counterpart, the $d$-way tensor factor model can be used to identify the latent factors driving correlations in tensor datasets. In traditional factor models, correlations are captured with a matrix of factors in the time dimension and a matrix of factor loadings in the cross-sectional dimension. Likewise, one can decompose a tensor into a collection of matrices with respect to each dimension. Hence, the matrix component in the `time' dimension correspond to factors that vary over time, and the matrix component in the other dimensions correspond to loadings determining the heterogeneous exposure of each dimension to the latent factors. 

\smallskip

Traditional factor models can also be viewed as decomposing a matrix representing a panel dataset as a sum of a low-rank matrix product of loading and factor matrices and a matrix of idiosyncratic shocks. The low-rank component is then approximated using the truncated SVD decomposition of the observed data. The $d$-way factor model also describes a $d$-way tensor dataset as a sum of a low-rank component and an idiosyncratic shocks tensor. However, there several different ways to describe the rank of a tensor. The two most widely known definitions of tensor ranks are called the Canonical Polyadic (CP) rank and the \textit{multilinear rank}; see \cite{carroll1970analysis}, \cite{harshman1970foundations}, and \cite{tucker1966some}.

\smallskip

The multilinear rank of a $d$-dimensional tensor is described by the $d$-tuple $(R_1,R_2\dots,R_d)$ corresponding to ranks of each of its mode-$j$ matricizations and is more general than than the CP rank.\footnote{The concepts of multilinear rank and CP rank for tensors are implicitly used in \cite{hitchcock1927expression}. In fact, \cite{hitchcock1927expression} considered even more general concept of a tensor rank based on arbitrary matricizations, not just along one of its $d$ modes.}  It leads to the so-called Tucker model that has been previously considered in statistics literature; see \cite{han2020tensor}, \cite{wang2022high}, \cite{chen2022factor}, \cite{han2022rank}, and \cite{zhang2018tensor} among others. Our paper introduces the factor dynamics differently in contrast to these papers. More importantly, we show that under the strong factor model assumption, commonly used in economics and finance, the simple PCA-type estimators for tensors have optimal convergence rates and derive the corresponding asymptotic distributions.

\smallskip

The simple tensor PCA (TPCA) estimators are easy to compute and consists of reshaping (or matricizing) a $d$-way tensor into $d$ matrices along each of its $d$ dimensions and applying the standard PCA these matrices. For this simple TPCA estimators, we (a) demonstrate that they can identify and consistently estimate the factors and loadings; (b) describe the associated convergence rates and  asymptotic distribution; (c) show that tensor dimensions can improve the estimation accuracy for factors/loadings and the estimator is rate-optimal for strong factor models; and (d) develop a formal test for the number of factors in the tensor factor model. For the moderately weak factor model, the TPCA procedure is not optimal and we also consider additional results for the optimal iterative alternating least-squares (ALS) algorithm, where the TPCA estimators are used as a starting value. Our theoretical results rely on the powerful perturbation theory results for singular subspaces and singular values which are not commonly used in econometrics, cf. \cite{bai2023approximate} and references therein.

\smallskip

Monte Carlo simulations support our asymptotic results in finite samples. We find that the $d$-way tensor factor model reduces dimensions more efficiently than the naively pooled $2$-way factor model. We also verify the convergence rates and the distribution theory.

\smallskip

In the empirical application we study the issue of missing firm characteristics in widely used data sources which are the cornerstone of research in finance.
\cite{bryzgalova2022missing} provide a comprehensive analysis of missing data in firm characteristics, and propose a statistical model for imputing missing values and investigate the impact of missingness on asset returns. Firm characteristic data is a 3-dimension tensor, with (a) firm, (b) characteristic and (c) time as the three dimensions.
The statistical methods proposed by \cite{bryzgalova2022missing} do not exploit the tensor data structure, while in this empirical application we do and we show that using our proposed tensor data models improve on their time series of cross-sectional models.

\smallskip

The paper is organized as follows. Section \ref{sec:tensorfac} introduces a new class of tensor factor models. The convergence rates and large sample distributions of the TPCA estimator are covered in Section \ref{sec:PCAtensor} which also covers the improved ALS estimator.  Section \ref{sec:test} presents a novel testing procedure for the number of factors. Small sample simulation evidence is reported in Section \ref{sec:sim}. An illustrative empirical example appears in Section \ref{sec:empirical}. Section~\ref{sec:conclusions} concludes. Lastly, in the Appendix, we provide proofs for all the main and auxiliary results.

\paragraph{Notation:} For two matrices $A\in\mathbb{R}^{N_1\times R_1}$ and $B\in\mathbb{R}^{N_2\times R_2}$, we use $A\otimes B\in\mathbb{R}^{N_1N_2\times R_1R_2}$ to denote their Kronecker product. In addition, for a collection of matrices $(\Lambda_j)_{j=1}^d$, we define $\bigotimes_{k\ne j}\Lambda_k=\Lambda_d\otimes\dots\otimes \Lambda_{j+1}\otimes \Lambda_{j-1}\otimes\dots\otimes \Lambda_1$. For a positive integer $p$, we use $I_p$ to denote the $p\times p$ identity matrix. For a matrix $A$, we use $P_A = A(A^\top A)^\dagger A^\top$ to denote the the projection matrix on the column space of $A$, where $^\dagger$ denotes the generalized inverse.  For two sequences $(a_n)_{n\in\Nn}$ and $(b_n)_{n\in\Nn}$, we write $a_n\lesssim b_n$ if and only if there exists $C<\infty$ such that $a_n\leq Cb_n$ for all $n\in\Nn$. The operator norm of a matrix $A$ is defined as $\|A\|_{\rm op}=\sup_{\|x\|=1}\|Ax\|$, where $\|.\|$ is the Euclidean norm. More generally, we use $\|.\|_p$ to denote the $\ell_p$ norm. For two tensors $A,B\in\R^{N_1\times\dots\times N_d}$, the Frobenius inner product is defined as $\langle A,B\rangle_F = \sum_{i_1,\dots,i_d}A_{i_1,\dots,i_d}B_{i_1,\dots,i_d}$. Let $\|A\|_{\rm F} = \sqrt{\langle A,A\rangle_{\rm F}}$ be the Frobenius norm of a tensor $A$ induced by the inner product. For a matrix $A$ with columns $(a_1,\dots,a_n)$, the $\ell_{2,1}$ matrix norm defined as $\|A\|_{2,1}=\sum_{j=1}^n\|a_j\|$. We also use $\mathbb{O}_{N}$ to denote the set of $N\times N$ orthogonal matrices, i.e. matrices $A$ such that $A^\top A = I_{N}$, where $I_N$ is $N\times N$ identity matrix. Lastly, for $a,b\in\R$, put $a\wedge b=\min(a,b)$ and $a\vee b = \max(a,b)$.

\section{Tensor Factor Models \label{sec:tensorfac}}
Traditional factor models apply to $2$-dimensional panel data represented by a matrix $\mathbf{Y}\in \R^{N\times T}$. The factor model with $R$ factors can be expressed as a sum of a low-rank matrix and a matrix of idiosyncratic shocks:
\begin{equation}\label{eq:2-way}
	\mathbf{Y} = \Lambda F^\top + \mathbf{U},\qquad \E\mathbf{U}=0,
\end{equation}
where $F\in\R^{T\times R}$ is a matrix of $R$ factors, $\Lambda\in\R^{N\times R}$ is a matrix of corresponding factor loadings, and $\mathbf{U}\in\mathbb{R}^{N\times T}$ are idiosyncratic random shocks.\footnote{The factors and loadings can be random, but all results are stated conditional on their realization.} The estimation of factors and their loadings can be done via PCA which can be computed with the singular-value decomposition (SVD) of $\mathbf{Y}$. More precisely, the fators/loadings can be estimated as the leading $R$ right/left singular vectors of $\mathbf{Y}$. Equivalently, the loadings can be estimated as the leading $R$ eigenvectors of $\mathbf{Y}\mathbf{Y}^\top$ while and the factors as the leading $R$ eigenvectors of $\mathbf{Y}^\top\mathbf{Y}$.

\smallskip

A $d$-dimensional panel dataset can be represented by a  $d$-way tensor, denoted $\mathbf{Y}\in\R^{N_1\times\ldots\times N_d}$. We can describe $\mathbf{Y}$ by enumerating all its elements along the $d$ ways (or modes):
\begin{equation*}
	\mathbf{Y} = \left\{y_{i_1,i_2,\dots, i_d},\; 1\leq i_j\leq N_j,\; 1\leq j\leq d\right\};
\end{equation*}
see Figure~\ref{appfig:tensors} for a graphical illustration. Similarly to a 2-way factor model, we can define a $d$-way factor model for a $d$-way tensor $\mathbf{Y}\in\mathbb{R}^{N_1\times\dots\times N_d}$ as follows
\begin{equation*}
	\mathbf{Y} = \text{low-rank tensor }\mathbf{X} + \mathbf{U}.
\end{equation*}
Unlike for matrices, the rank of a tensor can be defined in several different ways. To describe the notion of rank used in this paper, we need to introduce the notion of tensor matricization.

\begin{figure}[!ht]
	\centering
	\caption{A scalar, and tensors of order $1$ (vector), $2$ (matrix), and $3$}
	\label{appfig:tensors}
	\includegraphics[width=0.9\textwidth]{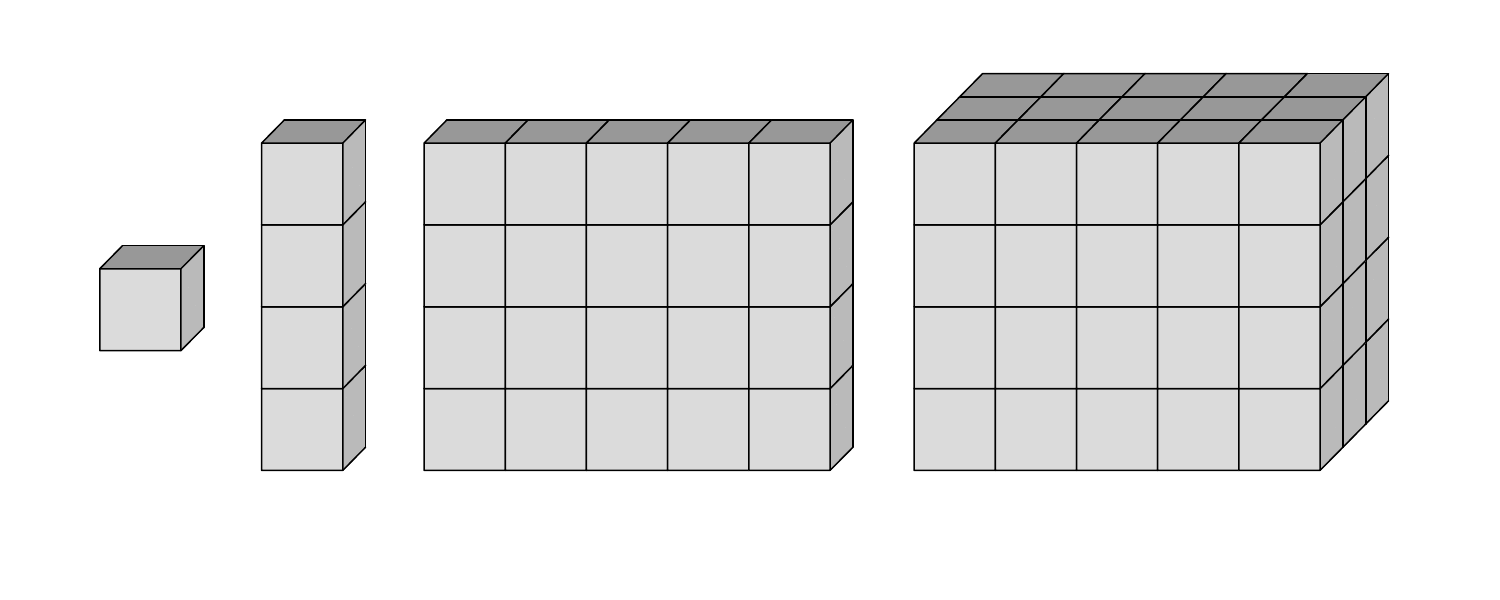}
\end{figure}

A tensor matricization operation can be described in terms of tensor \textit{fibers}---a generalization of matrix rows and columns. A fiber is defined by fixing all but one of its dimensions, e.g., a matrix column is a mode-$1$ fiber and a matrix row is a mode-$2$ fiber. The mode-$j$ fibers of a higher-order tensor are defined similarly; see Figure~\ref{appfig:fibers} for an illustration in the case of a $3$-way tensor. 
\begin{figure}[h]
	\centering
	\caption{Mode-$1,2$ and $3$ \textit{fibers} of a $4\times 5\times 3$ tensor}
	\label{appfig:fibers}
	\vskip0.1in
	\includegraphics[width=\textwidth]{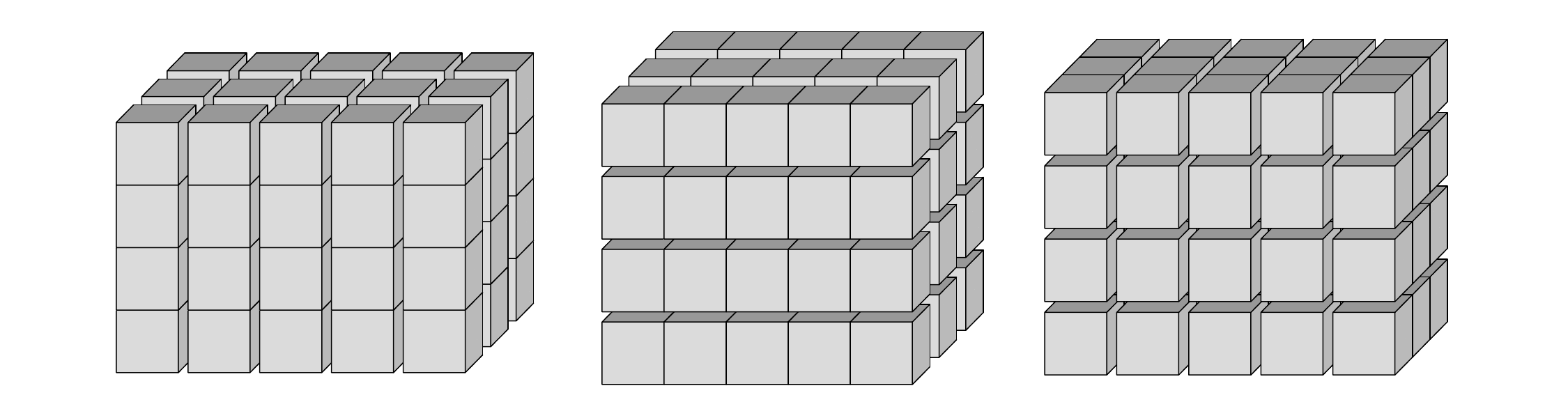}
\end{figure}

\smallskip

Recall that a matrix $\mathbf{Y}\in\mathbb{R}^{N_1\times N_2}$ can be vectorized by stacking all its columns to obtain an $N_1N_2\times 1$ vector. There is also a second way to vectorize a matrix by stacking all its rows into an $1\times N_1N_2$ vector. Similarly, a $d$-way tensor $\mathbf{Y}\in\mathbb{R}^{N_1\times\dots\times N_d}$ can be matricized in $d$ different ways across each of its $d$ modes.\footnote{Sometimes the term ``flattening'' of a tensor is used. We prefer the term ``matricization'' as it makes clear we are creating matrices, i.e., tensors could be flattened to lower dimensions that are not necessarily matrices.} The mode-$j$ matricizations of a tensor are obtained by stacking its mode-$j$ fibers as columns of  a $N_j\times\prod_{l\ne j}N_l$  matrix, denoted $\mathbf{Y}_{(j)}$. Formally, the elements of a mode-$j$ matricization of a tensor $\mathbf{Y}\in\mathbb{R}^{N_1\times\dots\times N_d}$ is obtained by the following mapping:
\begin{equation}
	\label{eq:matrizationmapping}
	y_{i_1,i_2,\dots, i_d} \mapsto \dot y_{i_j,k} \quad \text{with} \quad k=1+\sum_{\substack{n=1\\ n\neq j}}^d \left( (i_n-1)\prod_{\substack{m=1\\m\neq j}}^{n-1} N_m \right),
\end{equation}
where $\dot y_{i_j,k}$ is $(i_j,k)$ element of $\mathbf{Y}_{(j)}\in\mathbb{R}^{N_j\times\prod_{l\ne j}N_l}$; see Appendix~\ref{appsec:matricization} for a numerical example of how a $3$-way tensor is matricized along each of its three ways.

\smallskip

As we have already mentioned there are several different definitions of tensor rank, each leading to a different $d$-way factor model. The two most widely used definitions of tensor rank go back at least to \cite{hitchcock1927expression} and lead to the so-called Canonical Polyadic (CP) and Tucker decompositions.\footnote{The CP decomposition is also known as CANDECOMP or PARAFAC; see \cite{carroll1970analysis}, \cite{harshman1970foundations}. The name Tucker decomposition comes from \cite{tucker1966some}.} The Tucker decomposition is more general than the CP decomposition and is based on a notion of a multilinear rank which corresponds to the $d$-tuple of ranks of each of its mode-$j$ matricizations:
\begin{definition}
	A tensor $\mathbf{X}\in\mathbb{R}^{N_1\times\dots\times N_d}$ has a multilinear rank $(R_1,\dots,R_d)$ if 
	\begin{equation*}
		\mathrm{rank}(\mathbf{X}_{(j)}) = R_j,\qquad 1\leq j\leq d.
	\end{equation*}
\end{definition}
Note that the matricizations of a tensor can have in general different ranks which is allowed in the Tucker decomposition and is not allowed in the CP decomposition. We can decompose a tensor $\mathbf{X}\in\mathbb{R}^{N_1\times\dots\times N_d}$ with multilinear rank $(R_1,\dots,R_d)$ as
\begin{equation*}
	 \mathbf{X}_{i_1,\dots,i_d} = \sum_{r_1=1}^{R_1} \sum_{r_2=1}^{R_2} \cdots \sum_{r_d=1}^{R_d} g_{r_1,r_2,\cdots ,r_d}\lambda_{i_1,r_1}^{(1)}\lambda_{i_2,r_2}^{(2)}\dots\lambda_{i_d,r_d}^{(d)}
\end{equation*}
or more concisely as
\begin{equation*}
	\begin{aligned}
		\mathbf{X} & = \mathbf{G}\times_1\Lambda_1\times_2\Lambda_2\dots\times_d\Lambda_d \\
		& =  \mathbf{G}\bigtimes_{j=1}^d\Lambda_j,
	\end{aligned}
\end{equation*}
where $\Lambda_j\in\mathbb{R}^{N_j\times R_j}$ is a matrix of factors/loadings with entries $\lambda_{i_j,r_j}^{(j)}$; $\mathbf{G}\in\mathbb{R}^{R_1\times\dots\times R_d}$ is the so-called \textit{core tensor} with elements $g_{r_1,\dots,r_d}$; and $\times_j$ is the mode-$j$ product. The mode-$j$ product is defined as a multiplication along the mode-$j$, e.g., the mode-$1$ product is $\times_1:\mathbb{R}^{R_1\times R_2\dots\times R_d}\times \mathbb{R}^{N_1\times R_1}\to \mathbb{R}^{N_1\times R_2\dots\times R_d}$ is 
\begin{equation*}
	\mathbf{G}\times_1\Lambda_1 = \left(\sum_{r=1}^{R_1}g_{r,r_2,\dots,r_d}\lambda^{(1)}_{i_1,r} \right)_{1\leq i_1\leq N_1, 1\leq r_j\leq R_j,2\leq j\leq d}
\end{equation*}

Following the idea of defining a tensor factor model as a sum of a low-rank tensor and a tensor of idiosyncratic shocks, we define the Tucker tensor factor model as
\begin{equation*}\label{eq:tucker}
	\begin{aligned}
		\mathbf{Y} & = \mathbf{X} + \mathbf{U},\qquad \E\mathbf{U}=0 \\
		& =  \mathbf{G}\bigtimes_{j=1}^d\Lambda_j + \mathbf{U}.
	\end{aligned}	
\end{equation*}
 It is known that the Tucker decomposition is in general not unique, however, it is possible to identify the core tensor and all factors/loadings in the Tucker tensor model under additional orthogonality restrictions; see Section~\ref{sec:id}. The orthogonality restrictions are also commonly used to identify 2-way factor models.

\begin{remark}
	In economics and finance, we often have $3$-dimensional panels, where one of the dimensions corresponds to time. Setting $F:=\Lambda_3$ and $T:=N_3$, the model in equation~(\ref{eq:tucker}) becomes
	\begin{equation*}
		\mathbf{Y} = \mathbf{G}\times_1 \Lambda_1\times_2\Lambda_2\times_3F + \mathbf{U},\qquad \E\mathbf{U}=0.
	\end{equation*}
	If $f_{t,r_3}$ is the $(t,r_3)$ element of the factor matrix $F\in\mathbb{R}^{T\times R_d}$, then the entries of $\mathbf{Y}$ can be written equivalently as
	\begin{equation*}
		\begin{aligned}
			y_{i_1,i_{2},t} & = \sum_{r_1=1}^{R_1} \sum_{r_2=1}^{R_2}\sum_{r_3=1}^{R_3} g_{r_1,r_2,r_3}\lambda_{i_1,r_1}^{(1)}\lambda_{i_2,r_2}^{(2)}f_{t,r_3} + u_{i_1,i_2,t} \\
			& = \sum_{r_1=1}^{R_1}  \lambda_{i_1,r_1}^{(1)} f_{i_2,t,r_1}^{(1)} + u_{i_1,i_2,t}  = \sum_{r_2=1}^{R_2}\lambda_{i_2,r_2}^{(2)} f_{i_1,t,r_2}^{(2)} + u_{i_1,i_2,t}  = \sum_{r_3=1}^{R_3} \lambda_{i_1,i_2,r_3}^{(3)}f_{t,r_3} + u_{i_1,i_2,t},
		\end{aligned}	
	\end{equation*}
	where $f^{(1)}_{i_2,t,r_1},f^{(2)}_{i_1,t,r_2},\lambda^{(3)}_{i_1,i_2,r_3}$ are suitably defined. The last expression suggests that the Tucker tensor factor model has $R_3$ underlying time series factors with heterogeneous exposures $\lambda_{i_1,i_2,r_3}^{(3)}$ for modes-1 and 2. However, we could also rewrite it as a factor model with $R_1$ and $R_2$ factors for modes-1 and 2 respectively. The model is consistent with \cite{lettau2022estimating,lettau2023high} who considers the exact factor model with $\mathbf{U}=0$. In contrast, in statistics it is common to consider a model, where the dynamics is driven by the core tensor $\mathbf{G}_t\in\mathbb{R}^{R_1\times R_2}$ and there are $R_1\times R_2$ time series factors; see \cite{han2020tensor}, \cite{chen2022factor}, and \cite{barigozzi2022statistical} among others.
\end{remark}

\begin{remark}
	When $R_1=R_2=\dots=R_d$ and the core tensor $\mathbf{G}$ is diagonal with elements $g_{r_1,\dots,r_d}=\one_{r_1=\dots =r_d}$, we obtain the CP tensor factor model:
	\begin{equation*}
		\mathbf{Y} = \sum_{r=1}^R\lambda_{r}^{(1)}\circ\dots\circ\lambda_{r}^{(d)} + \mathbf{U},\qquad \E\mathbf{U}=0,
	\end{equation*}
	where $\lambda_{r}^{(j)}\in\mathbb{R}^{N_j}$ are some vectors and $\circ$ denotes the tensor outer product. The CP factor model corresponds to the notion of CP rank. Formally, we say that a $d$-way tensor is a rank-1 tensor if it can be expressed as an outer product of $d$ vectors. Every tensor $\mathbf{X}\in\mathbb{R}^{N_1\times\dots\times N_d}$ can be expressed as a finite sum of rank-1 tensors and the smallest number $R$ of such rank-1 tensors is called the CP rank of $\mathbf{X}$. The CP model requires that all mode-$j$ matricizations $\mathbf{Y}_{(j)},j\leq d$ have the same rank $R$ which may be restrictive in some applications.\footnote{A previous version of our paper, \cite{babii2022tensor}, considered the CP factor model with orthogonal loadings and factors which is now a special case of a more general framework.}
\end{remark}

In the remaining part of this section, we consider several examples of tensor data in economics and finance, showing that such type of data appear in many applications. Using 3-dimensional examples with the notation 	$\mathbf{Y}$ = $y_{i,j,k}$ or 	$y_{i,j,t},$ they include:
\begin{itemize}
	\item Input-Output Models: \(i\) is industry sector, \(j\) type of input (e.g., labor, materials, capital), and \(k\) region or country
	\item Macroeconomic Data Across Countries \(i\) country, \(j\) macroeconomic variables (e.g., GDP, inflation, unemployment rate) and \(t\) time period (e.g., quarterly, annually)
     \item High-Frequency Trading Data: \(i\) asset (e.g., stocks, commodities), \(j\) attributes (e.g., bid price, ask price, volume), and \(t\) timestamps (e.g., milliseconds)
	\item Energy Markets and Economics: \(i\) country, \(j\) energy types (e.g., coal, wind, solar, oil), and \(t\) time periods
	\item Supply Chain Analysis: \(i\) products, \(j\) locations (e.g., factories, warehouses, retail outlets), and \(t\) time periods
	\item Real Estate and Urban Economics: \(i\) region or city, \(j\) property types (e.g., residential, commercial), and \(t\) time periods.
\end{itemize}
As well as 4-dimensional examples such as:
\begin{itemize}
	\item Banking and Credit Risk Models: \(i\) customer, \(j\) loan types, \(k\) credit scores, and \(t\) time periods
	\item Consumer Behavior Analysis: \(i\) consumer demographics (e.g., age group, income level), \(j\) product categories (e.g., electronics, groceries), \(k\) regions, and \(t\) time periods.
\end{itemize}

\section{Tensor PCA \label{sec:PCAtensor}}
In this section, we consider two tensor PCA (TPCA) algorithms to estimate the Tucker tensor factor model in equation~(\ref{eq:tucker}). We begin by discussing the sufficient identifying conditions and present a simple TPCA algorithm. We argue that the simple TPCA is optimal when factors are strong in the second subsection. The third subsection describes an improved iterative alternating least-squares algorithm and discusses the improvements for the weak factor model. The next subsection provides the large sample distributions for loadings/factors estimated with simple TPCA.

\subsection{Identification and Simple Tensor PCA}\label{sec:id}
It is known that the Tucker decomposition of a tensor is not in general unique. In this section, we argue that the loadings/factors in the Tucker tensor factor model in equation~(\ref{eq:tucker}) can be identified under the following assumption:
\begin{assumption}\label{as:orthogonal}
	For every $1\leq j\leq d$, (i) $\Lambda_j^\top \Lambda_j = I_{R_j}$; and (ii) $\mathbf{G}_{(j)}\mathbf{G}_{(j)}^\top=\mathrm{diag}(\sigma^2_{j,1},\dots,\sigma^2_{j,R_j})$ for some $\sigma_{j,1}>\dots>\sigma_{j,R_j}>0$.
\end{assumption}

By Appendix Lemma~\ref{lemma:tensor_matricization}, the mode-$j$ matricization of equation~(\ref{eq:tucker}) is
\begin{equation}\label{eq:unfolded}
	\mathbf{Y}_{(j)} = \Lambda_j\mathbf{G}_{(j)}\left(\bigotimes_{l\ne j}\Lambda_l\right)^\top + \mathbf{U}_{(j)},\qquad 1\leq j\leq d,
\end{equation}
where $\bigotimes_{l\ne j}\Lambda_l=\Lambda_d\otimes\dots\otimes \Lambda_{j+1}\otimes \Lambda_{j-1}\otimes\dots\otimes \Lambda_{1}$.  The latter satisfies the following property:
\begin{proposition}\label{prop:id}
	Under Assumption~\ref{as:orthogonal} (i)
	\begin{equation*}
		\left(\bigotimes_{k\ne j}\Lambda_k\right)^\top\left(\bigotimes_{k\ne j}\Lambda_k\right) = I_{\prod_{k\ne j}R_k},\qquad 1\leq j\leq d.
	\end{equation*}
\end{proposition}

\noindent Then the tensor factor model is identified from the PCA applied to $\mathbf{Y}_{(j)}$. Indeed, if $\mathbf{U}=0$, then by Proposition~\ref{prop:id}
\begin{equation*}
	\mathbf{Y}_{(j)}\mathbf{Y}_{(j)}^\top = \Lambda_jD_j\Lambda_j^\top,
\end{equation*}
where under Assumption~\ref{as:orthogonal}, (ii) $D_j=\mathbf{G}_{(j)}\mathbf{G}_{(j)}^\top$ is a diagonal matrix with distinct elements and $\Lambda_j^\top\Lambda_j=I_{R_j}$. Then $\Lambda_j$ is identified as eigenvectors of $\mathbf{Y}_{(j)}\mathbf{Y}_{(j)}^\top$ or equivalently as the left singular vectors of $\mathbf{Y}_{(j)}$. Moreover, we have
\begin{equation}\label{eq:projection_core}
	\mathbf{Y}\bigtimes_{j=1}^d\Lambda_j = \mathbf{G}
\end{equation}
since $\Lambda_j^\top \Lambda_j=I_{R_j}$ for all $1\leq j\leq d$. If $\mathbf{U}\ne 0$, then assuming that the factor part dominates the noise asymptotically, cf. Assumption~\ref{as:strong_factors} below, identification is achieved in large samples.

\smallskip

This leads us to the TPCA estimation Algorithm~\ref{alg:hosvd} often credited to \cite{tucker1966some} with some aspects already in \cite{hitchcock1927expression}.\footnote{It is also known as higher-order singular values decomposition in the numerical analysis; see \cite{de2000multilinear}.}

\begin{algorithm}
	\caption{TPCA}
	\label{alg:hosvd}
	\begin{algorithmic}[1] % The number tells where the line numbering should start
		\Procedure{TPCA}{$\mathbf{Y},R_1,\dots,R_d$}
		\For{$j=1,\dots,d$}
		\State $\hat \Lambda_j\gets N_j\times R_j$ matrix of $R_j$ leading left singular vectors of $\mathbf{Y}_{(j)}$ 
		\EndFor
		\State $\hat {\mathbf{G}}\gets \mathbf{Y}\bigtimes_{j=1}^d\hat \Lambda_j^\top$ \Comment{Estimate the core tensor}
		\State \textbf{return} $\hat {\mathbf{G}},\hat \Lambda_1,\hat \Lambda_2,\dots,\hat \Lambda_d$ 
		\EndProcedure
	\end{algorithmic}
\end{algorithm}

\smallskip

\subsection{Rates of Consistency}
The following assumption imposes mild restrictions on the data generating process.
\begin{assumption}\label{as:data}
	The idiosyncratic errors $\mathbf{U}=\{u_{i_1,\dots,i_d}:\; 1\leq i_j\leq N_j,1\leq j\leq d\}$ are i.i.d.\ such that $\E(u_{i_1,\dots,i_d})=0$, $\Var(u_{i_1,\dots,i_d})=\sigma^2$, and for some $c>0$, $\E[e^{\lambda u_{i_1,\dots,i_d}}] \leq e^{c\lambda},\forall\lambda\in\mathbb{R}$.
\end{assumption}
Assumption~\ref{as:data} does not impose any restrictions on the dependence structure in factors and loadings. The factors may be generated by a non-stationary process and the loadings may fail to be i.i.d. It is plausible that after controlling for cross-sectional and time-series dependence through the factor structure, the idiosyncratic errors are i.i.d. Nonetheless, Assumption~\ref{as:data} can also be relaxed to heterogeneous and dependent arrays at costs of heavier notation and proofs and modified TPCA algorithm, see \cite{zhang2022heteroskedastic}, which is left for future work.

\smallskip

Next, we make an assumption about the factor strength. Let $\sigma_{j,R_j}$ be the smallest non-zero singular value of the matricized tensor $\mathbf{X}_{(j)}$ and let $\delta=\min_{1\leq j\leq d}\sigma_{j,R_j}$ be a measure of factor strength in the tensor factor model. We assume next that factors are not extremely weak in the sense that the smallest singular values of matricizations are well-separated from zero: 
\begin{assumption}\label{as:strong_factors}
	The factor strength is such that $\delta^2 \gtrsim \prod_{j=1}^d\sqrt{N_j} + \max_{1\leq j\leq d}N_j$.
\end{assumption}
The strong factors model in the tensor setting can be described as $\delta^2\sim \prod_{j=1}^dN_j$.\footnote{In the 2-way case, this corresponds to assuming that the squared singular values of $\mathbf{Y}$ scale at the $NT$-rate; see also \cite{onatski2012asymptotics,onatski2022uniform} for the weak factors in the 2-way case.} In this case, Assumption~\ref{as:strong_factors} is automatically satisfied. Let $O_j$ be the optimal rotation matrix solving $\min_{O\in\mathbb{O}_{R_j}}\|\hat \Lambda_jO - \Lambda_j\|_{\rm F}^2$.\footnote{It is known that the optimal rotation matrix for the Frobenius norm has a closed-form expression $O_j =\mathrm{sign}(\hat\Lambda_j^\top\Lambda_j)$, where the sign of a matrix $A$ with SVD $A=\Lambda DV^\top$ is defined as $\mathrm{sign}(A)=\Lambda V^\top$.} The following result holds:

\begin{theorem}\label{thm:rates}
	Suppose that Assumptions~\ref{as:orthogonal}, \ref{as:data}, and \ref{as:strong_factors} are satisfied. Then with probability at least $1-Ce^{-cN_j}$, we have
	\begin{equation*}
		\|\hat \Lambda_jO_j - \Lambda_j\|_{\rm F}^2 \lesssim \frac{R_jN_j}{\delta^2} + \frac{R_j\prod_{l=1}^dN_l}{\delta^4} ,\qquad 1\leq \forall j\leq d.
	\end{equation*}
\end{theorem}
The proof appears in the Appendix. Theorem~\ref{thm:rates} is a non-asymptotic result valid for any values of tensor dimensions $(N_1,\dots,N_d)$, provided that Assumption~\ref{as:strong_factors} is satisfied. It also applies to the CP-factor model, when matricizations have the same ranks and factors/loadings are assumed to be orthogonal. The orthogonality restriction for the CP factor model can also be relaxed; see \cite{chang2023modelling} and \cite{chen2024estimation} for recent contributions.

\begin{remark}
	Consider the 2-way case, $\mathbf{Y}=\Lambda F^\top + \mathbf{U}$ with $\mathbf{Y}\in\mathbb{R}^{N\times T}$, $\Lambda\in\mathbb{R}^{N\times R}$, and $F\in\mathbb{R}^{T\times R}$, where $R$ is the number of factors. In the strong factor model, the smallest singular value of $\Lambda F^\top$ is $\delta\sim \sqrt{NT}$. Theorem~\ref{thm:rates} implies that
	\begin{equation*}
		\|\hat \Lambda O_1 - \Lambda\|_{\rm F}^2 = O_P\left(\frac{R}{T}\right)\qquad  \text{and}\qquad \|\hat FO_2 - F\|_{\rm F}^2 = O_P\left(\frac{R}{N}\right),
	\end{equation*}
	To the best of our knowledge, these rates are faster than the best currently known results in the factor literature, cf. \cite{bai2023approximate}.  In particular, the factor loadings are consistently estimated when $N$ is fixed while factors are consistently estimated when $T$ is fixed. We also allow for the number of factors to diverge slowly with the sample size; see also \cite{freeman2023linear} and \cite{beyhum2022factor} for results with a diverging number of factors in the 2-way case.
\end{remark}

\begin{remark}
	For the strong factor model, when $\delta^2\sim\prod_{j=1}^dN_j$, Theorem~\ref{thm:rates} shows that
	\begin{equation*}
		\|\hat \Lambda_jO_j - \Lambda_j\|_{\rm F}^2 = O_P\left(\frac{R_j}{\prod_{l\ne j}N_l}\right),\qquad 1\leq \forall j\leq d.
	\end{equation*}
	More generally, the $O_P(R_jN_j/\delta^2)$ rate is obtained provided that the factor/loadings strength is such that $\delta^2\gtrsim\prod_{l\ne j}N_l$. This rate is known to be minimax-optimal when $d=3$ see \cite{zhang2018tensor}, Theorem 3. Therefore, the naive TPCA Algorithm~\ref{alg:hosvd} is optimal under the strong factors asymptotics.
\end{remark}

\begin{remark}
	If the last tensor dimension corresponds to time, i.e., $N_d=T$ and $\Lambda_d=F$, then for the strong factor model, Theorem~\ref{thm:rates} shows that factors are estimated at the rate
	\begin{equation*}
		\|\hat FO_d - F\|_{\rm F}^2 = O_P\left( \frac{R_j}{N_1N_2\dots N_{d-1}}\right).
	\end{equation*}
\end{remark}

\subsection{Alternating Least-Squares}
For weak factors/loadings, the estimation accuracy of factors/loadings in Theorem~\ref{thm:rates} does not always improve when more data is available in all tensor dimensions because the second term with $\prod_{l=1}^dN_l$ can dominate, creating a bottleneck. This comes from the fact that the matricized tensor $\mathbf{Y}_{(j)}$ has a very large number of columns, namely $\prod_{l\ne j}N_l$. Under Assumption~\ref{as:orthogonal}, using the mode-$l$ multiplication of $\mathbf{Y}$ by loadings $\Lambda_l,l\ne j$, we obtain
\begin{equation}\label{eq:project}
	\mathbf{Y}\bigtimes_{l\ne j}\Lambda_l^\top = \mathbf{G}\times_j\Lambda_j^\top + \mathbf{U}\bigtimes_{l\ne j}\Lambda_l^\top;
\end{equation}
cf. equation~(\ref{eq:projection_core}). The mode-$j$ matricization of equation~(\ref{eq:project}),
\begin{equation*}
	\left(\mathbf{Y}\bigtimes_{l\ne j}\Lambda_l^\top\right)_{(j)} = \Lambda_j\mathbf{G}_{(j)} + \left(\mathbf{U}\bigtimes_{l\ne j}\Lambda_l^\top\right)_{(j)},
\end{equation*}
has a much smaller number of columns, namely $\prod_{l\ne j}R_l$ instead of $\prod_{l\ne j}N_l$. Therefore, applying PCA to the mode-$j$ matricization of the left-hand side in the equation~(\ref{eq:project}), may allow us to estimate $\Lambda_j$ more precisely, provided that the noise is negligible; see also equation~(\ref{eq:projection}). Iterating this procedure and updating the factors/loadings leads to the alternating least-squares Algorithm~\ref{alg:hooi} for Tucker decomposition; see \cite{kroonenberg1980principal}.\footnote{It is also known as the higher-order orthogonal iterations; see \cite{de2000best} for more details.}
The following result establishes the convergence rate of estimators computed using this algorithm:
\begin{theorem}\label{thm:als}
	Suppose that assumptions of Theorem~\ref{thm:rates} are satisfied and for all $j\leq d$, we have $\sigma_{1,j}\lesssim \delta$, $N_j\sim N\uparrow\infty$, and $R_j=O(1)$. Then if the factor strength is such that $N^d/\delta^4=o(1)$ and the number of iterations is $\bar k\geq \log(N^{(d-1)/2}/\delta)$, we obtain with probability at least $1-Ce^{-cN}$ that:
	\begin{equation*}
		\left\|\hat \Lambda_j^{(\bar k)}O_j - \Lambda_j\right\|_{\rm F}^2 \lesssim \frac{N}{\delta^2},\qquad 1\leq \forall j\leq d.
	\end{equation*}
\end{theorem}
\noindent The proof of this result appears in the Appendix.
\begin{algorithm}
	\caption{ALS}
	\label{alg:hooi}
	\begin{algorithmic}[1] % The number tells where the line numbering should start
		\Procedure{ALS}{$\mathbf{Y},R_1,\dots,R_d$}
		\State Initialize $\hat \Lambda_j^{(0)}$  for $1\leq j\leq d$  using Algorithm~\ref{alg:hosvd}.
		\For{$k=1,2,\dots,\bar k$}
		\For{$j=1,2,\dots,d$}
		\State $\hat{\mathbf{Y}}^{(k)} \gets \mathbf{Y}\bigtimes_{l\ne j}\hat \Lambda_l^{(k-1)\top}$
		\State $\hat \Lambda_j^{(k)}\gets R_j$ leading left singular vectors of $\hat{\mathbf{Y}}^{(k)}_{(j)}$
		\Comment{Update loadings/factors}
		\EndFor
		\EndFor
		\State $\hat {\mathbf{G}}\gets \mathbf{Y}\bigtimes_{j=1}^d\hat \Lambda_j^{(\bar k)\top}$ \Comment{Estimate the core tensor}
		\State \textbf{return} $\hat {\mathbf{G}},\hat \Lambda_1^{(\bar k)},\hat \Lambda_2^{(\bar k)},\dots,\hat \Lambda_d^{(\bar k)}$ 
		\EndProcedure
	\end{algorithmic}
\end{algorithm}

Theorem~\ref{thm:als} makes a simplifying assumption that the tensor dimensions increase proportionally to each other which in the 2-way factor case requires that $N/T\to c>0$. Then for the strong factor model, we have $\delta^2\sim N^d$ and the condition $N^d/\delta^4=o(1)$ simply requires that $N\to\infty$ which is trivially satisfied. Consequently, we obtain the rate of order $O_P(1/N^{d-1})$ which is the same as the optimal rate of the naive TPCA Algorithm~\ref{alg:hosvd}. However, when the factors are weak, the rate of the ALS Algorithm~\ref{alg:hooi} can be faster than that of the naive TPCA. Therefore, we can expect that the ALS with naive TPCA used as a starting point may estimate factors/loadings more accurately, provided that a sufficiently large number of iterations is made.

\begin{remark}
	It is known that the rate achieved by the ALS Algorithm~\ref{alg:hooi} in Theorem~\ref{thm:als} is minimax-optimal; see \cite{zhang2018tensor}, Theorem 3. Interestingly, if the factors become so weak, that the $N^d/\delta^4=o(1)$ condition fails, the rate-optimal MLE estimator, is a solution to the NP-hard optimization problem and there is a gap between statistical and computational limits; see also \cite{richard2014statistical}.
\end{remark}

\subsection{Inference}
In  this section, we consider inference on loadings and factors in the general Tucker model, covering the orthogonal CP model as a special case. The loadings/factors are estimated with the TPCA Algorithm~\ref{alg:hosvd}. Let $\hat\Lambda_{j,i:}$ and $\Lambda_{j,i:}$ be $R_j\times 1$ vectors corresponding to the transposed $i^{\rm th}$ rows of $\hat\Lambda_jO_j$ and $\Lambda_j$ respectively with $O_j=\hat\Lambda_j^\top\Lambda_j$. Let $\|.\|_{2,\infty}$ be the entry-wise $\ell_{2,\infty}$ matrix norm.

\smallskip
We need the following set of assumptions for inference:
\begin{assumption}\label{as:clt}
	(i) the factor strength is $\delta\sim \prod_{j=1}^dN_j$; (ii) the coherence is $\|\Lambda_j\Lambda_j^\top\|_{2,\infty}^2=O(R_j/N_j)$ and $\|V_j\|_{2,\infty}=o(1)$ for all $j\leq d$.
\end{assumption}
The coherence condition roughly requires that the values of factors/loadings are spread out instead of being concentrated in a finite number of entries; see  \cite{candes2012exact}, Definition 1.8.

The following result holds for factors/loadings, depending on the dimension $j=1,\dots,d$:
\begin{theorem}\label{thm:clt}
	Suppose that Assumptions~\ref{as:orthogonal}, \ref{as:data}, and \ref{as:clt} are satisfied. Then
	\begin{equation*}
		D_j^{1/2}(\hat\Lambda_{j,i:} - \Lambda_{j,i:}-B_{j,i}) \xrightarrow{d} N(0,\sigma^2I_{R_j}),
	\end{equation*}
	provided that $\prod_{l\ne j}N_l/N_j^3=o(1)$ and the expression of $B_{j,i}$ can be found in the proof.
\end{theorem}
\noindent Note that for $d=3$, conditions of Theorem~\ref{thm:clt} are satisfied when the tensor dimensions grow proportionally, i.e., $N_j\sim N\uparrow\infty$ for $j=1,2,3$. Note also that the results of \cite{bai2003inferential}, Theorem 2, for 2-way factor models are not applicable to tensors because the condition $\sqrt{T}/N=o(1)$ does not hold for $d=3$ when the tensor dimensions grow proportionally since $T\sim N^2$ in this case. Lastly, eliminating the bias and/or relaxing the i.i.d. homoskedastic errors may require modifying the TPCA estimator, see \cite{zhang2022heteroskedastic}, and is left for future research.

\smallskip

\noindent We now turn to the estimation of scale components. Let $(\hat \sigma_{j,r}^2)_{1\leq r\leq R}$ be the eigenvalues of $\mathbf{Y}_{(j)}\mathbf{Y}_{(j)}^\top$. The following result holds:
\begin{theorem}\label{thm:clt_scale}
	Under Assumptions of Theorem~\ref{thm:clt}, we have
	\begin{equation*}
		\left(\frac{\hat\sigma_{j,r}^2 - \sigma_{j,r}^2}{\sigma_{j,r}} \right)_{1\leq r\leq R_j} \xrightarrow{d} N(0,4\sigma^2I_{R_j}).
	\end{equation*}
\end{theorem}
\noindent It immediately follows from Theorem~\ref{thm:clt_scale} that $\hat\sigma_{j,r}^2/\prod_{j=1}^dN_j$ is a consistent estimator of the asymptotic factor strength constant $d_{j,r}$.

\section{Testing the Number of Factors}\label{sec:test}
In this section, we develop a novel test for the number of factors in the tensor factor model. Specifically, we consider the following hypotheses:
%The test builds on the eigenvalue ratio statistics of \cite{onatski_ecma_2009} for different matricizations and the p-value combinations; see \cite{vovk2020combining}. 
\begin{equation*}
	\mathrm{H_0}:\; \text{$\leq k$ factors} \qquad \rm{vs.} \qquad H_1:\; \text{the number of factors is $>k$, but $\leq K$}.
\end{equation*}
The quantity $K$ is selected by the user and reflects the upper bound on the total possible number of factors. 

Let $\hat\sigma^2_{j,1}\geq \hat\sigma^2_{j,2}\geq \dots\geq \hat\sigma^2_{j,N_j}\geq 0$ be the eigenvalues of $\mathbf{Y}_{(j)}\mathbf{Y}_{(j)}^\top$. Under Assumption~\ref{as:strong_factors}, the first $R_j$ eigenvalues diverge from zero, at least in population. Consider the following statistics
\begin{equation*}
	S_j = \max_{k<r\leq K}\frac{\hat\sigma^2_{j,r} - \hat\sigma^2_{j,r+1}}{\hat\sigma^2_{j,r+1} - \hat\sigma^2_{j,r+2}},\qquad 1\leq j\leq d,
\end{equation*}
and put
\begin{equation*}
	Z = \max_{0<r\leq K-k}\frac{\xi_r - \xi_{r+1}}{\xi_{r+1} - \xi_{r+2}},
\end{equation*}
where $(\xi_1,\dots,\xi_{K-k+2})$ follow the joint type-1 Tracy-Widom distribution; see \cite{karoui2003largest} and \cite{soshnikov2002note}.

Then, consider the sequence:
{\footnotesize
\begin{equation*}
	\tau = \left(\sqrt{N_j\vee \prod_{k\ne j}N_k -1} + \sqrt{N_j\wedge \prod_{k\ne j}N_k}\right)\left(\left(N_j\vee \prod_{k\ne j}N_k -1 \right)^{-1/2} + \left(N_j\wedge \prod_{k\ne j}N_k \right)^{-1/2}\right)^{1/3}.
\end{equation*}}

\noindent The following result holds provided that $N_j\lesssim \prod_{k\ne j}N_k$:
\begin{theorem}\label{thm:test}
		Suppose that Assumptions~\ref{as:orthogonal} and \ref{as:data} are satisfied, $\sigma_{j,r}^2\sim \prod_{j=1}^dN_j$, and $u_{i_1,\dots,i_d}\sim N(0,\sigma^2)$. Suppose also that $N_j/\tau+\prod_{k\ne j}N_k/(N_j\tau)=o(1)$. Then under $H_0$, $	S_j \xrightarrow{d} Z$, while under $H_1$, we have $S_j\uparrow\infty$ for every $j\leq d$.
\end{theorem}
Note that the rate condition for $\tau$ is satisfied in the 3-dimensional case when $N_1\sim N_2\sim N_3$.  Theorem~\ref{thm:test} leads to the following testing procedure:
\begin{enumerate}
	\item Let $(Z_i)_{i=1}^m$ be $m$ independent random variables drawn from the same distribution as $Z$. To approximate the distribution of $(\xi_1,\xi_2,\dots)$, we use the eigenvalues of a symmetric $N_j\times N_j$ Gaussian matrix $\Xi=(\zeta_{i,j})$ with $\zeta_{i,j}\sim_{i.i.d.}N(0,\tau_{i,j})$ with $\tau_{i,j}=1$ if $i<j$ and $\tau_{i,j}=2$ for $i=j$.
	\item Compute the p-value $p_j = 1-F_m(S_j)$ for each $1\leq j\leq d$, where $F_m(x)=\frac{1}{m}\sum_{i=1}^m\one_{Z_i\leq z}$.
\end{enumerate}
The p-values would correspond to the null hypothesis that there are at most $k$ factors for the mode-$j$ matricization. The joint tests for the number of factors across all matricizations can be obtained using the Bonferroni's correction. The $\alpha$-level test would reject $H_0$ whenever
\begin{equation*}
	p_j \leq \alpha/d. 
\end{equation*}
In practice, one could run the test for several different values of $k$ to determine the number of factors and control the size with Bonferroni correction. Note also that the test statistic $S_j$ is decreasing with $k$, which implies that p-values are increasing with $k$.

\begin{remark}
	In contrast to \cite{onatski_ecma_2009}, the dimensions of matrices obtained from tensor matricization do not grow proportionally and the Tracy-Widom asymptotics is recovered thanks to \cite{karoui2003largest}. Note also that in our case we have the type-1 Tracy-Widom distribution. 
\end{remark}

\section{Monte Carlo Experiments
\label{sec:sim}}

The objective of this section is to assess the finite sample properties of our estimation procedure.

\subsection{Simulation Design}\label{sec:dgp}
We consider the Tucker model for the 3-way tensor for $N\times J\times T$ tensor
\begin{equation*}
	\mathbf{Y} = \mathbf{G}\times_1\Lambda\times_2 M\times_3 F + \mathbf{U},
\end{equation*}
where the core tensor is $\mathbf{G}\in\mathbb{R}^{R_1\times R_2\times R_3}$, the loadings matrices are $\Lambda\in\mathbb{R}^{N\times R_1}$ and $M\in\mathbb{R}^{J\times R_2}$, and the factor matrix is $F\in\mathbb{R}^{T\times R_3}$. 

\medskip

We set $R_1=1$ and $R_2=R_3=2$ and generate the $1\times 2\times 2$ core tensor $\mathbf{G}$ with $1\times 2$ slices
\begin{equation*}
	G_1 = [\sigma_{1}, 0]\qquad \text{and}\qquad G_2=[0, \sigma_{2}],
\end{equation*}
where $\sigma_1>\sigma_2>0$. The corresponding matricizations are
\begin{equation*}
	\mathbf{G}_{(1)} = [\sigma_{1},0,0,\sigma_{2}],\qquad \mathbf{G}_{(2)}=\begin{bmatrix}
		\sigma_{1} & 0 \\
		0 & \sigma_{2}
	\end{bmatrix},\qquad \mathbf{G}_{(3)}=\begin{bmatrix}
	\sigma_{1} & 0 \\
	0 & \sigma_{2}
\end{bmatrix}.
\end{equation*}
Clearly, we have $\mathbf{G}_{(j)}\mathbf{G}_{(j)}^\top$ for $j=1,2,3$ begin diagonal, hence, the core tensor satisfies Assumption~\ref{as:orthogonal}. The smallest singular values of $\mathbf{G}_{(j)}$ are $\sqrt{\sigma_1^2+\sigma_2^2}$, $\sigma_2$, and $\sigma_2$ respectively for $j=1,2,3$. The factor strength is $\delta=\sigma_2$. We first set $\sigma_1=d_1\sqrt{NJT}$ and $\sigma_2=d_2\sqrt{NJT}$ which corresponds to the strong tensor factor model.

\medskip

Therefore, we simulate the $\mathbf{Y}$ tensor with $(i,j,t)^{\rm th}$ observation
\begin{equation*}
		y_{i,j,t} = \sigma_{1}\lambda_{i}\mu_{j,1}f_{t,1} + \sigma_{2}\lambda_{i}\mu_{j,2}f_{t,2} + u_{i,j,t},\qquad u_{i,j,t}\sim_{\rm i.i.d.}N(0,s^2_u).
\end{equation*}
Note that this rank-$(1,2,2)$ Tucker tensor factor model features fewer parameters than the rank-2 CP model, where the first loading vector would be different in the two terms.

\medskip

The rest of the DGP is as follows. We generate $T\times 2$ matrix of factors from two independent AR(1) processes
\begin{equation*}
	\dot{f}_{t,r} = \rho \dot{f}_{t-1,r} + \varepsilon_{t,r},\qquad \varepsilon_{t,r}\sim_{\rm i.i.d.} N(0,s_{\varepsilon}^2), \qquad r=1,2.
\end{equation*}
For the factor series $\dot f_r\in\mathbb{R}^T$, we compute $f_r = \dot f_r/\|\dot f_r\|$ to ensure that the factor matrix satisfies Assumption~\ref{as:orthogonal}. All loading vectors are generated by taking norm-1 eigenvectors of a symmetric positive semi-definite matrix $A^\top A,$ where each element is $A_{i,j}\sim_{\rm i.i.d.} U(0,1).$

%{\color{red} We also need to add simulation results for the central limit Theorem~\ref{thm:clt}.}
%
%We focus on inference for the $i^{\rm th}$ entry of the second loading vector $\mu_{j,1}\in\mathbb{R}$. If we take $\mu_{j,r}=1/\sqrt{J},r=1,2$, then 
%\begin{equation*}
%	d_1s_u^{-1}\sqrt{NJT}(\hat\mu_{j,1} - \mu_{j,1} - b_{i,r}) \xrightarrow{d}N(0,1).
%\end{equation*}

\subsection{Small Sample Properties}\label{sec:finite}
In this subsection, we first assess how changing sample sizes affects the estimation accuracy of factors and loadings, and show that the estimation improvements are perfectly aligned with the convergence rates given in Theorem~\ref{thm:rates}. We also show that the finite sample distribution of the estimation error is aligned with the asymptotic distribution shown in Theorem \ref{thm:clt}.

\smallskip 

For convergence rate, we simulate the 3-way tensor as described in Section \ref{sec:dgp}
where the parameters are set to be: (1) the AR(1) process $\dot{f}_r$ takes $\rho$ = 0.5 and $s_\varepsilon$ = 0.1; (2) the signal strength $d_1=2, d_2=1$ and the noise strength $s_u$ = 1; (4) the sample size of the baseline model is $(N,J,T)$=(30,30,30), and we compare the estimation error of the baseline model with that of the modified model. We consider three different cases of modification: (a) doubling sample sizes of all dimensions, $(N,J,T)$ = (60,60,60), (b) doubling the sample sizes of two dimensions, $(N,J,T)$ = (60,60,30), (c) doubling the sample size of only one dimension, $(N,J,T)$ = (60,30,30). We evaluate the estimates using the $\ell_2$ norm. Although the rank for the $M,F$ is 2, the convergence rates are the same for the first and the second rank and it is enough to examine the first rank. As the signs of $\hat\lambda_r, \hat\mu_r$, and $\hat f_r$ are undetermined, we calculate the  errors as follows: 
\begin{equation}\label{eq:l2err}
	\begin{aligned}
		\mathbb{L}_\lambda &=  \|\hat\lambda_r\times \mathrm{sign}(\hat\lambda_r^\top\lambda_r) - \lambda_r\|, \\
		\mathbb{L}_\mu &= \|\hat\mu_r \times \mathrm{sign}(\hat\mu_r^\top\mu_r)- \mu_r\|,  \\
		\mathbb{L}_f &= \|\hat f_r \times \mathrm{sign}(\hat f_r^\top f_r)- f_r\|,
	\end{aligned}
\end{equation}
where $\mathrm{sign}(a)=\one_{a>0} - \one_{a<0}$.

\smallskip

Figure \ref{fig:sizes} plots the histograms of the $\ell_2$ losses of the baseline versus modified DGPs. In panel (a) - (c), as we double the sizes of all dimensions, the estimation of the factor $\hat f_1$ and loadings $\hat \lambda_1,$ $\hat\mu_1$ all improve. The average error is reduced roughly by a half for the factor and two loadings vectors. This is aligned with the convergence rate in Theorem~\ref{thm:rates}, since doubling all $3$ dimensions of a tensor reduces the $\ell_2$ error by $1/2$. In panels (d) - (f), when we only double $N$ and  $J,$ the improvement for the average error of $\hat\lambda_1$ is 0.01/0.015, while the improvement for $\hat\mu_1$ is 0.012/0.017. Both are roughly aligned with the reduction in the $\ell_2$ error by $1/\sqrt{2}$. On the other hand, the improvement for $\hat f_1$ is 0.0083/0.017,  which is aligned with the reduction of the $\ell_2$ error by $1/2$. In panels (g) - (i), when we only double $N,$ there is no improvement for $\hat\lambda_1$ because $J$ and $T$ are unchanged in the $O_P(1/\sqrt{JT})$ rate; the improvement for both $\hat\mu_1$ and $\hat f_1$ are 0.012/0.017, which is aligned with the $1/\sqrt{2}$ improvement factor.

\setcounter{subfigure}{0}
\begin{figure}
	\caption{ \label{fig:sizes} Estimation Accuracy of Tensor PCA: Changing Sizes of Dimensions} \vskip0.1in
	{\footnotesize The DGP appears in Section \ref{sec:dgp}. We plot the histograms of $\ell_2$ losses of the estimated factor/loading of baseline vs modified DGP in 5000 MC simulations. The baseline DGP has sizes $(N,J,T)$ = (30,30,30), and modified DGP has sizes $(N,J,T)$ shown in the subtitles. The blue histogram corresponds to the baseline DGP while the orange histogram to the modified DGP with the increased sample size. The dotted line marks the mean of the $\ell_2$ errors.}
	\begin{center}
		\subfigure[{\bf $ \hat\lambda $ - $ \mathbf{60\times 60\times 60} $}]{\includegraphics[width=0.32\textwidth]{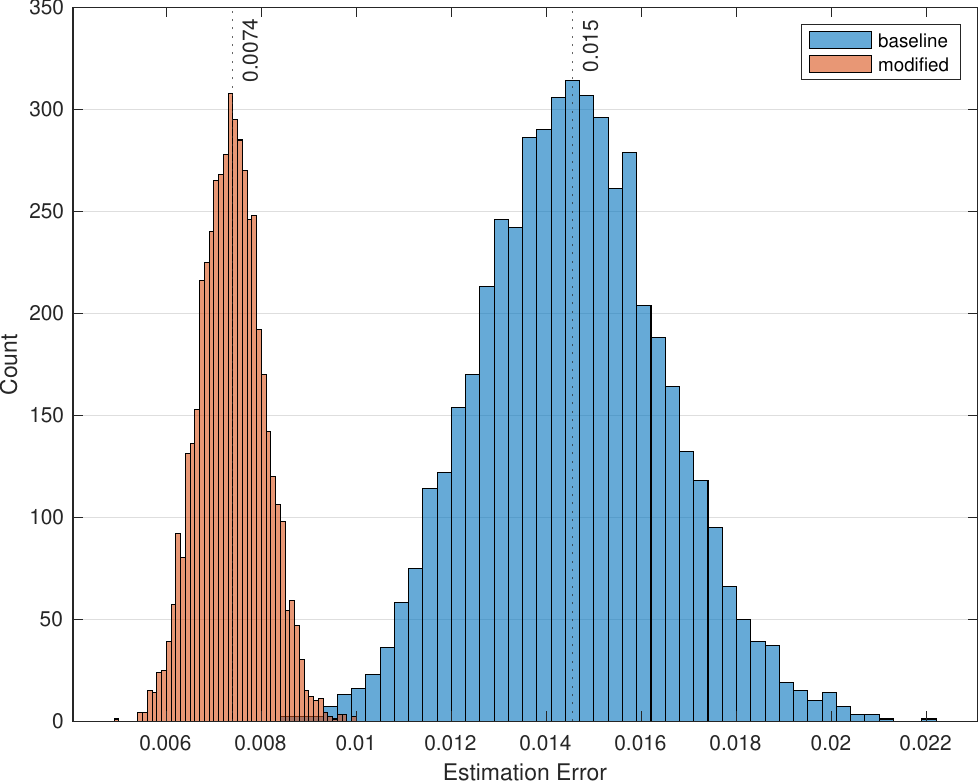}}
		\subfigure[{\bf $ \hat\mu $ - $ \mathbf{60\times 60\times 60} $}]{\includegraphics[width=0.32\textwidth]{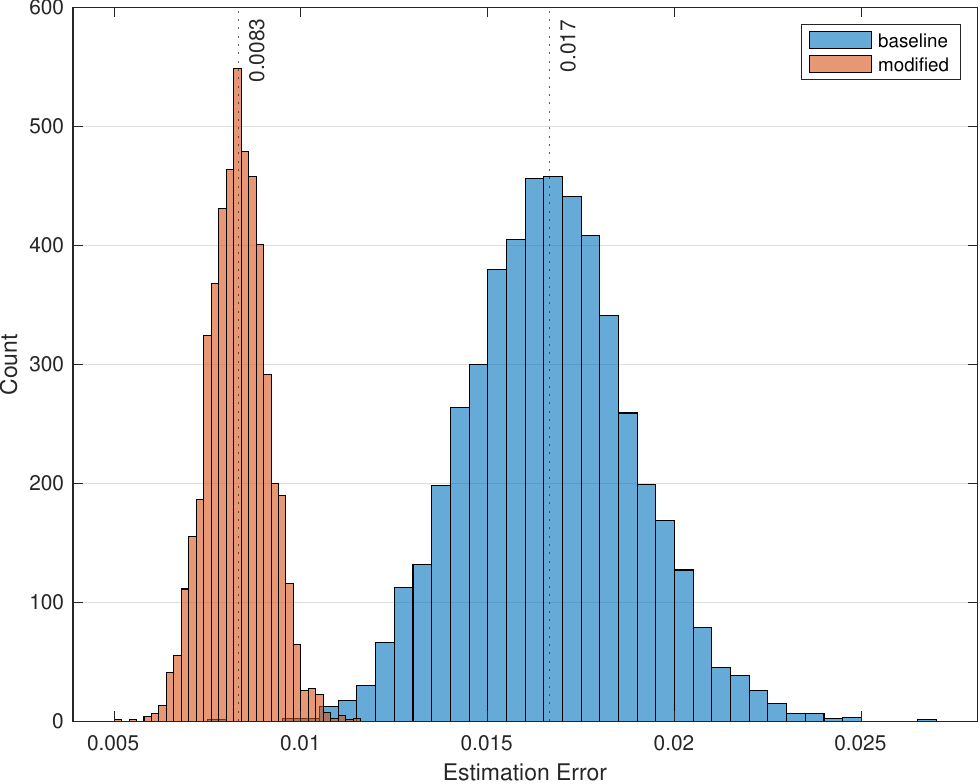}}
		\subfigure[{\bf $ \hat f $ - $ \mathbf{60\times 60\times 60} $}]{\includegraphics[width=0.32\textwidth]{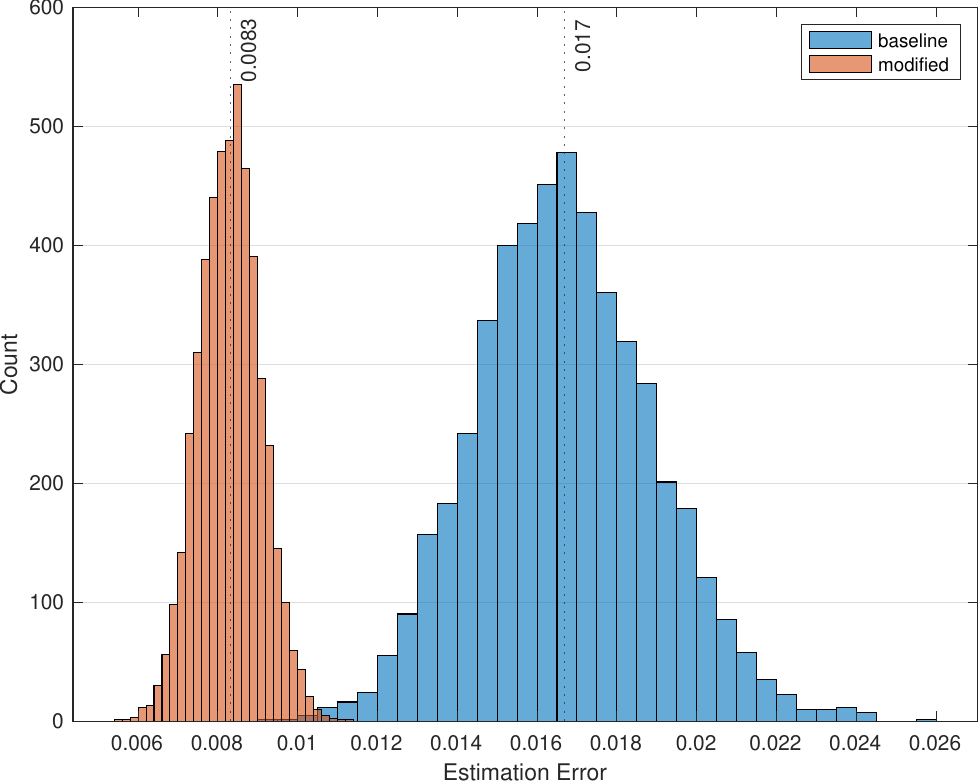}}
		
		\subfigure[{\bf $ \hat\lambda $ - $ \mathbf{60\times 60\times 30} $}]{\includegraphics[width=0.32\textwidth]{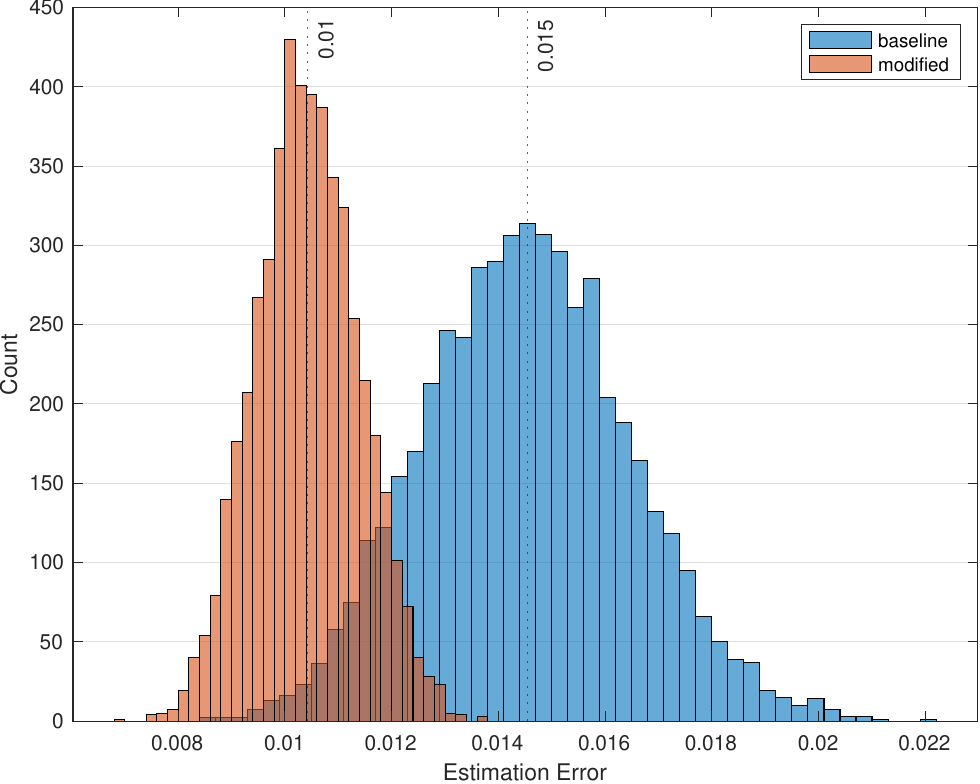}}
		\subfigure[{\bf $ \hat\mu $ - $ \mathbf{60\times 60\times 30} $}]{\includegraphics[width=0.32\textwidth]{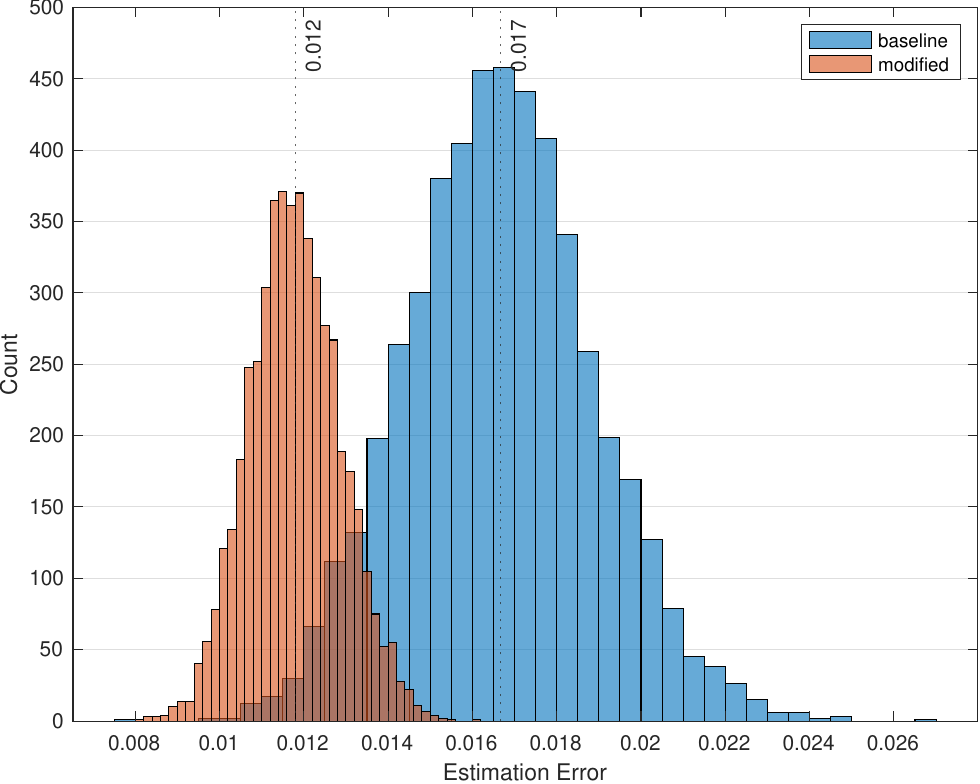}}
		\subfigure[{\bf $ \hat f $ - $ \mathbf{60\times 60\times 30} $}]{\includegraphics[width=0.32\textwidth]{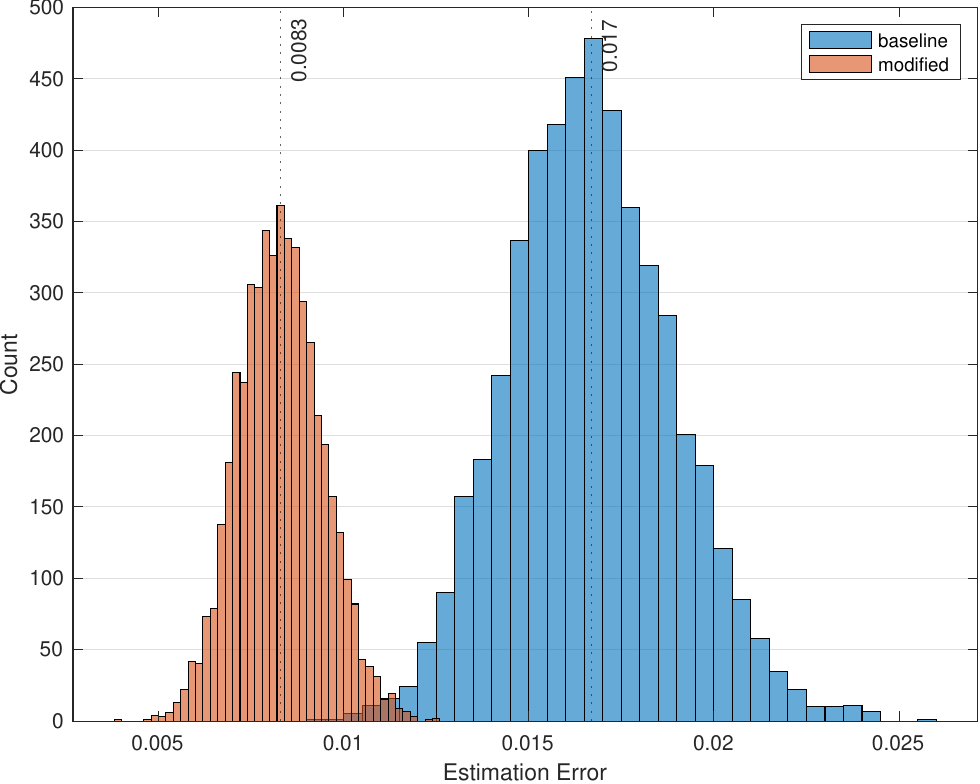}}
		
		\subfigure[{\bf $ \hat\lambda $ - $ \mathbf{60\times 30\times 30} $}]{\includegraphics[width=0.32\textwidth]{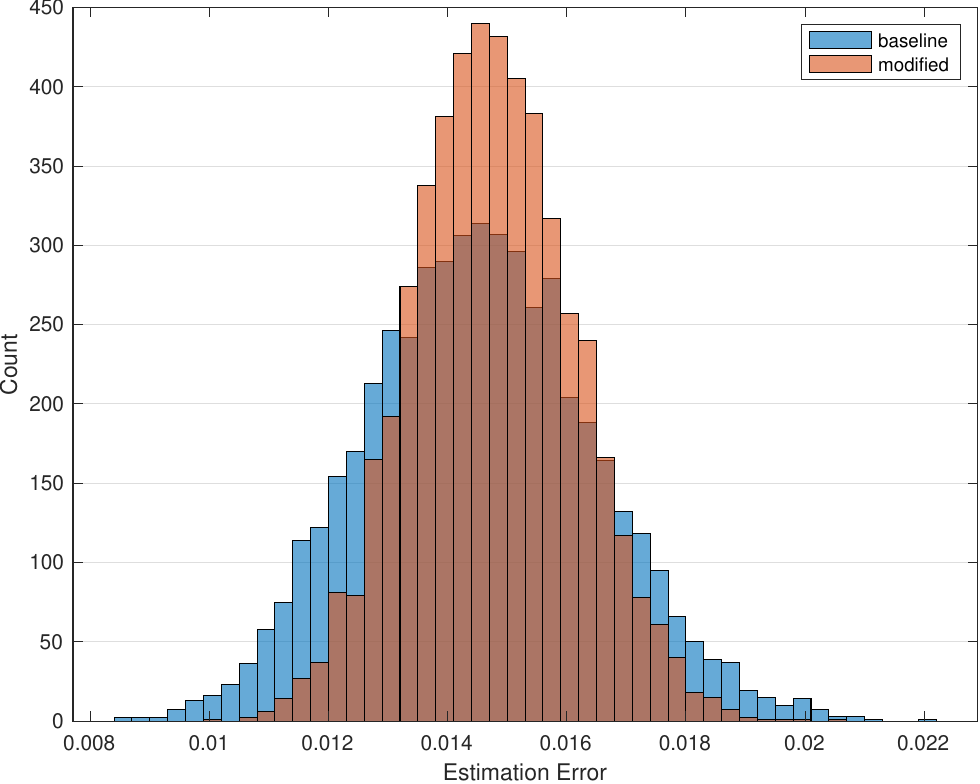}}
		\subfigure[{\bf $ \hat\mu $ - $ \mathbf{60\times 30\times 30} $}]{\includegraphics[width=0.32\textwidth]{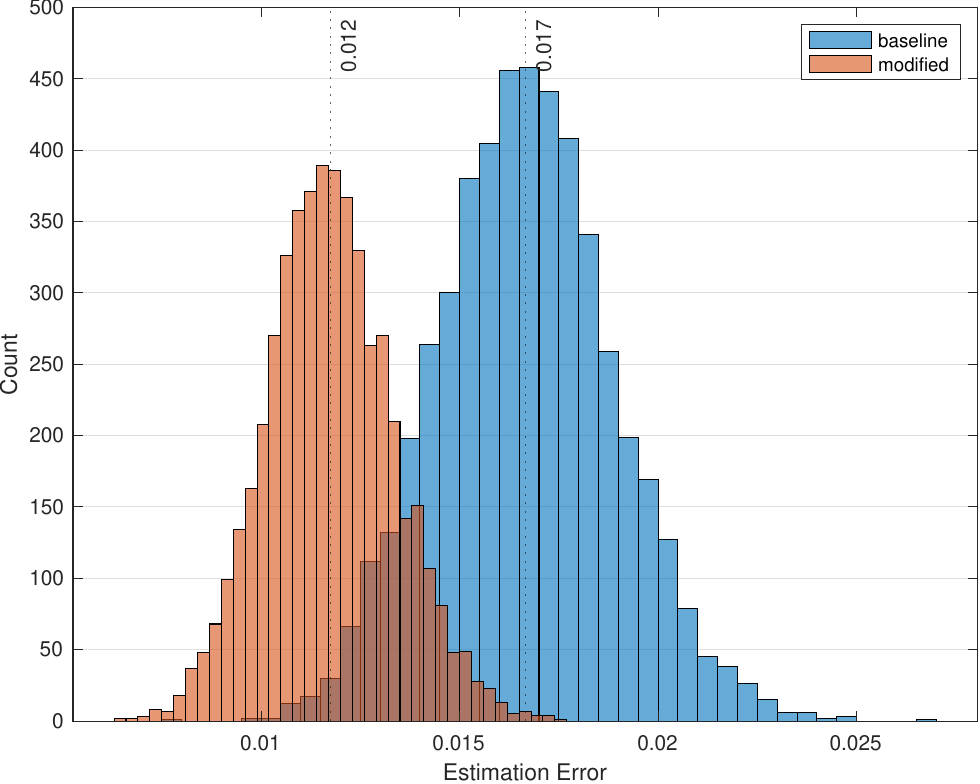}}
		\subfigure[{\bf $ \hat f $ - $ \mathbf{60\times 30\times 30} $}]{\includegraphics[width=0.32\textwidth]{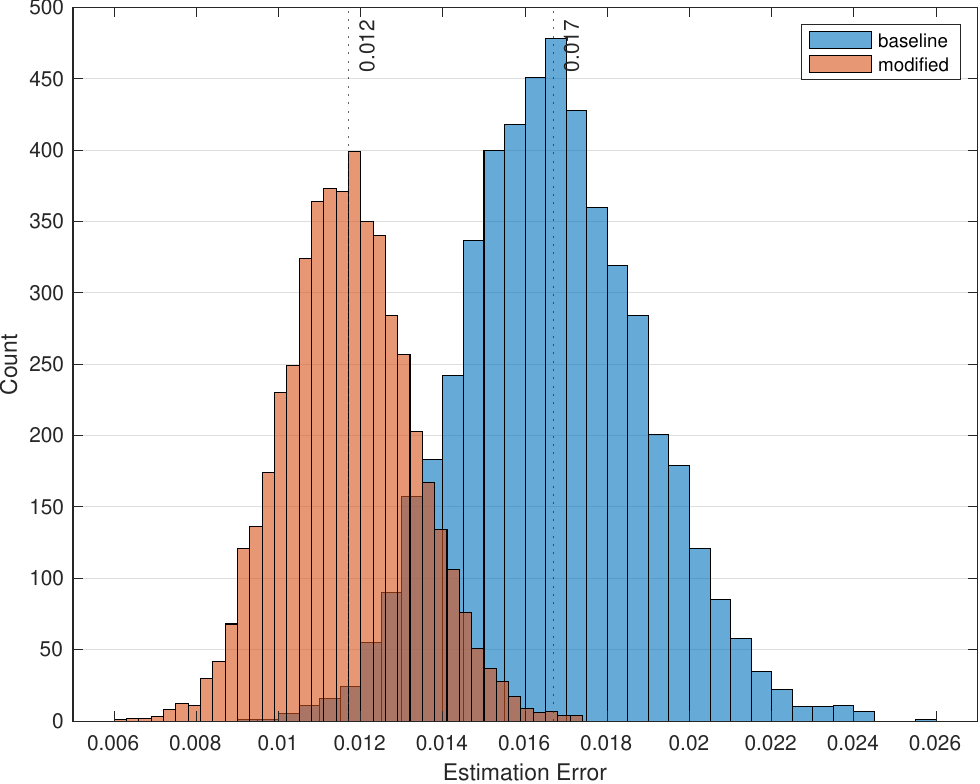}}
	\end{center}
\end{figure}

For the asymptotic distribution, due to symmetry, we focus on inference for the loading vector $\mu_{j,r}\in\mathbb{R}$. We simulate the 3-way tensor using the same parameters as previous simulation except that we fix the sample size at $(N,J,T)=(30,30,30)$. Theorem~\ref{thm:clt} shows that
\begin{equation*}
	d_r s_u^{-1}\sqrt{NJT}(\hat\mu_{j,r} - \mu_{j,r} - b_{j,r}) \xrightarrow{d}N(0,1)
\end{equation*}
which is what we aim to verify with simulations. Figure \ref{fig:clt} reports the histograms of the scaled estimation error $d_r s_u^{-1}\sqrt{NJT}(\hat{\mu}_{j,r} - \mu_{j,r})$ in 5000 MC simulations. For convenience, we only report the results of the first two elements of the loadings $\mu_{j,r}$ with $r=1,2$ and $j=1,2$. It can be seen that the finite sample distribution closely matches the asymptotic distribution, the QQ plot of the errors against the asymptotic distribution again confirms the finding. Interestingly we find that the bias appears to be negligible. Overall, these results show that the predictions of the asymptotic theory are valid in finite samples.

\setcounter{subfigure}{0}
\begin{figure}
	\caption{ \label{fig:clt} Asymptotic Distribution Reconciliation} \vskip0.1in
	{\footnotesize The DGP appears in Section \ref{sec:dgp}. We plot the histograms of the scaled estimation error $d_r s_u^{-1}\sqrt{NJT}(\hat{\mu}_{j,r} - \mu_{j,r})$ in 5000 MC simulations and compare it with the asymptotic distribution - standard normal. We also show the QQ plot of the errors against the asymptotic. For simplicity, we only report the results of the first two elements of the loadings $\mu_{j,r}, r=1,2$ and $j=1,2$ as the theoretical results hold for all elements.}
	\begin{center}
		\subfigure[{\bf Estimation Error - $\mu_{1,1}$}]{\includegraphics[width=0.45\textwidth]{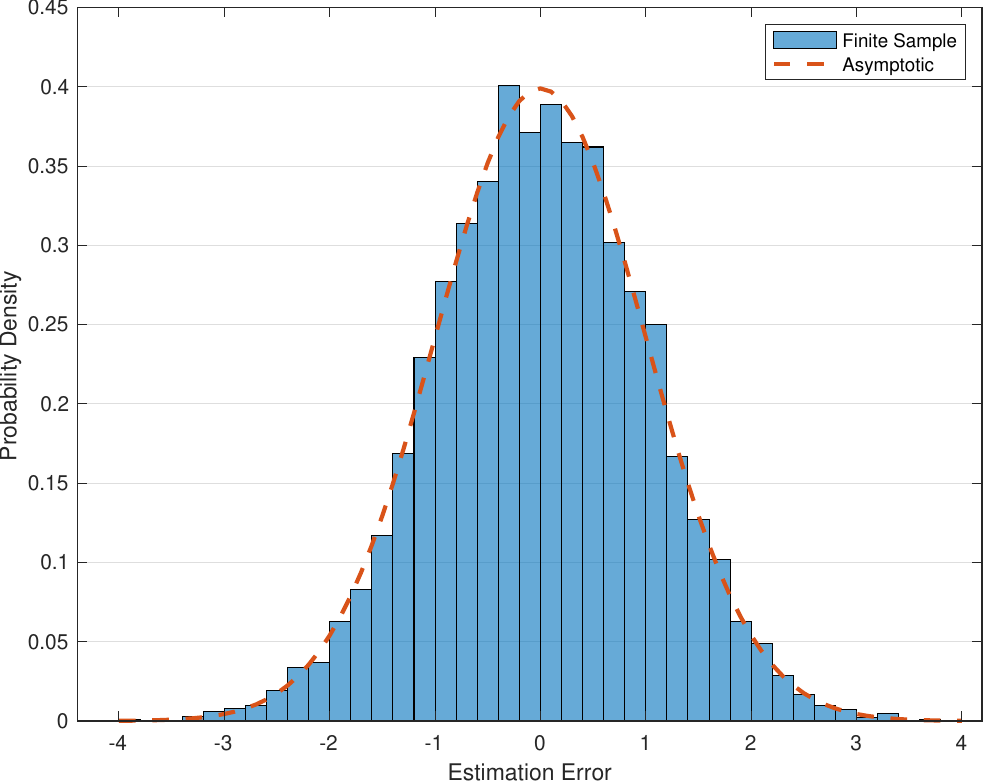}}
		\subfigure[{\bf QQ Plot - $\mu_{1,1}$ }]{\includegraphics[width=0.45\textwidth]{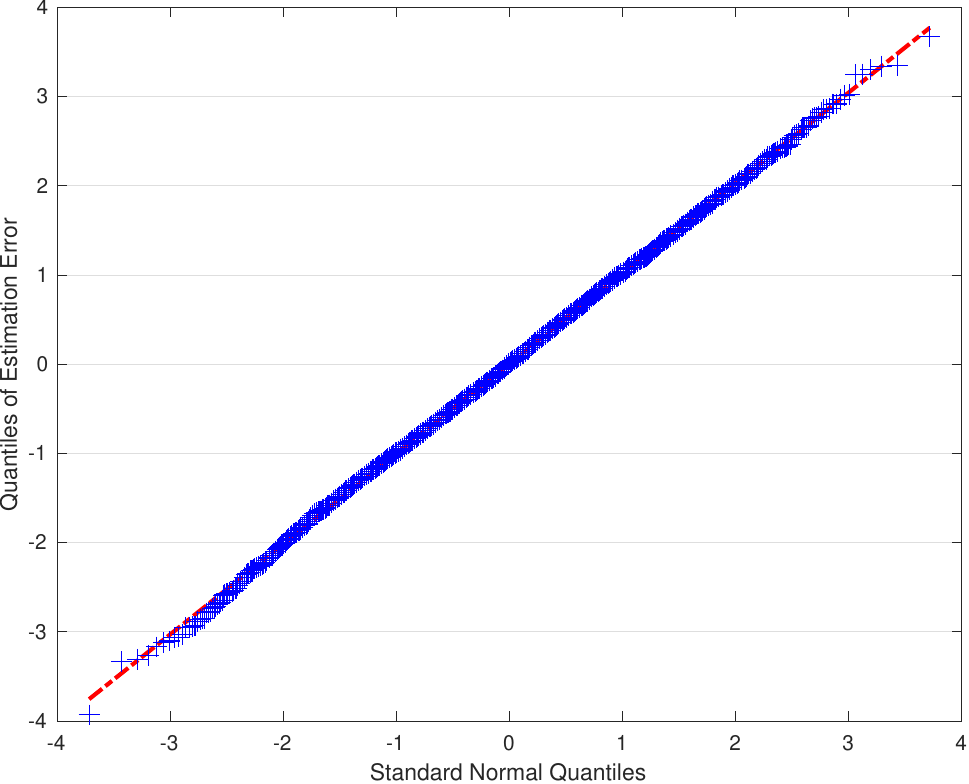}}
		
		\subfigure[{\bf Estimation Error - $\mu_{2,1}$}]{\includegraphics[width=0.45\textwidth]{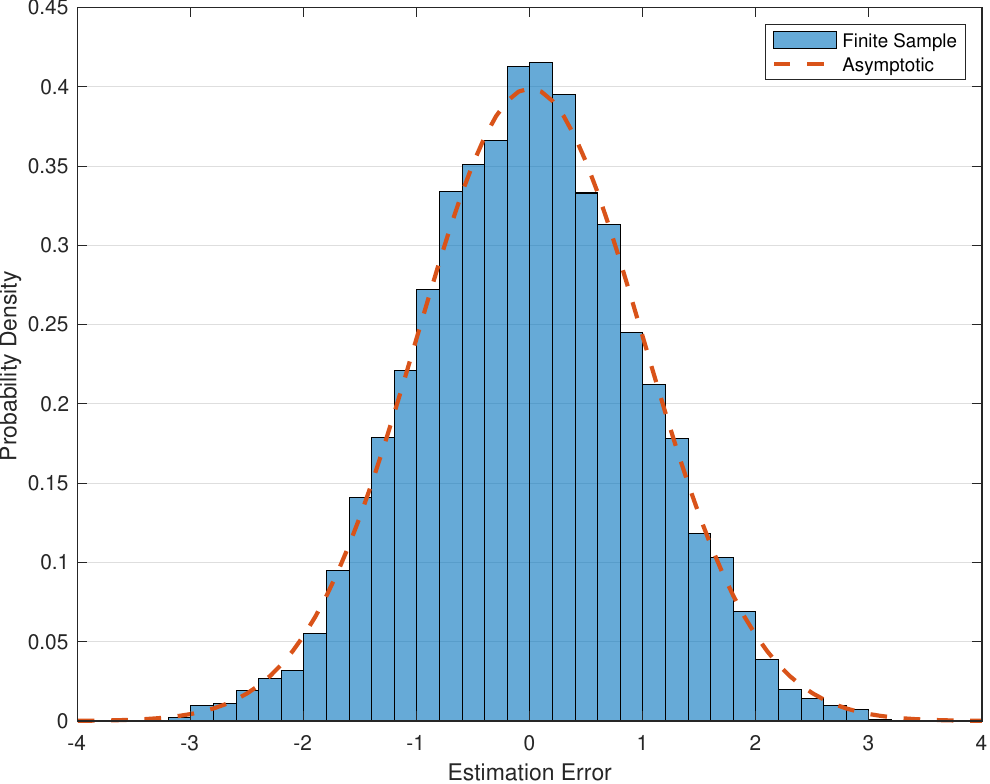}}
		\subfigure[{\bf QQ Plot - $\mu_{2,1}$}]{\includegraphics[width=0.45\textwidth]{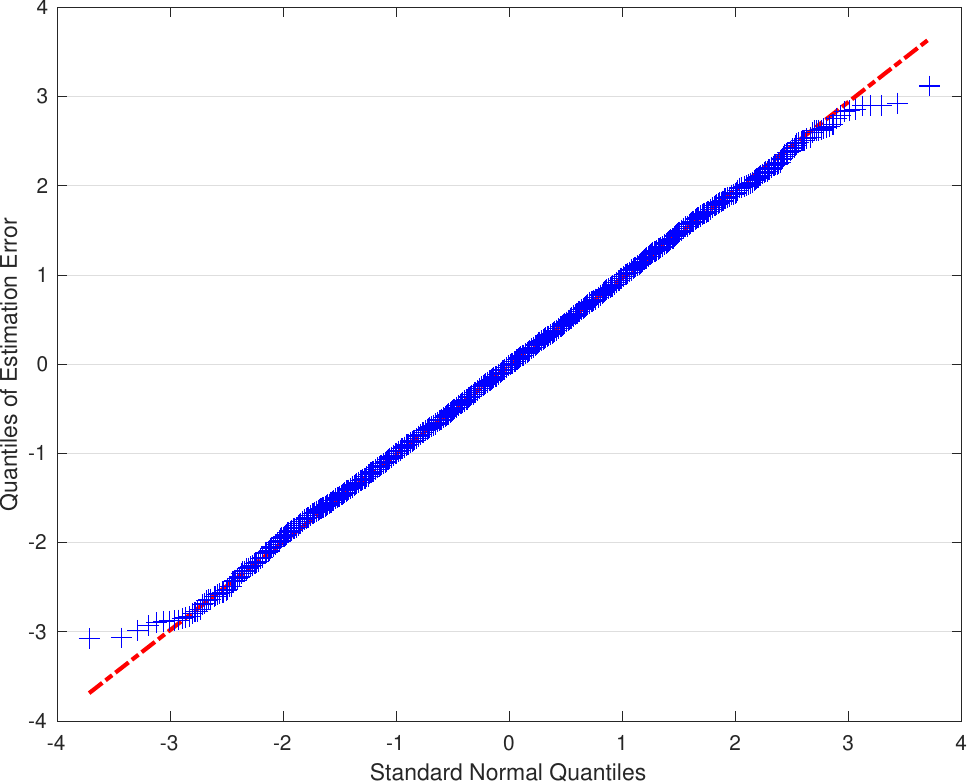}}

	\end{center}
\end{figure}

\setcounter{subfigure}{0}
\begin{figure}
	\caption{ \label{fig:clt2} Asymptotic Distribution Reconciliation (continued)} \vskip0.1in
	{\footnotesize The DGP appears in Section \ref{sec:dgp}. We plot the histograms of the scaled estimation error $d_r s_u^{-1}\sqrt{NJT}(\hat{\mu}_{j,r} - \mu_{j,r})$ in 5000 MC simulations and compare it with the asymptotic distribution - standard normal. We also show the QQ plot of the errors against the asymptotic. For simplicity, we only report the results of the first two elements of the loadings $\mu_{j,r}, r=1,2$ and $j=1,2$ as the theoretical results hold for all elements.}
	\begin{center}
		
		\subfigure[{\bf Estimation Error - $\mu_{1,2}$}]{\includegraphics[width=0.45\textwidth]{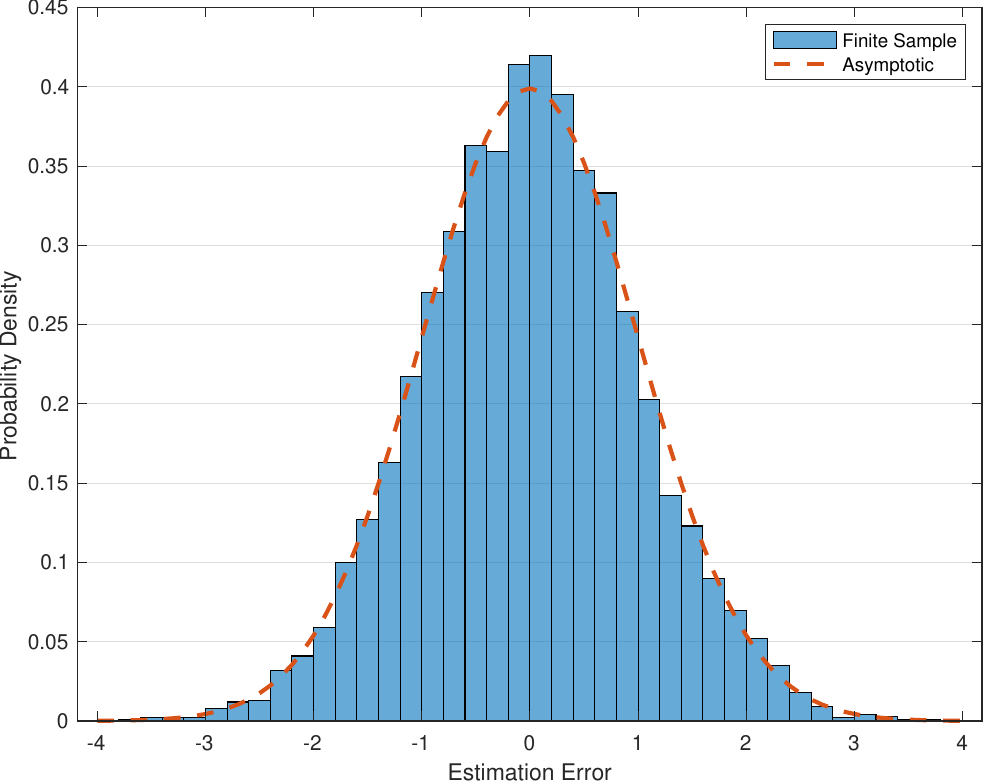}}
		\subfigure[{\bf QQ Plot - $\mu_{1,2}$}]{\includegraphics[width=0.45\textwidth]{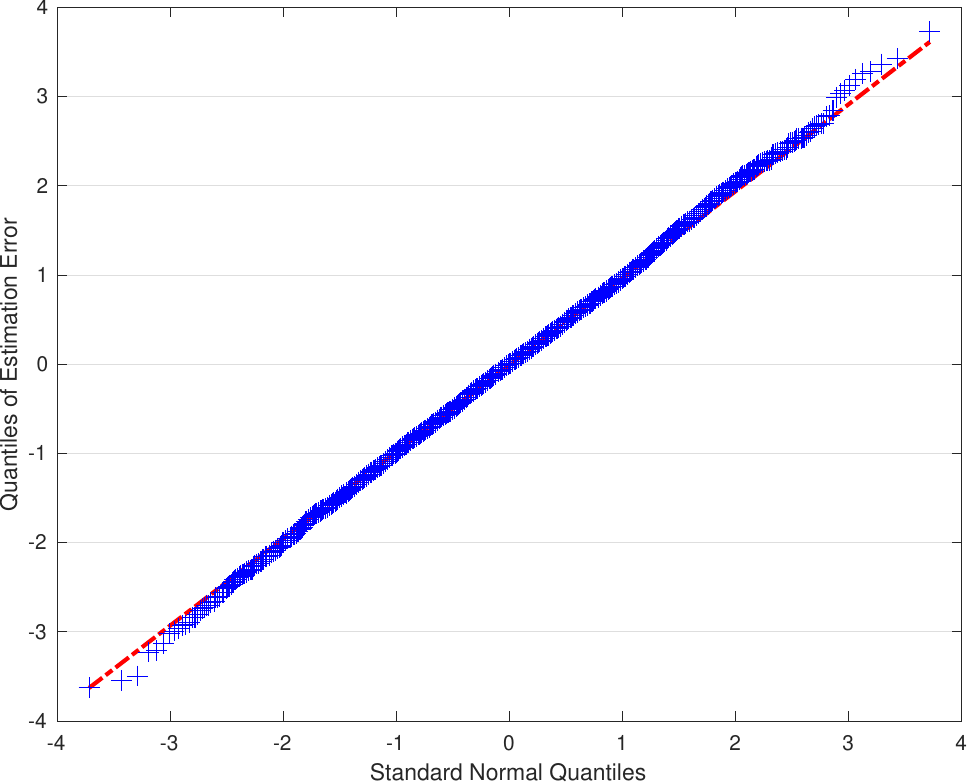}}
		
		\subfigure[{\bf Estimation Error - $\mu_{2,2}$}]{\includegraphics[width=0.45\textwidth]{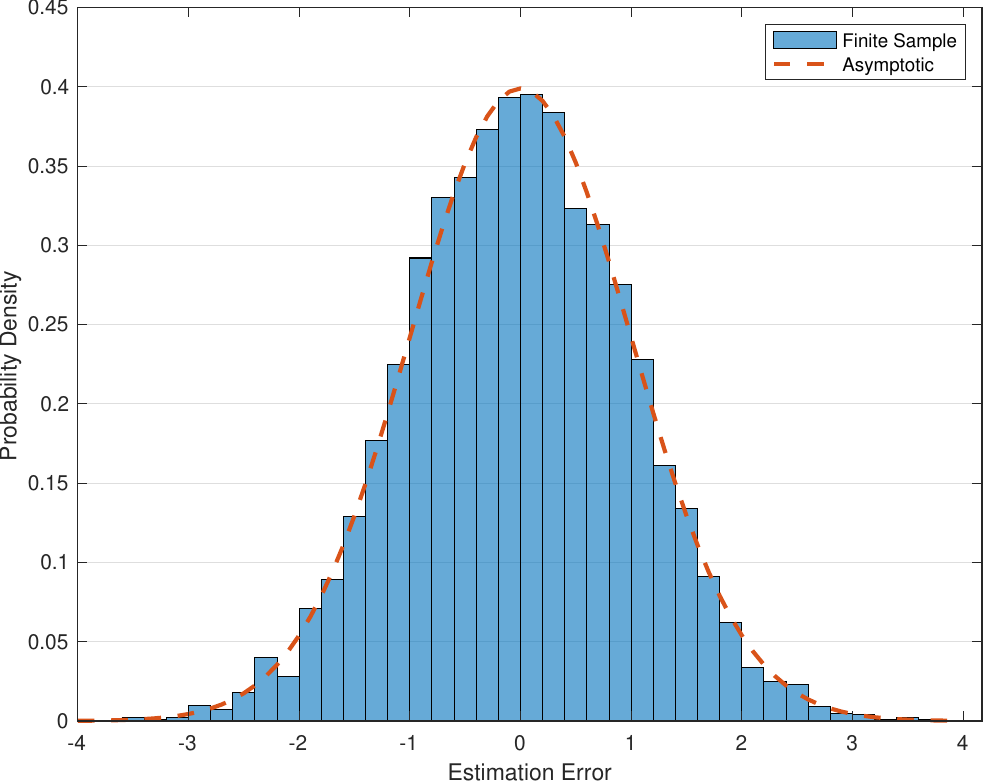}}
		\subfigure[{\bf QQ Plot - $\mu_{2,2}$}]{\includegraphics[width=0.45\textwidth]{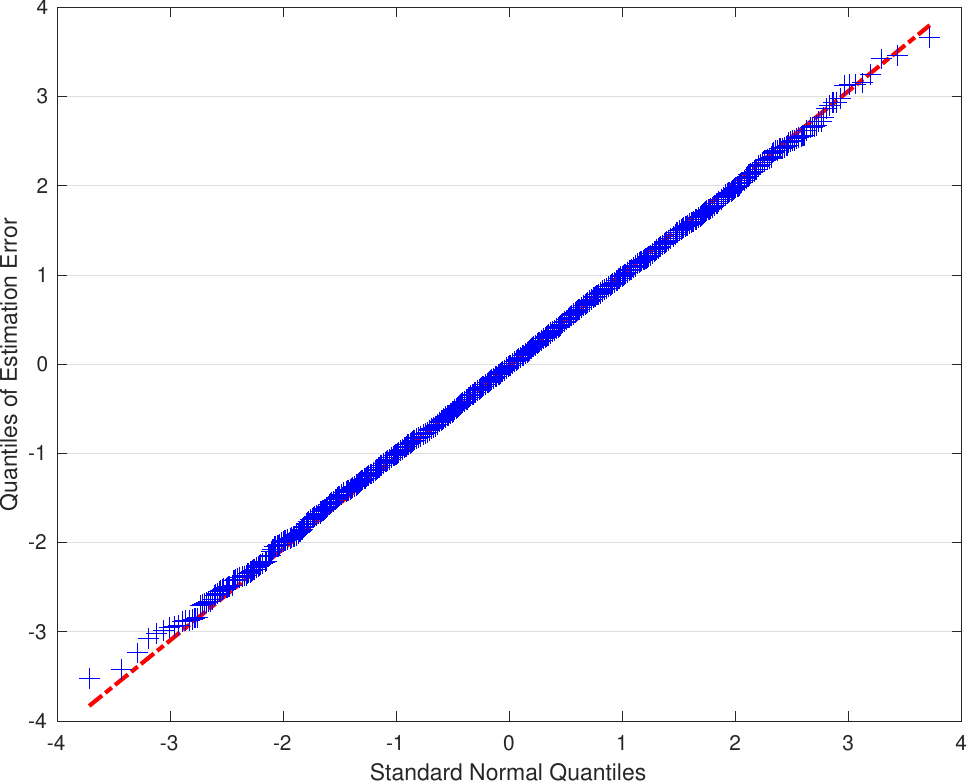}}
		
	\end{center}
\end{figure}

\subsection{TPCA vs. ALS}\label{sec:weak}

%{\color{red} Add results for weak factors, comparing TPCA to ALS/HOOI. For the weak factor model, we can take the proportional asymptotics, doubling all dimensions, and $\delta^2\sim \sigma_1^2\sim N^{1.6}$ instead of $\sigma_1^2\sim N^3$. According to Theorem~\ref{thm:rates}, the TPCA rate is
%	\begin{equation*}
%		O_P\left(\frac{1}{N^{0.6}} + \frac{1}{N^{0.2}}\right)
%	\end{equation*}
%	while the ALS rate is 
%	\begin{equation*}
%		O_P\left(\frac{1}{N^{0.6}}\right)
%	\end{equation*}
%	for $N$ large enough, provided that the number of iterations is $\bar k$ is large enough. It would be interesting to see how many iterations needed in practice to see improvements.
%}

In this subsection, we show that ALS (Algorithm \ref{alg:hooi}) has faster convergence rate than TPCA (Algorithm \ref{alg:hosvd}) under weak factor scenarios. According to Theorem \ref{thm:rates} and Theorem \ref{thm:als}, under weak factors, the TPCA rate is $O_P\left(\frac{1}{N^{0.6}} + \frac{1}{N^{0.2}}\right)$
while the ALS rate is $O_P\left(\frac{1}{N^{0.6}}\right)$ for a sufficiently large $N$, provided that the number of iterations is $\bar k$ is large enough. We follow the same DGP as Section \ref{sec:finite} except we set the signal strength $\sigma_1 = \sqrt{(NJT)^{1.6/3}}$ for weak factor. The sizes of all dimensions are set to be the same so that the finite sample properties are the same for all dimensions, and it is enough to examine only one dimension.

Figure \ref{fig:tpcavsals} (a) plots the comparison between TPCA and ALS when all dimensions are doubled from $(30,30,30)$ to $(60,60,60).$ The dotted line marks the mean of the $\ell_2$ errors. It can be noted that ALS shows more improvement in error reduction, from the mean of 0.16 to 0.13, while TPCA error is reduced from 0.21 to 0.18. The error reduction rate of ALS - 0.13/0.16 - is roughly in line with the theoretical value of $1/(2^{0.6})$, and the error reduction of TPCA - 0.18/0.21 - is also roughly aligned with the theoretical value of $1/(2^{0.2}),$ which is slightly slower under weak factors.

Figure \ref{fig:tpcavsals} (b) shows the improvement of ALS over TPCA under weak factors. It can be noted that ALS with 2 iterations already improves TPCA, and ALS with 10 iterations further improves the estimation, though by slightly less. 

\setcounter{subfigure}{0}
\begin{figure}
	\caption{\label{fig:tpcavsals} TPCA vs. ALS when doubling dimensions}
	\vskip0.1in
	{\footnotesize The weak factor DGP appears in Section \ref{sec:weak}. We plot the histograms of $\ell_2$ losses of the estimated factor/loading of TPCA vs ALS in 5000 MC simulations. The baseline DGP has sizes $(N,J,T)$ = (30,30,30), and it is compared to the DGP has double the sizes $(N,J,T)$ = (60,60,60). The dotted line marks the mean of the $\ell_2$ errors. The plot (a) shows that ALS has larger improvement than TPCA under weak factors, and the plot (b) shows how ALS iterations improves the convergence rate.}
	\begin{center}
		\subfigure[{\bf TPCA vs. ALS}]{\includegraphics[width=0.45\textwidth]{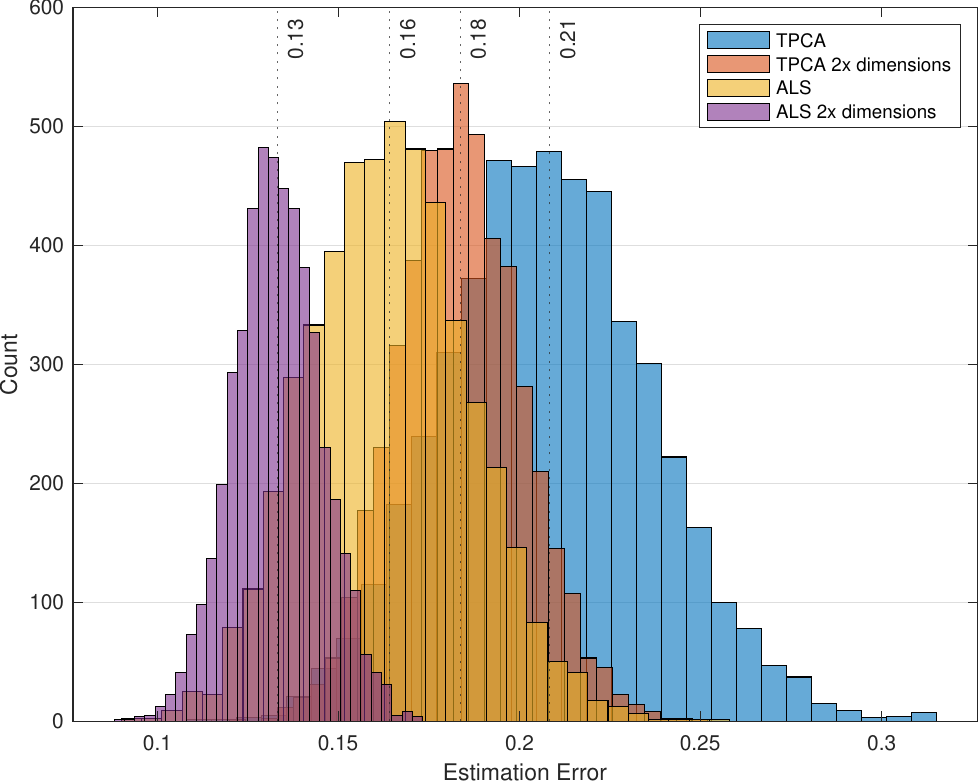}}
		\subfigure[{\bf ALS iterations}]{\includegraphics[width=0.45\textwidth]{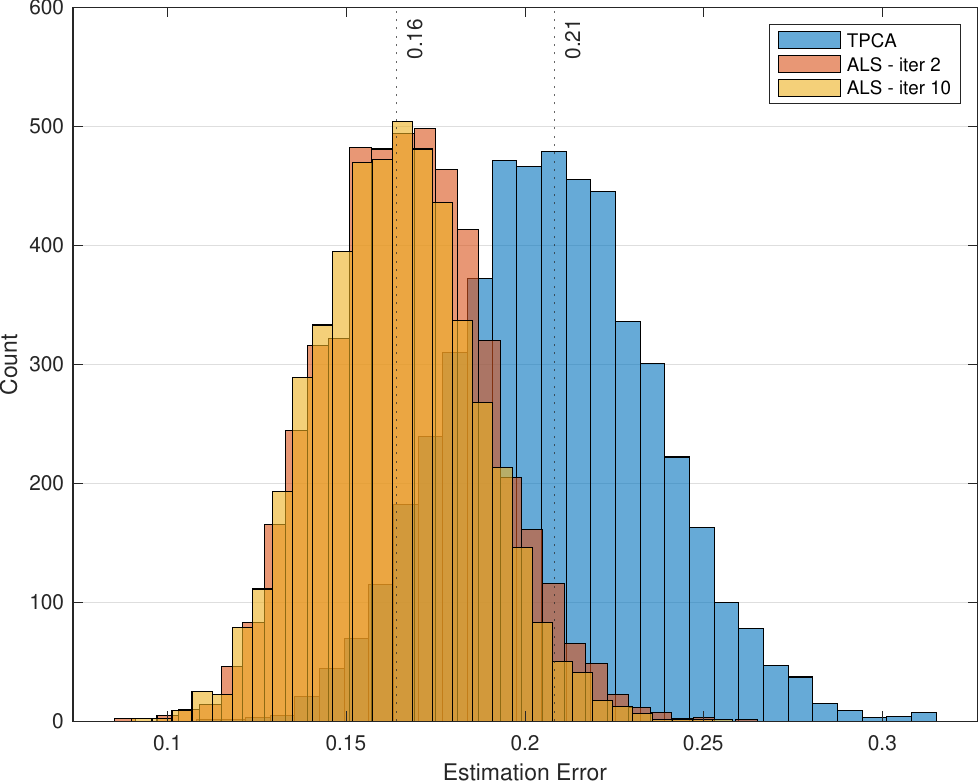}}
	\end{center}
\end{figure}

\subsection{Testing the Number of Factors}
We conclude with power properties of our maximum eigenvalue ratio test for the number of factors discussed in Section \ref{sec:test}. The DGP is generated as described in Section \ref{sec:dgp}. We generate rank $(1,2,2)$ model, and test the null hypothesis that the rank is 1 against the alternative that there are more than 1 but less than $K$ factors for each matricization. 

\smallskip

The parameters are designed as follows: (1) the scale component is $\sigma_r$ = $d_r \times \sqrt{NJT}$ with $d_1=2$ and we gradually increase $d_2$ to study the power properties of the test, so when $d_2=0$, the empirical rejection probability corresponds to empirical size of the test, and none zero $d_2$ corresponds to empirical power for rank 2 dimensions, (2) we study cases where the idiosyncratic errors are generated with Gaussian distribution or Student's t distribution (the degrees of freedom are 5), and in both cases the variance of the errors is normalized to be $s_u=1,$ (3) we study finite sample properties of the test using 3-way tensor, and the sizes of the tensors is $(N,J,T) = 30\times40\times 50$, with the first dimension being the smallest.

%\smallskip
%
%The $p$-value combination strategies we consider are maximum, minimum, median, and mean. The $p$-value combinations have to be scaled properly in order to be considered as valid $p$-values, please see \cite{vovk2020combining} for discussion of different $p$-value combinations. Specifically, for a $d$-way we combine the p-values using 	the $p_{\rm min},$ 
%$p_{\rm median},$ and 
%$p_{\rm mean}$ rules - see equation (\ref{eq:pvaluerules}). We also a fourth combination rule: $p_{\max}$ := $\max (p_1,\ldots,p_d).$

\smallskip

We perform the test for the number of factors as discussed in Section \ref{sec:test}, where we approximate the asymptotic distribution of the statistics by randomly generating Gaussian matrices 5,000 times, and we also replicate the simulations for 5,000 times to calculate the empirical rejection probability for each scenario.

\setcounter{subfigure}{0}
\begin{figure}
	\caption{ \label{fig:power1} Testing the Number of Factors: Power Curves} \vskip0.1in
	{\footnotesize The power curves are plotted for testing the null hypothesis of 1 factor against alternative of more than $1$  but less than $K$ factors, with $K=3,5,7$. Plots (a) through (c) report the tests for each matricization separately for a 3rd order tensor with Gaussian errors.}
	\begin{center}
		\subfigure[{\bf 1st Matricization}]{\includegraphics[width=0.45\textwidth]{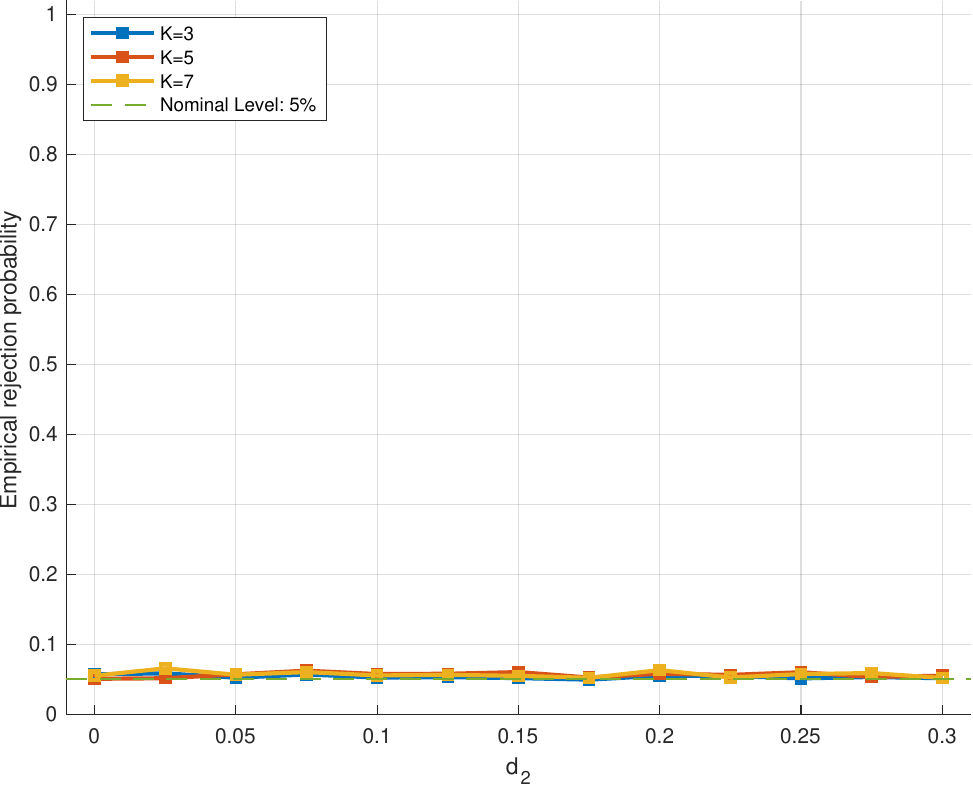}}
		\subfigure[{\bf 2nd Matricization}]{\includegraphics[width=0.45\textwidth]{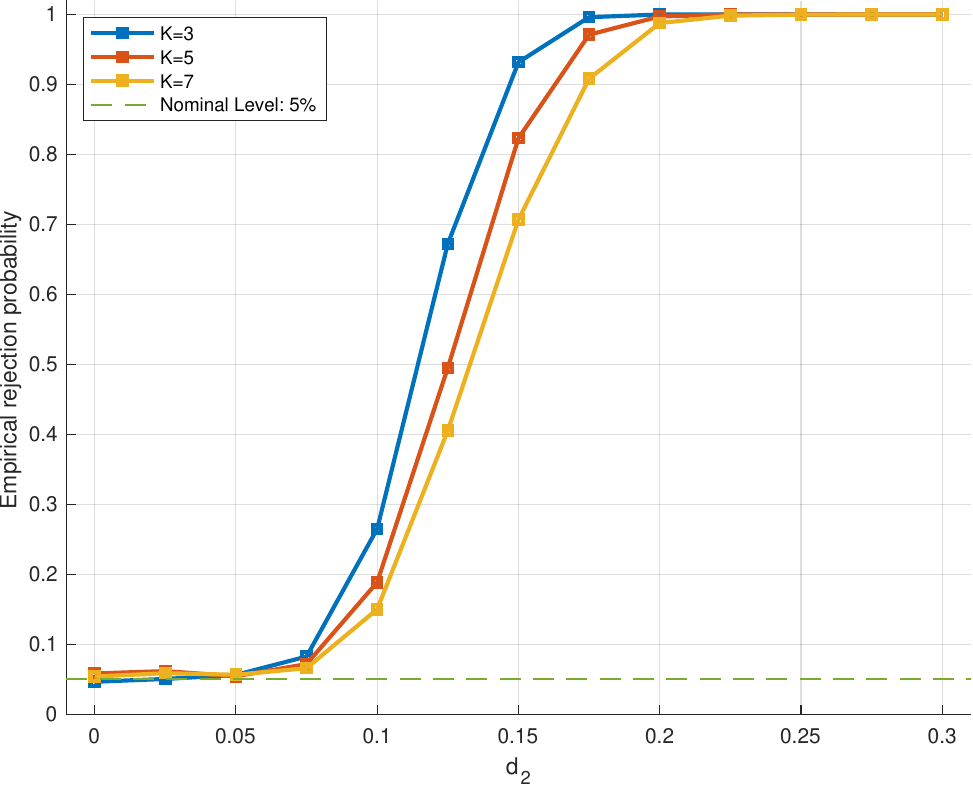}}
		
		\subfigure[{\bf 3rd Matricization}]{\includegraphics[width=0.45\textwidth]{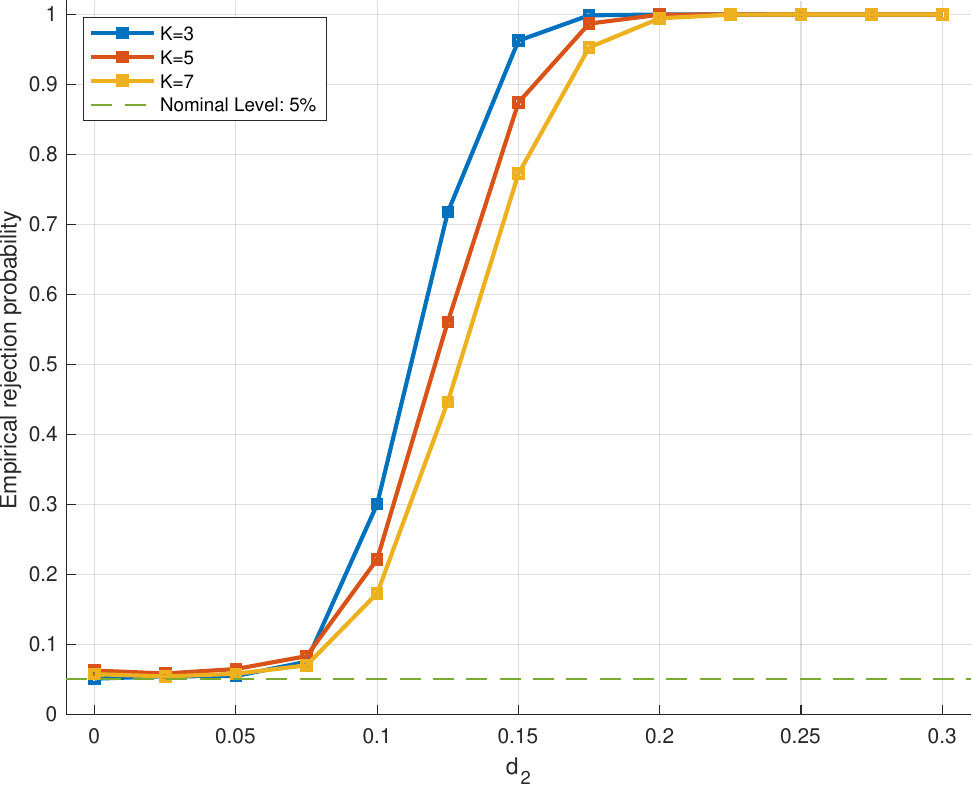}}
	\end{center}
\end{figure}

\setcounter{subfigure}{0}
\begin{figure}
	\caption{ \label{fig:power2} Testing the Number of Factors: Power Curves for Different DGPs} \vskip0.1in
	{\footnotesize The power curves are plotted for testing the null hypothesis of 1 factor against the alternative of 2-5 factors, for cases of 3-dim tensor normal/t-distribution errors. We report the empirical power of testing on individual matricizations.}
	\begin{center}
		\subfigure[{\bf 3 dim - Gaussian Errors}]{\includegraphics[width=0.45\textwidth]{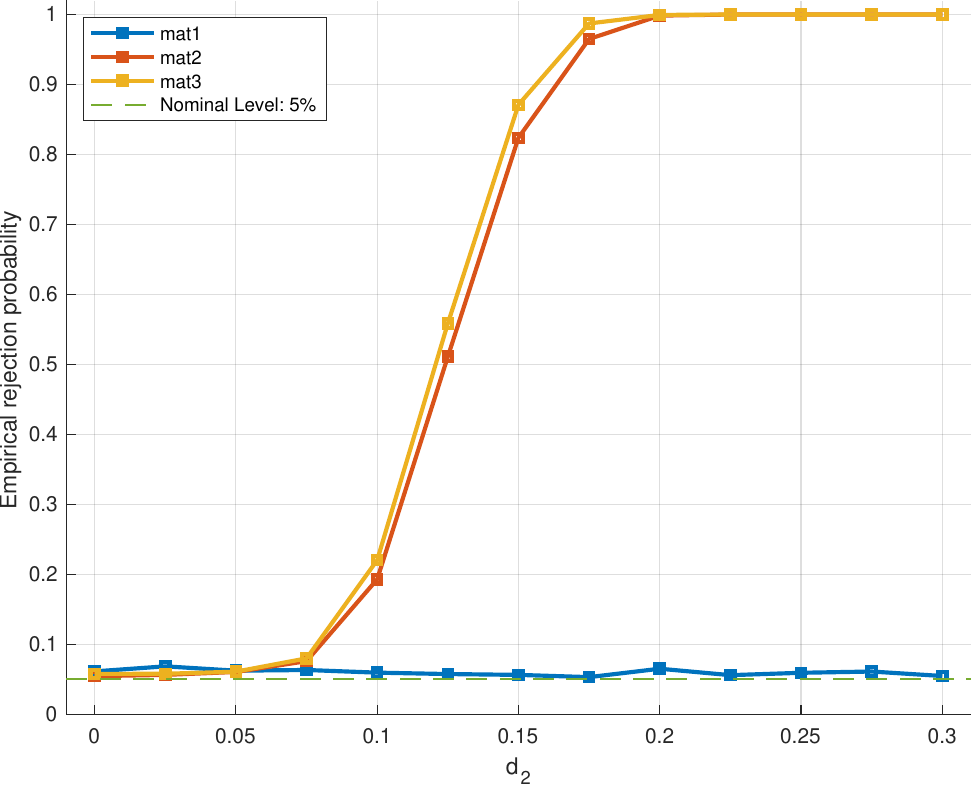}}
%		\subfigure[{\bf 3 dim - Gaussian Errors}]{\includegraphics[width=0.45\textwidth]{threewayfactors/test_3dim_2_g.pdf}}
		\subfigure[{\bf 3 dim - t Errors}]{\includegraphics[width=0.45\textwidth]{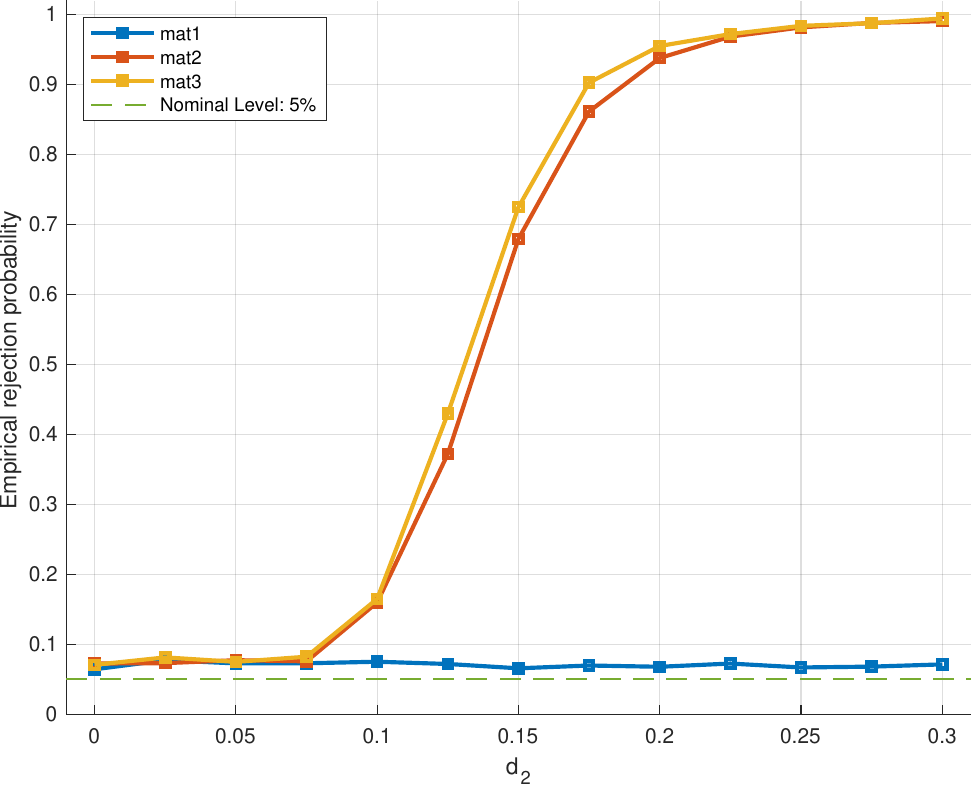}}
%		\subfigure[{\bf 3 dim - t Errors}]{\includegraphics[width=0.45\textwidth]{threewayfactors/test_3dim_2_t.pdf}}
		
	\end{center}
\end{figure}

Figure \ref{fig:power1} reports the empirical rejections probabilities of the test for the null hypothesis of 1 factor against alternative of more than 1 factor but less than $K$ factors, with $K=3,5,7$, on a $3^{\text{rd}}$ order tensor. The test is performed on each matricization individually and plotted separately. The power curves show that the empirical size of test is very close to the nominal level of $5\%$,\footnote{Since the rank is 1 for the first dimension, the power stays at the nominal level.} and the empirical power of the test reaches $1$ as the strength of the second factor $d_2$ increases. The comparison among different $K$ values implies that, when the true number of factors is within range of the alternative hypothesis, the tighter the range of the alternative is, the more likely we reject the null when the alternative is true. This finding is true no matter which matricization we perform the test on.

\smallskip

We also make comparisons of the test for different types of errors. The two plots of Figure \ref{fig:power2} correspond to tensor with: (a)  Gaussian errors, (b) Student's t distributed errors. Among different matricizations, the dimension with the larger size (`mat3' ) tend to climb faster to one than smaller size (`mat2'), meaning the matricization with the larger dimension gives the higher empirical power.  

\smallskip

Finally, while our simulation results rely on the critical values that assume Gaussian errors, we learn from Figure \ref{fig:power2} (b) that the test performs well for non-Gaussian distributions albeit with some loss of power.

\section{Empirical Illustration: Dealing with Missing Firm Characteristics
\label{sec:empirical}}

\cite{bryzgalova2022missing} drew attention to the phenomenon of missing data in
firm characteristics, which are the cornerstone of academic research in finance. They provide a comprehensive analysis of missing data in firm characteristics, and show that patterns of missing characteristics vary substantially across characteristics. As a remedy they propose a statistical factor model for imputing missing values and investigate the impact of missingness on asset returns.\footnote{Factor models can also be used to impute the missing counterfactual potential outcomes; see \cite{bai2021matrix}.}

Firm characteristics data are of the tensor type. Indeed, for each firm, one has across time a set of observed or missing characteristics. The statistical methods proposed by \cite{bryzgalova2022missing} do not exploit the tensor data structure, while in this empirical application we do and we show that using our proposed tensor data models improve on their time series of cross-sectional models.

We conduct the empirical study using the dataset created by \cite{freyberger2020dissecting}. The data cover 35 firm characteristics, which are described in Table \ref{tab:chars}. The dataset is monthly and ranges from January 1966 to December 2020, and there are a total of 13588 firms in the entire sample. This number is smaller than 22630 in \cite{bryzgalova2022missing}, because we only have access to the pre-cleaned dataset that has no missingness in the cross-section of characteristics, i.e., all characteristics exist for all firms at any time period. However, this difference does not affect the empirical results because the RMSE are calculated based on simulated (masked) missing characteristics both in this paper and in \cite{bryzgalova2022missing}. The results calculated based on cross-sectional factor model are also very close to those in \cite{bryzgalova2022missing}. This tensor dataset is highly unbalanced where only $15.78\%$ of the data are not missing, and the average length of non-missing characteristics of all firms is 104.16 months.

Prior to estimating the tensor factor and the cross-sectional factor models, the raw characteristics are converted into rank quantiles within range $\left[ -0.5,0.5\right].$ This is done because: 1) different characteristics have different scales, and PCA has best performance when applied to data of similar scales, 2) centered rank-normalized characteristics are stationary in the cross-section and over time, and 3) it handles outliers. In practice, the predicted missing rank quantiles can be converted back to the original scale of the raw data. 

Instead of fitting cross-sectional two-way factor models to each slice of the tensor (i.e., each matrix at period $t$), we fit three-way factor models on non-overlapping 5-year horizons. There is a considerable amount of parameter proliferation when estimating time series of cross-sections, thus we would expect a better performance from the parsimonious tensor-based approach. We compare our model to the traditional two-way approach in terms of imputation accuracy, and several other alternatives are also considered. See Section \ref{appsec:imputation} for detailed description of the model and other approaches being compared.

\subsection{Evaluation Metrics}

Similar to \cite{bryzgalova2022missing}, we consider the standard evaluation metrics for model prediction errors - RMSE (root-mean-squared errors). The aggregate RMSE is averaged over all firms, characteristics, and time periods, which is calculated as follows:

\begin{equation*}
	\textrm{RMSE} = \sqrt{\frac{1}{T} \sum_{t=1}^{T} \frac{1}{L} \sum_{\ell=1}^{L} \frac{1}{N_t} \sum_{i=1}^{N_t} (C_{i,t,\ell} - \hat{C}_{i,t,\ell})^2}.
\end{equation*}

%We also consider the RMSE for each characteristic separately, and for each time period as follows:
%
%\begin{equation*}
%	\textrm{RMSE}_\ell = \sqrt{\frac{1}{T} \sum_{t=1}^{T} \sum_{i=1}^{N_t} (C_{i,t,\ell} - \hat{C}_{i,t,\ell})^2}, \qquad \textrm{RMSE}_t = \sqrt{\frac{1}{L} \sum_{\ell=1}^{L} \frac{1}{N_t} \sum_{i=1}^{N_t} (C_{i,t,\ell} - \hat{C}_{i,t,\ell})^2}.
%\end{equation*}

We report both in-sample (IS) and out-of-sample (OOS) RMSE. For measuring out-of-sample RMSE, we randomly mask observed characteristics and calculate RMSE comparing predicted characteristics and the true masked ones. We take the same appraoch as \cite{bryzgalova2022missing} to randomly mask $10\%$ of the characteristics using two schemes:
\begin{enumerate}
	\item Missing Completely at Random (MCAR). $10\%$ of the characteristics are masked completely at random.
	\item Block Missing. We again mask $10\%$ of the characteristics. In order to capture the time-series pattern of the missingness, we randomly mask characteristics in blocks of one year. Based on the observational result of missing characteristics in \cite{bryzgalova2022missing}, $40\%$ of the blocks are at the beginning.
\end{enumerate}

We also report the $R^2$ that measures the explained variation of the model relative to the total variation, which is a transformation of the RMSE.

\subsection{Number of Factors for Characteristics Tensors}

Estimating both the tensor and cross-sectional factor models involves determining the number of factors. As noted in \cite{bryzgalova2022missing}, the optimal number of factors for the cross-sectional models is between 6 and 8, and they select 6 factors as their final model. In comparison, we use the test introduced in Section \ref{sec:test} to determine the number of factors for each dimension. The test is conducted with the following hypothesis:
\begin{equation*}
	\mathrm{H_0}:\; \text{$\leq k$ factors} \qquad \rm{vs.} \qquad H_1:\; \text{the number of factors is $>k$, but $\leq K$}.
\end{equation*}
where $K = k+1$ and we increase $k$ from $1$ to $20$ ($20$ is the upper bound of the number of factors we deem reasonable). Since the test is conducted multiple times on the 5-year sub-samples, we use Bonferroni correction to reduce the probability of type I errors. 

Figure \ref{fig:nfactors} plots the p-values of test where the x-axis is test parameter $k$, and for each $k$ there are 11 dots representing the p-value of each 5-year sub-sample. The test is conducted on each dimension and for each missing scheme - missing-completely-at-random (MCAR) and missing in blocks of one year (Block). The red dash line is the Bonferroni adjusted significance level. The test result suggests that there is large number of factors in the characteristic dataset (i.e., there is no sharp change in the eigenvalues detected). Therefore, we select 20 factors as the baseline model for optimal results\footnote{The test result is aligned with the empirical finding that, when the number of factors is smaller than 20, the more factors the better the imputation performance. We also considered several alternative testing methods leading to the same conclusion.}, and we also present the results with 6 factors for a comparison with the cross-sectional model.

\begin{figure}
	\caption{ \label{fig:nfactors} Testing the Number of Factors: p-values} \vskip0.1in
	{\footnotesize This figure plots the p-values of test where the x-axis is test parameter $k$, and for each $k$ there are 11 dots representing the p-value of each 5-year sub-sample. The test is conducted on each dimension and for each missing scheme - missing-completely-at-random (MCAR) and missing in blocks of one year (Block). The red dash line is the Bonferroni adjusted significance level.}
	\begin{center}
		\subfigure[{\bf Dimension N - MCAR}]{\includegraphics[width=0.45\textwidth]{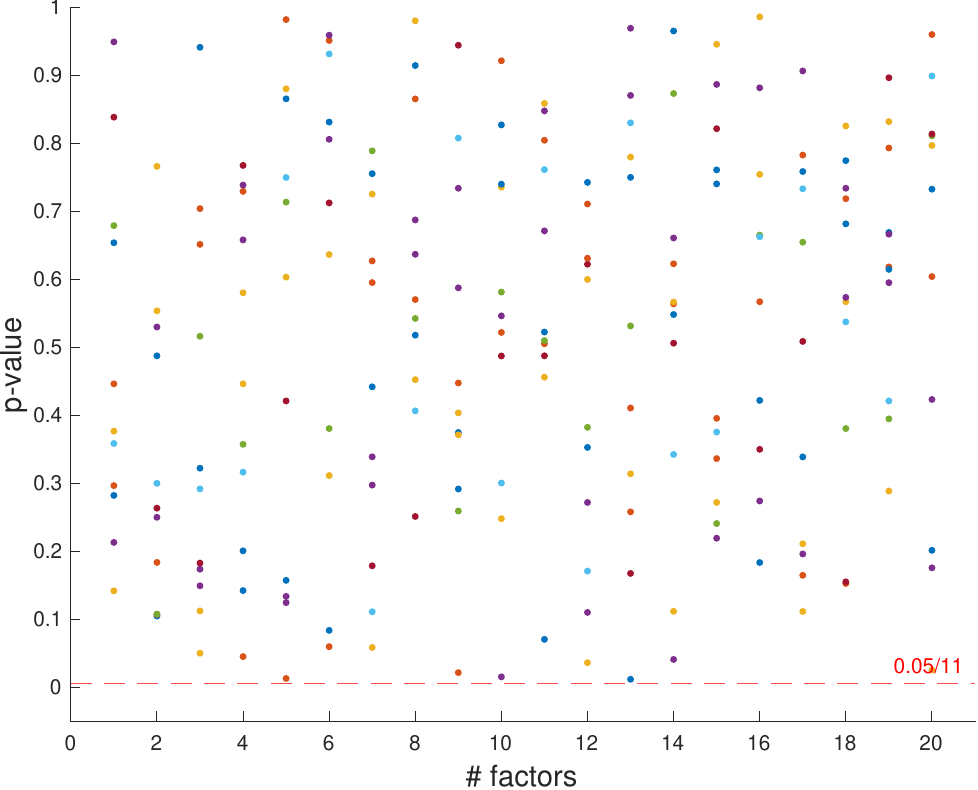}}
		\subfigure[{\bf Dimension T- MCAR}]{\includegraphics[width=0.45\textwidth]{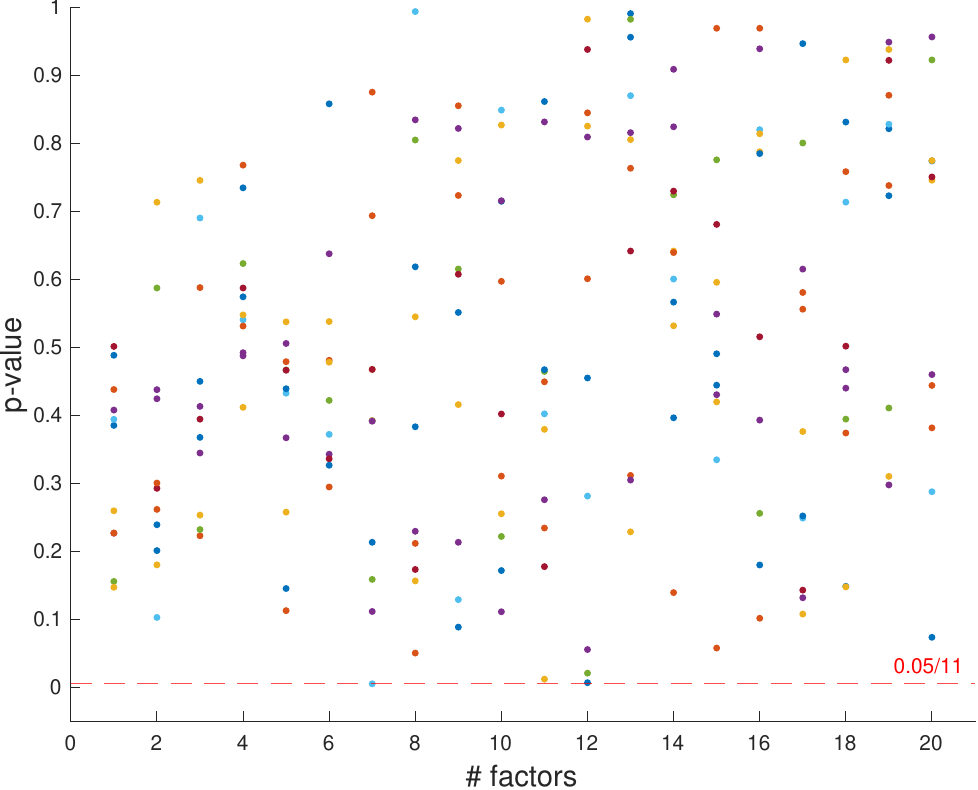}}
		
		\subfigure[{\bf Dimension L - MCAR}]{\includegraphics[width=0.45\textwidth]{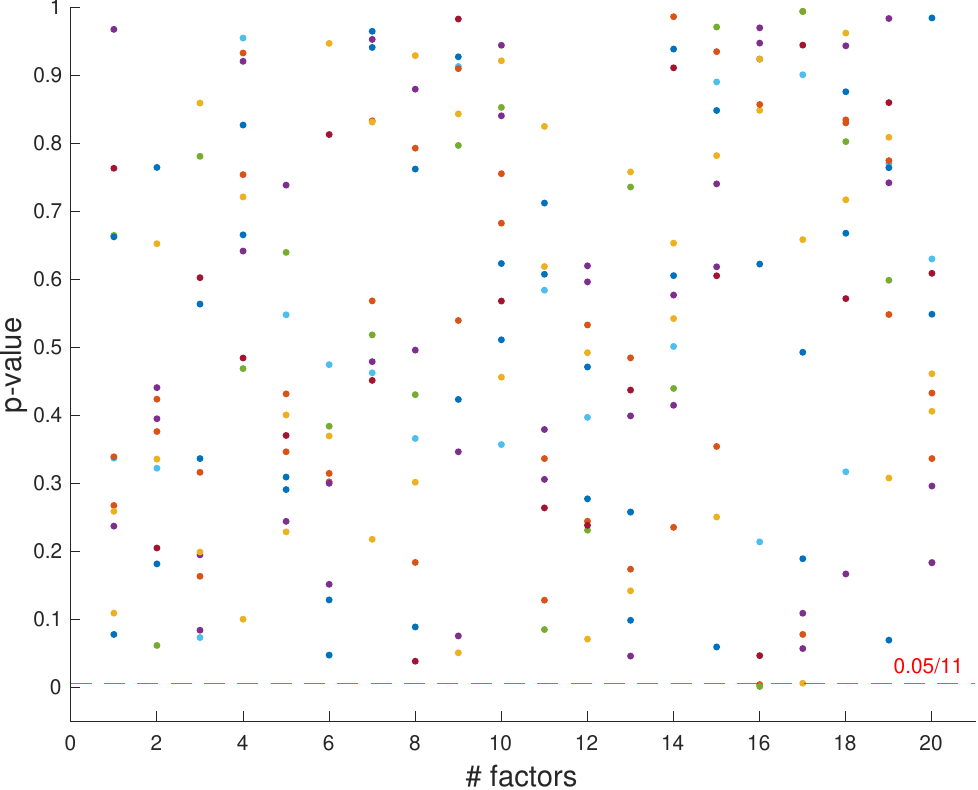}}
		\subfigure[{\bf Dimension N - Block}]{\includegraphics[width=0.45\textwidth]{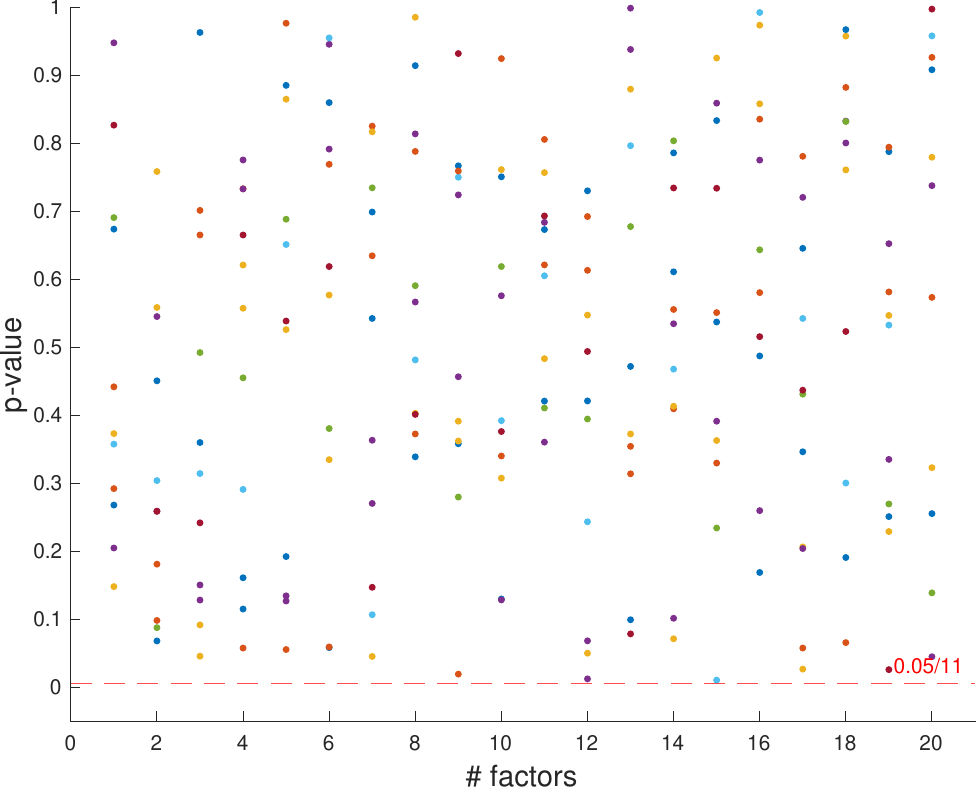}}
		
		\subfigure[{\bf Dimension T - Block}]{\includegraphics[width=0.45\textwidth]{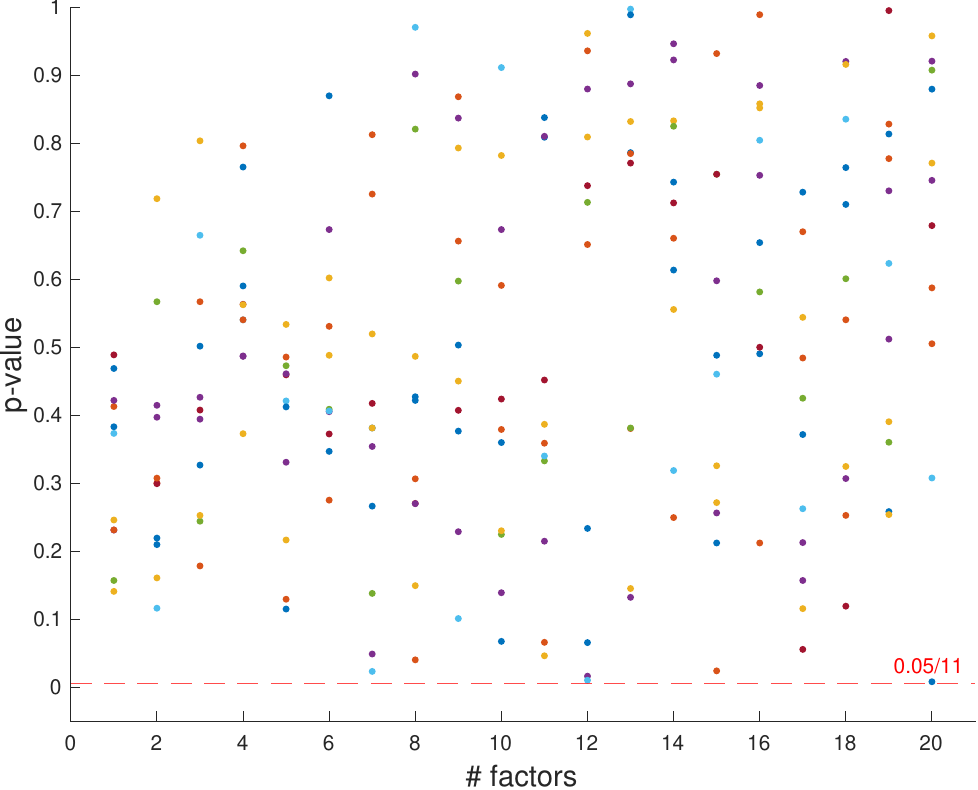}}
		\subfigure[{\bf Dimension L - Block}]{\includegraphics[width=0.45\textwidth]{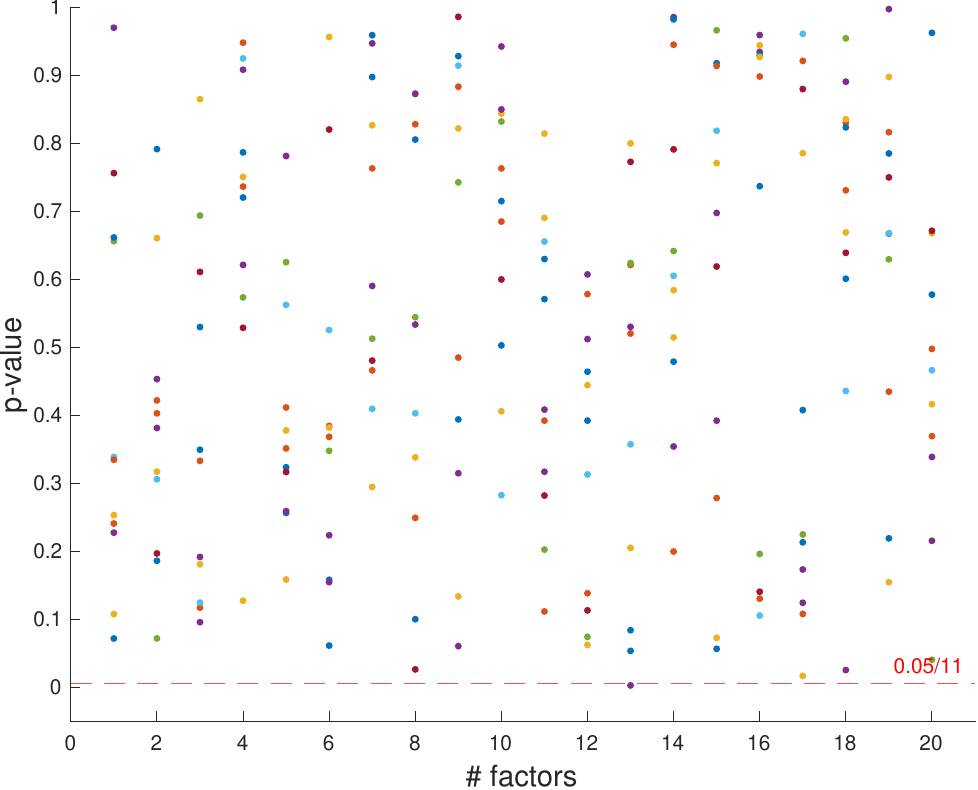}}
	\end{center}
\end{figure}

\subsection{Imputation Results}

In this subsection, we compare all the methods listed in Table \ref{tab:methods}\footnote{BF-TPCA is not reported because it shows slight worse performance than B-TPCA, and B-TPCA is the considered the best model among all.}. Some methods require backward and/or forward information, and they become inapplicable when this information is not available. Hence, we use the natural fallback approach, where we replace a method that requires time-series information by a pure lower-level method. For example, B-TPCA is replaced by TPCA and B-XS-ridge is replaced by XS-ridge when the previous characteristic is not observed. 

\smallskip

Table \ref{tab:results} summarizes the main imputation results of different methods. We report the in-sample, out-of-sample missing completely at random (OOS MCAR), out-of-sample block missing (OOS Block) for all characteristics using the full sample period. The cross-sectional (XS-ridge) based methods and the Tensor PCA (TPCA) based method are estimated with different number of factors $R$ with $R\in \{6, 20\}$. The first thing to note is that the cross-sectional median method produces imputation errors that are more than twice as large compared to previous value and autoregression methods, both in-sample and OOS missing completely at random. This demonstrates the importance of considering time-series dependency among firm characteristics. The previous value and autoregression methods are relatively reliable when the characteristics are not missing in blocks. However, these two method are no longer applicable when previous values are not available under the case of block missing, which is very common in firm characteristics. Therefore, the results of previous value and autoregression are very close to that of cross-sectional median. Without the combination of backward/forward information, the previous value and autoregression methods are better than XS-ridge in both in-sample and OOS MCAR, but are a lot worse when characteristics are missing in blocks. However, TPCA is consistently better than the standard approaches in both out-of-sample performances.

\smallskip

The results comparing TPCA and XS-ridge methods indicate that XS-ridge method is overfitting the characteristics data. XS-ridge method has lower RMSE and higher $R^2$ in-sample than TPCA with the same number of factors, while it falls behind TPCA out-of-sample under both missing patterns. The differences are most prominent with 20 factors and under OOS MCAR, where B-XS-ridge is about $40\%$ higher than B-TPCA in terms of RMSE and about $10\%$ lower in terms of $R^2.$ The differences are smaller under block missingness because of lacking in time series information. However, the difference still demonstrates that tensor factor model is superior to cross-sectional factor model by incorporating the time dimension.

\smallskip

TPCA remains superior to cross-sectional factor model when combined with backward and forward information. Since TPCA more accurately predicts the missing characteristics than XS-ridge, the superiority is preserved after being combined with the same information. In addition, ALS further improves the imputation accuracy by providing a faster convergence rate.

\begin{table}
	\caption{ \label{tab:results} Imputation Results of All Approaches}
	\vskip0.1in
	{\footnotesize This table reports RMSE and $R^2$ of all methods listed in Table \ref{tab:methods} computed on the full sample. The results are reported for the in-sample, out-of-sample missing-completely-at-random (OOS MCAR), and out-of-sample missing in blocks of one year (OOS Block).}
	\vskip0.1in
	\resizebox{\textwidth}{!}{
		\begin{tabular}{llllllllll}
			\hline
			&  &               &               &  &               &              &  &               &               \\ [-3mm]
			&  & \multicolumn{2}{c}{In-Sample} &  & \multicolumn{2}{c}{OOS MCAR} &  & \multicolumn{2}{c}{OOS Block} \\ \hline
			&  &               &               &  &               &              &  &               &               \\ [-3mm]
			Method (\# factors) &  & RMSE          & $R^2$         &  & RMSE          & $R^2$        &  & RMSE          & $R^2$         \\ \cline{1-1} \cline{3-4} \cline{6-7} \cline{9-10} 
			&  &               &               &  &               &              &  &               &               \\ [-3mm]
			
			TPCA (6)   &  & 0.1841        & 0.5471        &  & 0.1856        & 0.5402       &  & 0.1903        & 0.5183        \\
			ALS (6)     &  & 0.1778        & 0.5777        &  & 0.1792        & 0.5714       &  & 0.1848        & 0.5453        \\
			XS-ridge (6)       &  & 0.1657        & 0.6713        &  & 0.1980        & 0.5309      &  & 0.2001       & 0.5232        \\
			&  &               &               &  &               &              &  &               &               \\ [-1mm]
			
			B-TPCA (6) &  & 0.0904        & 0.8908        &  & 0.1044        & 0.8545       &  & 0.1847        & 0.5458        \\
			B-ALS (6)   &  & 0.0900        & 0.8918       &  & 0.1030        & 0.8585       &  & 0.1796        & 0.5708        \\
			B-XS-ridge (6)           &  & 0.0862        & 0.9110        &  & 0.1169        & 0.8363       &  & 0.1950        & 0.5469        \\
			BF-XS-ridge (6)           &  & 0.0789       & 0.9254        &  & 0.1220        & 0.8217       &  & 0.2001        & 0.5232        \\
			&  &               &               &  &               &              &  &               &               \\ [-1mm]
			
			TPCA (20)   &  & 0.1279        & 0.7814        &  & 0.1329        & 0.7644       &  & 0.1474        & 0.7109        \\
			ALS (20)     &  & 0.1267        & 0.7855        &  & 0.1319        & 0.7679       &  & 0.1471        & 0.7123        \\
			XS-ridge (20)       &  & 0.0598        & 0.9572        &  & 0.1613        & 0.6888       &  & 0.1642        & 0.6790        \\
			&  &               &               &  &               &              &  &               &               \\ [-1mm]
			
			B-TPCA (20) &  & 0.0835        & 0.9070        &  & 0.0911        & 0.8893       &  & 0.1441        & 0.7239        \\
			B-ALS (20)   &  & 0.0830        & 0.9081        &  & 0.0906        & 0.8904       &  & 0.1437        & 0.7252        \\
			B-XS-ridge (20)           &  & 0.0317        & 0.9880        &  & 0.1287        & 0.8018       &  & 0.1621        & 0.6870        \\
			BF-XS-ridge (20)           &  & 0.0282        & 0.9905        &  & 0.1257        & 0.8110       &  & 0.1642        & 0.6790        \\
			&  &               &               &  &               &              &  &               &               \\ [-1mm]
			
			Prev. Value         &  & 0.1282        & 0.8031        &  & 0.1522        & 0.7229       &  & 0.2814        & 0.0570        \\
			Median              &  & 0.2890        & 0.0000        &  & 0.2890        & 0.0000       &  & 0.2898        & 0.0000        \\
			AR(1)               &  & 0.1153        & 0.8407        &  & 0.1426        & 0.7566       &  & 0.2810        & 0.0596        \\ [1mm] \hline
	\end{tabular}}
\end{table}

\section{Conclusion}\label{sec:conclusions}
Modern datasets are often multidimensional, extending beyond the $2$-dimensional panel data structures used in traditional factor models and PCA. In this paper, we study a class of $d$-way factor models for high-dimensional tensor data, which are a natural generalization of traditional $2$-way factor models. We demonstrate that $d$-way factor models can be estimated with a variation of PCA, which we call TPCA. This simple algorithm is optimal for the strong factor model. We also consider an improved iterative ALS algorithm which is optimal when the factors are moderately weak.

Additionally, we propose the first formal statistical test for the number of factors in a tensor factor model. Our findings indicate that the tensor factor model offers efficient dimensionality reduction compared to naively pooled traditional factor models. Simultaneously, the model is parsimonious with easily identifiable factors and loadings. These conclusions are supported by extensive simulation results. Lastly, we consider an empirical application to imputing missing firm characteristics.

Interesting applications of TPCA and our results could potentially include more refined panel data models with covariates, e.g., see \cite{freeman2022multidimensional} and \cite{beyhum2020factor} as well as causal inference and imputations with tensor data; see \cite{agarwal2020synthetic}.

\clearpage
\bibliographystyle{plain}
\bibliography{references,demand,granularity_refs,bib_fac_zoo}

\begin{thebibliography}{76}
\providecommand{\natexlab}[1]{#1}
\providecommand{\url}[1]{\texttt{#1}}
\expandafter\ifx\csname urlstyle\endcsname\relax
  \providecommand{\doi}[1]{doi: #1}\else
  \providecommand{\doi}{doi: \begingroup \urlstyle{rm}\Url}\fi

\bibitem[Abadie et~al.(2010)Abadie, Diamond, and Hainmueller]{abadie2010synthetic}
A.~Abadie, A.~Diamond, and J.~Hainmueller.
\newblock Synthetic control methods for comparative case studies: Estimating the effect of california’s tobacco control program.
\newblock \emph{Journal of the American statistical Association}, 105\penalty0 (490):\penalty0 493--505, 2010.

\bibitem[Abadir and Magnus(2005)]{abadir2005matrix}
K.~M. Abadir and J.~R. Magnus.
\newblock Matrix algebra.
\newblock \emph{Cambridge Books}, 2005.

\bibitem[Agarwal et~al.(2020)Agarwal, Shah, and Shen]{agarwal2020synthetic}
A.~Agarwal, D.~Shah, and D.~Shen.
\newblock Synthetic interventions.
\newblock \emph{arXiv preprint arXiv:2006.07691}, 2020.

\bibitem[Andreou et~al.(2019)Andreou, Gagliardini, Ghysels, and Rubin]{andreou2019inference}
E.~Andreou, P.~Gagliardini, E.~Ghysels, and M.~Rubin.
\newblock Inference in group factor models with an application to mixed-frequency data.
\newblock \emph{Econometrica}, 87\penalty0 (4):\penalty0 1267--1305, 2019.

\bibitem[Athey et~al.(2021)Athey, Bayati, Doudchenko, Imbens, and Khosravi]{athey2021matrix}
S.~Athey, M.~Bayati, N.~Doudchenko, G.~Imbens, and K.~Khosravi.
\newblock Matrix completion methods for causal panel data models.
\newblock \emph{Journal of the American Statistical Association}, 116\penalty0 (536):\penalty0 1716--1730, 2021.

\bibitem[Babii et~al.(2022)Babii, Ghysels, and Pan]{babii2022tensor}
A.~Babii, E.~Ghysels, and J.~Pan.
\newblock Tensor principal component analysis.
\newblock \emph{arXiv preprint arXiv:2212.12981}, 2022.

\bibitem[Bai(2003)]{bai2003inferential}
J.~Bai.
\newblock Inferential theory for factor models of large dimensions.
\newblock \emph{Econometrica}, 71\penalty0 (1):\penalty0 135--171, 2003.

\bibitem[Bai(2009)]{bai2009panel}
J.~Bai.
\newblock Panel data models with interactive fixed effects.
\newblock \emph{Econometrica}, 77\penalty0 (4):\penalty0 1229--1279, 2009.

\bibitem[Bai and Ng(2021)]{bai2021matrix}
J.~Bai and S.~Ng.
\newblock Matrix completion, counterfactuals, and factor analysis of missing data.
\newblock \emph{Journal of the American Statistical Association}, 116\penalty0 (536):\penalty0 1746--1763, 2021.

\bibitem[Bai and Ng(2023)]{bai2023approximate}
J.~Bai and S.~Ng.
\newblock Approximate factor models with weaker loadings.
\newblock \emph{Journal of Econometrics}, 235\penalty0 (2):\penalty0 1893--1916, 2023.

\bibitem[Barigozzi et~al.(2023)Barigozzi, He, Li, and Trapani]{barigozzi2022statistical}
M.~Barigozzi, Y.~He, L.~Li, and L.~Trapani.
\newblock Statistical inference for large-dimensional tensor factor model by iterative projections.
\newblock \emph{arXiv preprint arXiv:2206.09800}, 2023.

\bibitem[Beyhum and Gautier(2022{\natexlab{a}})]{beyhum2020factor}
J.~Beyhum and E.~Gautier.
\newblock Factor and factor loading augmented estimators for panel regression with possibly nonstrong factors.
\newblock \emph{Journal of Business \& Economic Statistics}, 41\penalty0 (1):\penalty0 270--281, 2022{\natexlab{a}}.

\bibitem[Beyhum and Gautier(2022{\natexlab{b}})]{beyhum2022factor}
J.~Beyhum and E.~Gautier.
\newblock Factor and factor loading augmented estimators for panel regression with possibly nonstrong factors.
\newblock \emph{Journal of Business \& Economic Statistics}, 41\penalty0 (1):\penalty0 270--281, 2022{\natexlab{b}}.

\bibitem[Billingsley(1995)]{bill1995_book}
P.~Billingsley.
\newblock \emph{Probability and Measure}.
\newblock Wiley, 1995.

\bibitem[Bonhomme and Robin(2010)]{bonhomme2010generalized}
S.~Bonhomme and J.-M. Robin.
\newblock Generalized non-parametric deconvolution with an application to earnings dynamics.
\newblock \emph{The Review of Economic Studies}, 77\penalty0 (2):\penalty0 491--533, 2010.

\bibitem[Bonhomme et~al.(2016)Bonhomme, Jochmans, and Robin]{bonhomme2016estimating}
S.~Bonhomme, K.~Jochmans, and J.-M. Robin.
\newblock Estimating multivariate latent-structure models.
\newblock \emph{The Annals of Statistics}, pages 540--563, 2016.

\bibitem[Bryzgalova et~al.(2024)Bryzgalova, Lerner, Lettau, and Pelger]{bryzgalova2022missing}
S.~Bryzgalova, S.~Lerner, M.~Lettau, and M.~Pelger.
\newblock Missing financial data.
\newblock {\it Review Financial Studies}, (forthcoming), 2024.

\bibitem[Cai and Zhang(2018)]{cai2018rate}
T.~T. Cai and A.~Zhang.
\newblock Rate-optimal perturbation bounds for singular subspaces with applications to high-dimensional statistics.
\newblock \emph{The Annals of Statistics}, 46\penalty0 (1):\penalty0 60--89, 2018.

\bibitem[Candes and Recht(2012)]{candes2012exact}
E.~Candes and B.~Recht.
\newblock Exact matrix completion via convex optimization.
\newblock \emph{Communications of the ACM}, 55\penalty0 (6):\penalty0 111--119, 2012.

\bibitem[Carroll and Chang(1970)]{carroll1970analysis}
J.~D. Carroll and J.-J. Chang.
\newblock {Analysis of individual differences in multidimensional scaling via an N-way generalization of “Eckart-Young” decomposition}.
\newblock \emph{Psychometrika}, 35\penalty0 (3):\penalty0 283--319, 1970.

\bibitem[Chamberlain and Rothschild(1982)]{chamberlain1982arbitrage}
G.~Chamberlain and M.~Rothschild.
\newblock Arbitrage, factor structure, and mean-variance analysis on large asset markets, 1982.

\bibitem[Chang et~al.(2023)Chang, He, Yang, and Yao]{chang2023modelling}
J.~Chang, J.~He, L.~Yang, and Q.~Yao.
\newblock {Modelling matrix time series via a tensor CP-decomposition}.
\newblock \emph{Journal of the Royal Statistical Society Series B: Statistical Methodology}, 85\penalty0 (1):\penalty0 127--148, 2023.

\bibitem[Chen et~al.(2024)Chen, Han, and Yu]{chen2024estimation}
B.~Chen, Y.~Han, and Q.~Yu.
\newblock Estimation and inference for cp tensor factor models.
\newblock \emph{arXiv preprint arXiv:2406.17278}, 2024.

\bibitem[Chen et~al.(2022)Chen, Yang, and Zhang]{chen2022factor}
R.~Chen, D.~Yang, and C.-H. Zhang.
\newblock Factor models for high-dimensional tensor time series.
\newblock \emph{Journal of the American Statistical Association}, 117\penalty0 (537):\penalty0 94--116, 2022.

\bibitem[Cunha et~al.(2010)Cunha, Heckman, and Schennach]{cunha2010estimating}
F.~Cunha, J.~J. Heckman, and S.~M. Schennach.
\newblock Estimating the technology of cognitive and noncognitive skill formation.
\newblock \emph{Econometrica}, 78\penalty0 (3):\penalty0 883--931, 2010.

\bibitem[De~Lathauwer et~al.(2000{\natexlab{a}})De~Lathauwer, De~Moor, and Vandewalle]{de2000best}
L.~De~Lathauwer, B.~De~Moor, and J.~Vandewalle.
\newblock On the best rank-1 and rank-(r 1, r 2,..., rn) approximation of higher-order tensors.
\newblock \emph{SIAM journal on Matrix Analysis and Applications}, 21\penalty0 (4):\penalty0 1324--1342, 2000{\natexlab{a}}.

\bibitem[De~Lathauwer et~al.(2000{\natexlab{b}})De~Lathauwer, De~Moor, and Vandewalle]{de2000multilinear}
L.~De~Lathauwer, B.~De~Moor, and J.~Vandewalle.
\newblock A multilinear singular value decomposition.
\newblock \emph{SIAM journal on Matrix Analysis and Applications}, 21\penalty0 (4):\penalty0 1253--1278, 2000{\natexlab{b}}.

\bibitem[El~Karoui(2003)]{karoui2003largest}
N.~El~Karoui.
\newblock {On the largest eigenvalue of Wishart matrices with identity covariance when n, p and p/n tend to infinity}.
\newblock \emph{arXiv preprint math/0309355}, 2003.

\bibitem[Foerster et~al.(2011)Foerster, Sarte, and Watson]{foerster2011sectoral}
A.~T. Foerster, P.-D.~G. Sarte, and M.~W. Watson.
\newblock Sectoral versus aggregate shocks: A structural factor analysis of industrial production.
\newblock \emph{Journal of Political Economy}, 119\penalty0 (1):\penalty0 1--38, 2011.

\bibitem[Freeman(2022)]{freeman2022multidimensional}
H.~Freeman.
\newblock Multidimensional interactive fixed-effects.
\newblock \emph{arXiv preprint arXiv:2209.11691}, 2022.

\bibitem[Freeman and Weidner(2023)]{freeman2023linear}
H.~Freeman and M.~Weidner.
\newblock Linear panel regressions with two-way unobserved heterogeneity.
\newblock \emph{Journal of Econometrics}, 237\penalty0 (1):\penalty0 105498, 2023.

\bibitem[Freyberger et~al.(2020)Freyberger, Neuhierl, and Weber]{freyberger2020dissecting}
J.~Freyberger, A.~Neuhierl, and M.~Weber.
\newblock Dissecting characteristics nonparametrically.
\newblock \emph{The Review of Financial Studies}, 33\penalty0 (5):\penalty0 2326--2377, 2020.

\bibitem[Gobillon and Magnac(2016)]{gobillon2016regional}
L.~Gobillon and T.~Magnac.
\newblock Regional policy evaluation: Interactive fixed effects and synthetic controls.
\newblock \emph{Review of Economics and Statistics}, 98\penalty0 (3):\penalty0 535--551, 2016.

\bibitem[Graham(2020)]{graham2020network}
B.~S. Graham.
\newblock Network data.
\newblock In \emph{Handbook of Econometrics}, volume~7, pages 111--218. Elsevier, 2020.

\bibitem[Han et~al.(2022{\natexlab{a}})Han, Luo, Wang, and Zhang]{han2022exact}
R.~Han, Y.~Luo, M.~Wang, and A.~R. Zhang.
\newblock Exact clustering in tensor block model: Statistical optimality and computational limit.
\newblock \emph{Journal of the Royal Statistical Society Series B: Statistical Methodology}, 84\penalty0 (5):\penalty0 1666--1698, 2022{\natexlab{a}}.

\bibitem[Han et~al.(2022{\natexlab{b}})Han, Chen, and Zhang]{han2022rank}
Y.~Han, R.~Chen, and C.-H. Zhang.
\newblock Rank determination in tensor factor model.
\newblock \emph{Electronic Journal of Statistics}, 16\penalty0 (1):\penalty0 1726--1803, 2022{\natexlab{b}}.

\bibitem[Han et~al.(2025)Han, Chen, Yang, and Zhang]{han2020tensor}
Y.~Han, R.~Chen, D.~Yang, and C.-H. Zhang.
\newblock Tensor factor model estimation by iterative projection.
\newblock \emph{Annals of Statistics (forthcoming)}, 2025.

\bibitem[Harshman(1970)]{harshman1970foundations}
R.~A. Harshman.
\newblock {Foundations of the PARAFAC procedure: Models and conditions for an ``explanatory'' multimodal factor analysis}.
\newblock \emph{UCLA Working Papers in Phonetics}, 16:\penalty0 1--84, 1970.

\bibitem[Hitchcock(1927)]{hitchcock1927expression}
F.~L. Hitchcock.
\newblock The expression of a tensor or a polyadic as a sum of products.
\newblock \emph{Journal of Mathematics and Physics}, 6\penalty0 (1-4):\penalty0 164--189, 1927.

\bibitem[Horn and Johnson(2013)]{horn_joh_2013}
R.~A. Horn and C.~R. Johnson.
\newblock \emph{Matrix Analysis}.
\newblock Cambridge University Press, 2013.

\bibitem[Jolliffe(2002)]{jolliffe2002principal}
I.~T. Jolliffe.
\newblock \emph{Principal component analysis for special types of data}.
\newblock Springer, 2002.

\bibitem[Kato(1996)]{kato1996perturbation}
T.~Kato.
\newblock \emph{Perturbation theory for linear operators}, volume 132.
\newblock Springer Science \& Business Media, 1996.

\bibitem[Kneip and Utikal(2001)]{kneip2001inference}
A.~Kneip and K.~J. Utikal.
\newblock Inference for density families using functional principal component analysis.
\newblock \emph{Journal of the American Statistical Association}, 96\penalty0 (454):\penalty0 519--542, 2001.

\bibitem[Kolda(2006)]{kolda2006multilinear}
T.~G. Kolda.
\newblock Multilinear operators for higher-order decompositions.
\newblock Technical report, Sandia National Laboratories (SNL), Albuquerque, NM, and Livermore, CA, 2006.

\bibitem[Koltchinskii and Lounici(2017)]{koltchinskii2017normal}
V.~Koltchinskii and K.~Lounici.
\newblock Normal approximation and concentration of spectral projectors of sample covariance.
\newblock \emph{The Annals of Statistics}, 45\penalty0 (1):\penalty0 121--157, 2017.

\bibitem[Kroonenberg and De~Leeuw(1980)]{kroonenberg1980principal}
P.~M. Kroonenberg and J.~De~Leeuw.
\newblock Principal component analysis of three-mode data by means of alternating least squares algorithms.
\newblock \emph{Psychometrika}, 45:\penalty0 69--97, 1980.

\bibitem[Lata{\l}a(2005)]{latala2005some}
R.~Lata{\l}a.
\newblock Some estimates of norms of random matrices.
\newblock \emph{Proceedings of the American Mathematical Society}, 133\penalty0 (5):\penalty0 1273--1282, 2005.

\bibitem[Lei and Ross(2023)]{lei2023estimating}
L.~Lei and B.~Ross.
\newblock Estimating counterfactual matrix means with short panel data.
\newblock \emph{arXiv preprint arXiv:2312.07520}, 2023.

\bibitem[Lettau(2022)]{lettau2022estimating}
M.~Lettau.
\newblock High dimensional factor models with an application to mutual fund characteristics.
\newblock \emph{Working Paper}, 2022.

\bibitem[Lettau(2023)]{lettau2023high}
M.~Lettau.
\newblock High-dimensional factor models and the factor zoo.
\newblock Technical report, National Bureau of Economic Research, 2023.

\bibitem[Lewbel(1991)]{lewbel1991rank}
A.~Lewbel.
\newblock The rank of demand systems: theory and nonparametric estimation.
\newblock \emph{Econometrica: Journal of the Econometric Society}, pages 711--730, 1991.

\bibitem[Liu and Tsyvinski(2024)]{liu2024dynamic}
E.~Liu and A.~Tsyvinski.
\newblock A dynamic model of input-output networks.
\newblock \emph{Review of Economic Studies}, page rdae012, 2024.

\bibitem[Onatski(2006)]{onatski2006formal}
A.~Onatski.
\newblock A formal statistical test for the number of factors in the approximate factor models.
\newblock \emph{Columbia University. Working Paper}, 2006.

\bibitem[Onatski(2009)]{onatski_ecma_2009}
A.~Onatski.
\newblock Testing hypotheses about the number of factors in large factor models.
\newblock \emph{Econometrica}, 77:\penalty0 1447--1479, 2009.

\bibitem[Onatski(2012)]{onatski2012asymptotics}
A.~Onatski.
\newblock Asymptotics of the principal components estimator of large factor models with weakly influential factors.
\newblock \emph{Journal of Econometrics}, 168\penalty0 (2):\penalty0 244--258, 2012.

\bibitem[Onatski(2022)]{onatski2022uniform}
A.~Onatski.
\newblock Uniform asymptotics for weak and strong factors.
\newblock \emph{University of Cambridge Working Paper}, 2022.

\bibitem[Pearson(1901)]{pearson1901liii}
K.~Pearson.
\newblock On lines and planes of closest fit to systems of points in space.
\newblock \emph{The London, Edinburgh, and Dublin philosophical magazine and journal of science}, 2\penalty0 (11):\penalty0 559--572, 1901.

\bibitem[Pesaran(2006)]{pesaran2006estimation}
M.~H. Pesaran.
\newblock Estimation and inference in large heterogeneous panels with a multifactor error structure.
\newblock \emph{Econometrica}, 74\penalty0 (4):\penalty0 967--1012, 2006.

\bibitem[Petrov(1995)]{petrov1995limit}
V.~V. Petrov.
\newblock Limit theorems of probability theory; sequences of independent random variables.
\newblock Technical report, 1995.

\bibitem[Ricci and Levi-Civita(1900)]{ricci1900methodes}
M.~Ricci and T.~Levi-Civita.
\newblock M{\'e}thodes de calcul diff{\'e}rentiel absolu et leurs applications.
\newblock \emph{Mathematische Annalen}, 54\penalty0 (1):\penalty0 125--201, 1900.

\bibitem[Richard and Montanari(2014)]{richard2014statistical}
E.~Richard and A.~Montanari.
\newblock A statistical model for tensor pca.
\newblock \emph{Advances in neural information processing systems}, 27, 2014.

\bibitem[Ross(1976)]{ross1976arbitrage}
S.~A. Ross.
\newblock The arbitrage theory of capital asset pricing.
\newblock \emph{Journal of Economic Theory}, 13:\penalty0 341--360, 1976.

\bibitem[Sargent et~al.(1977)Sargent, Sims, et~al.]{sargent1977business}
T.~J. Sargent, C.~A. Sims, et~al.
\newblock Business cycle modeling without pretending to have too much a priori economic theory.
\newblock \emph{New methods in business cycle research}, 1:\penalty0 145--168, 1977.

\bibitem[Soshnikov(2002)]{soshnikov2002note}
A.~Soshnikov.
\newblock A note on universality of the distribution of the largest eigenvalues in certain sample covariance matrices.
\newblock \emph{Journal of Statistical Physics}, 108:\penalty0 1033--1056, 2002.

\bibitem[Spearman(1904)]{spearmen1904general}
C.~Spearman.
\newblock General intelligence objectively determined and measured.
\newblock \emph{American Journal of Psychology}, 15:\penalty0 107--197, 1904.

\bibitem[Stock and Watson(2002{\natexlab{a}})]{stock2002}
J.~Stock and M.~Watson.
\newblock Forecasting using principal components from a large number of predictors.
\newblock \emph{Journal of the American Statistical Association}, 97:\penalty0 1167--1179, 2002{\natexlab{a}}.

\bibitem[Stock and Watson(2002{\natexlab{b}})]{stock2002forecasting}
J.~H. Stock and M.~W. Watson.
\newblock Forecasting using principal components from a large number of predictors.
\newblock \emph{Journal of the American statistical association}, 97\penalty0 (460):\penalty0 1167--1179, 2002{\natexlab{b}}.

\bibitem[Tucker(1966)]{tucker1966some}
L.~R. Tucker.
\newblock Some mathematical notes on three-mode factor analysis.
\newblock \emph{Psychometrika}, 31\penalty0 (3):\penalty0 279--311, 1966.

\bibitem[Vershynin(2018)]{vershynin2018high}
R.~Vershynin.
\newblock \emph{High-dimensional probability: An introduction with applications in data science}, volume~47.
\newblock Cambridge University Press, 2018.

\bibitem[Wang et~al.(2022)Wang, Zheng, Lian, and Li]{wang2022high}
D.~Wang, Y.~Zheng, H.~Lian, and G.~Li.
\newblock High-dimensional vector autoregressive time series modeling via tensor decomposition.
\newblock \emph{Journal of the American Statistical Association}, 117\penalty0 (539):\penalty0 1338--1356, 2022.

\bibitem[Watson(1983)]{watson1983statistics}
G.~S. Watson.
\newblock \emph{Statistics on spheres}.
\newblock Wiley-Interscience, 1983.

\bibitem[Wedin(1972)]{wedin1972perturbation}
P.-{\AA}. Wedin.
\newblock Perturbation bounds in connection with singular value decomposition.
\newblock \emph{BIT Numerical Mathematics}, 12:\penalty0 99--111, 1972.

\bibitem[Xiong and Pelger(2022)]{xiong2022large}
R.~Xiong and M.~Pelger.
\newblock Large dimensional latent factor modeling with missing observations and applications to causal inference.
\newblock {\it Journal of Econometrics} (forthcoming), 2022.

\bibitem[Xiong and Pelger(2023)]{xiong2023large}
R.~Xiong and M.~Pelger.
\newblock Large dimensional latent factor modeling with missing observations and applications to causal inference.
\newblock \emph{Journal of Econometrics}, 233\penalty0 (1):\penalty0 271--301, 2023.

\bibitem[Zhang and Xia(2018)]{zhang2018tensor}
A.~Zhang and D.~Xia.
\newblock {Tensor SVD: Statistical and computational limits}.
\newblock \emph{IEEE Transactions on Information Theory}, 64\penalty0 (11):\penalty0 7311--7338, 2018.

\bibitem[Zhang et~al.(2022)Zhang, Cai, and Wu]{zhang2022heteroskedastic}
A.~R. Zhang, T.~T. Cai, and Y.~Wu.
\newblock {Heteroskedastic PCA: Algorithm, optimality, and applications}.
\newblock \emph{The Annals of Statistics}, 50\penalty0 (1):\penalty0 53--80, 2022.

\end{thebibliography}

\clearpage

\appendix

\bigskip

\begin{center}
	{\large \bf APPENDIX}
\end{center}

\bigskip

\setcounter{page}{1}
\setcounter{section}{0}
\setcounter{equation}{0}
\setcounter{table}{0}
\setcounter{figure}{0}
\renewcommand{\theequation}{A.\arabic{equation}}
\renewcommand\thetable{A.\arabic{table}}
\renewcommand\thefigure{A.\arabic{figure}}
\renewcommand\thesection{A.\arabic{section}}
\renewcommand\thesubsection{A.\arabic{section}.\arabic{subsection}}
\renewcommand\thepage{Appendix - \arabic{page}}

\section{Tensor Matricization: An Illustrative Example \label{appsec:matricization}}
Fixing all but two indices of a tensor, we obtain a matrix, called slice; see Figure~\ref{appfig:slices} for a graphical illustration. 

\begin{figure}[h]
	\centering
	\caption{Lateral, horizontal, and frontal slices of a $4\times 5\times 3$ tensor}
	\label{appfig:slices} 
	\vskip0.1in
	\includegraphics[width=\textwidth]{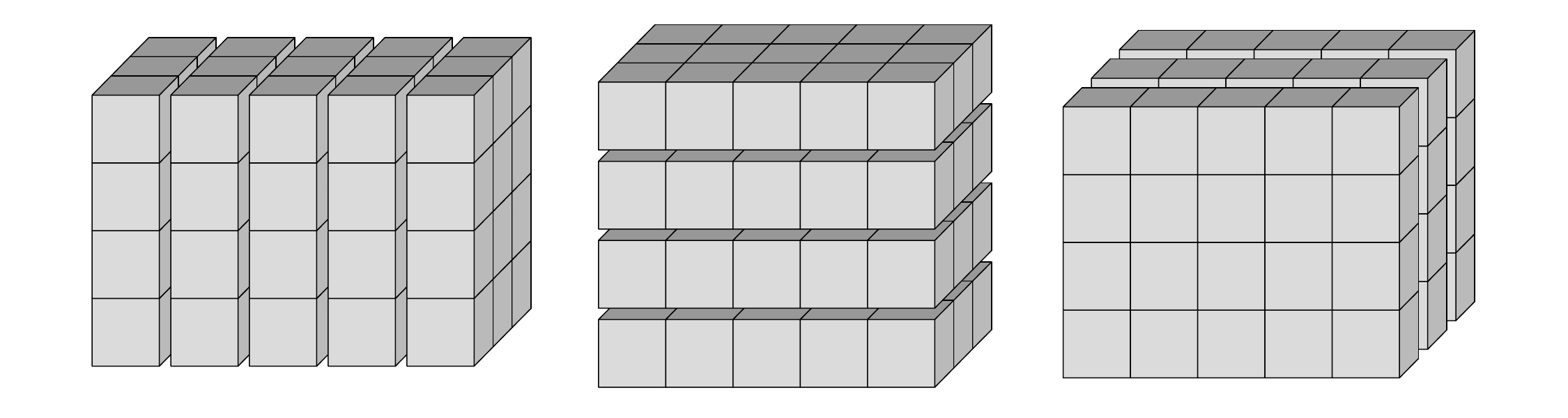}
\end{figure}

\smallskip
\textbf{Example 1}:\\

Let $\mathbf{Y}$ be a $3\times 4\times 2$ dimensional tensor of the following two frontal slices:

\begin{equation*}
	\mathbf{Y}_1=\begin{bmatrix}
	1 & 4 & 7 & 10\\
	2 & 5 & 8 & 11\\
	3 & 6 & 9 & 12
	\end{bmatrix} \qquad
	\mathbf{Y}_2=\begin{bmatrix}
	13 & 16 & 19 & 22\\
	14 & 17 & 20 & 23\\
	15 & 18 & 21 & 24
	\end{bmatrix}.
\end{equation*}

Then the mode-$1$, $2$ and $3$ matricizations of $\mathbf{Y}$ are respectively:
\begin{equation*}
	\mathbf{Y}_{(1)} = \begin{bmatrix}
	1 & 4 & 7 & 10 & 13 & 16 & 19 & 22\\
	2 & 5 & 8 & 11 & 14 & 17 & 20 & 23\\
	3 & 6 & 9 & 12 & 15 & 18 & 21 & 24
	\end{bmatrix}
\end{equation*}
\begin{equation*}
	\mathbf{Y}_{(2)} = \begin{bmatrix}
	1 & 2 & 3 & 13 & 14 & 15\\
	4 & 5 & 6 & 16 & 17 & 18\\
	7 & 8 & 9 & 19 & 20 & 21\\
	10 & 11 & 12 & 22 & 23 & 24
	\end{bmatrix}
\end{equation*}
\begin{equation*}
	\mathbf{Y}_{(3)} = \begin{bmatrix}
	1 & 2 & 3 & 4 & \cdots & 9 & 10 & 11 & 12\\
	13 & 14 & 15 & 16 & \cdots & 21 & 22 & 23 & 24
	\end{bmatrix}
\end{equation*}

\smallskip
\textbf{Example 2}:\\

Let $\mathbf{Y}$ be a $3\times 3\times 3$ dimensional tensor of the following three frontal slices:

\begin{equation*}
	\mathbf{Y}_1=\begin{bmatrix}
	1 & 4 & 7 \\
	2 & 5 & 8 \\
	3 & 6 & 9 
	\end{bmatrix} \quad
	\mathbf{Y}_2=\begin{bmatrix}
	10 & 13 & 16\\
	11 & 14 & 17\\
	12 & 15 & 18\\
	\end{bmatrix}.
	\mathbf{Y}_3=\begin{bmatrix}
	19 & 22 & 25\\
	20 & 23 & 26\\
	21 & 24 & 27
	\end{bmatrix}.
\end{equation*}

Then the mode-$1$, $2$ and $3$ matricization of $\mathbf{Y}$ are respectively:
\begin{equation*}
	\mathbf{Y}_{(1)} = \begin{bmatrix}
	1 & 4 & 7 & 10 & 13 & 16 & 19 & 22 & 25\\
	2 & 5 & 8 & 11 & 14 & 17 & 20 & 23 & 26\\
	3 & 6 & 9 & 12 & 15 & 18 & 21 & 24 & 27
	\end{bmatrix}
\end{equation*}
\begin{equation*}
	\mathbf{Y}_{(2)} = \begin{bmatrix}
	1 & 2 & 3 & 10 & 11 & 12 & 19 & 20 & 21\\
	4 & 5 & 6 & 13 & 14 & 15 & 22 & 23 & 24\\
	7 & 8 & 9 & 16 & 17 & 18 & 25 & 26 & 27
	\end{bmatrix}
\end{equation*}
\begin{equation*}
	\mathbf{Y}_{(3)} = \begin{bmatrix}
	1 & 2 & 3 & 4 & 5 & 6 & 7 & 8 & 9\\
	10 & 11 & 12 & 13 & 14 & 15 & 16 & 17 & 18\\
	19 & 20 & 21 & 22 & 23 & 24 & 25 & 26 & 27
	\end{bmatrix}
\end{equation*}

\section{Imputation Methods for Missing Firm Characteristics}\label{appsec:imputation}

\cite{bryzgalova2022missing} model the panel of firm characteristics for each month $t$ as follows:
\begin{equation}
	\label{eq:TSpanels}
	\mathbf{C}_{i,\ell}^t = \mathbf{F}_i^t \mathbf{\Lambda}_\ell^{t \top} + \mathbf{U}_{i,\ell}^t \quad \mathrm{with}\; i=1,\ldots,N_t \; \mathrm{and} \; \ell = 1,\ldots,L.
\end{equation}
The above equation makes clear that we de facto have a tensor data set, along three dimensions $i,$ $\ell$ and $t.$ Hence, firm characteristics are a 3-way or third-order tensor with its three dimensions: (a) firms, (b) characteristics, and (c) time.

\subsection{Tensor models for characteristics}
For firm characteristics $\mathbf{C}\in\R^{N\times T\times L}$ with possibly missing entries, we instead consider the tensor factor model 
\begin{equation}\label{eqn:model}
	\mathbf{C} = \mathbf{G}\times_1 \mathbf{M}\times_2\mathbf{F}\times_3\mathbf{\Lambda} + \mathbf{U} \equiv \llbracket \mathbf{G}; \mathbf{M}, \mathbf{F}, \mathbf{\Lambda} \rrbracket + \mathbf{U},
\end{equation}
where $\mu_r\in \R^N$ are the loadings corresponding to the firms dimension, $\lambda_r\in \R^L$ are the loadings corresponding to the characteristics dimension, and $f_r\in \R^T$ are the factors. We collect the $\lambda_r,$ $\mu_r,$ and $f_r$ in matrices $\mathbf{\Lambda} \in \R^{I \times R},$ $\mathbf{M} \in \R^{L \times R},$ and $\mathbf{F} \in \R^{T \times R}.$ The loadings and factors are of unit norm, and $\mathbf{G}$ is the core tensor. 

Before discussing our estimator, let us first review the approach of \cite{bryzgalova2022missing} to
estimate equation (\ref{eq:TSpanels}). They follow \cite{xiong2022large}, and compute the loadings $\mathbf{\Lambda}^t$ as eigenvectors of the estimated characteristic covariance matrix
\begin{equation}
	\hat{\mathbf{\Sigma}}^{\mathrm{XS},t}_{\ell,p} = \frac{1}{Q_{\ell,p}^t} \sum_{i\in Q_{\ell,p}^t} \mathbf{C}^{t}_{\ell,i} \mathbf{C}^{t}_{p,i},
\end{equation}
where $Q_{\ell,p}^t$ is the set of all stocks that are observed for the two characteristics $\ell$ and $p$ at time $t.$ The superscript $XS,t$ refers to the time series of cross-sections approach. For each $t$, the cross-sectional factors are estimated by a regression on the estimated loadings $\hat{\mathbf{\Lambda}}^t$:
\begin{equation}
	\label{eq:factorbryz}
	\hat{\mathbf{F}}_i^{t} = \left( \frac{1}{L} \sum_{\ell = 1}^{L} W_{i,\ell}^t \hat{\mathbf{\Lambda}}^t_\ell (\hat{\mathbf{\Lambda}}^t_\ell)^\top \right)^{-1} \left( \frac{1}{L}\sum_{\ell=1}^{L} W_{i,\ell}^t \hat{\mathbf{\Lambda}}^t_\ell (\mathbf{C}^{t}_{\ell,i})^\top \right),
\end{equation}
where $W_{i,\ell}^t = 1$ if characteristic $\ell$ is observed for stock $i$ at time $t$ and $W_{i,\ell}^t = 0$ otherwise. 

We now turn our attention to the estimation of the tensor factor model (\ref{eqn:model}). Let $\mathbf{C}^{(j)}$ be the matricization of firm characteristics along the dimension $j\in \{1, 2, 3\}$, and $\mathbf{C}^{(j)}$ is of size $N_j\times \prod_{k\ne j}N_k$. The factor matrices $\mathbf{M}, \mathbf{F}, \mathbf{\Lambda}$ can be estimated by applying PCA to the covariance matrix of the matricized tensor. Due to missingness in the firm characteristics data, we use the approach of \cite{xiong2023large} and calculate the covariance matrix based on partially observed data, denoted by $\hat \Sigma^{(j)}$, as
\begin{equation*}
	\hat{\mathbf{\Sigma}}^{(j)}_{\ell,p} = \frac{1}{Q_{\ell,p}} \sum_{i\in Q_{\ell,p}} \mathbf{C}^{(j)}_{\ell,i} \mathbf{C}^{(j)}_{p,i},
\end{equation*}
where $Q_{\ell,p}$ is the set of all columns that are observed for the $\ell^{\rm th}$ and the $p^{\rm th}$ rows of $\mathbf{C}^{(j)}.$\footnote{\cite{xiong2023large} provides inferential theory for estimating latent factor models from panel data with missing observations based on PCA. Matricized tensors are essentially panel data.} Then the core tensor $\mathbf{G}$ can be estimated by step 5 of the algorithm \ref{alg:hosvd} where the missing values in $\mathbf{Y}$ are filled with $0s$ (cross-sectional median). The ALS algorithm \ref{alg:hooi} can also be used to estimate the same model \ref{eqn:model} in the same fashion.

\smallskip

To highlight the difference between the time series of cross-sections and our tensor approach, we can write the former estimator as follows:
\begin{equation}
	\label{eq:CwithXSt}
	\hat{\mathbf{C}}_{i,\ell}^{\textrm{XS},t} = \sum_{r=1}^R\hat{\sigma}_{r,t} \hat{\mu}_{r,i,t}  \hat{\lambda}_{r,\ell,t},
\end{equation}
where all parameters $\hat{\sigma}_{r,t},$ $\hat{\mu}_{r,i,t},$ and $\hat{\lambda}_{r,\ell,t}$ have subscript $t$ because the estimation is done in each period in time. 
In contrast, the missing characteristics are predicted by TPCA as
\begin{equation}
	\label{eq:CwithTPCA}
	\hat{\mathbf{C}}^{\textrm{TPCA}} = \hat{\mathbf{G}} \times_1 \hat{\mathbf{M}} \times_2 \hat{\mathbf{F}} \times_3 \hat{\mathbf{\Lambda}}.
\end{equation}
Comparing equation (\ref{eq:CwithXSt}) with (\ref{eq:CwithTPCA}) we clear see that there is a considerable amount of parameter proliferation when estimating time series of cross-sections. Therefore, we would expect a better performance from the parsimonious tensor-based approach.

\subsection{Backward and Forward Information}

Firm characteristics have strong dependency with their previous values, and filling the missing values with its previous value is one of the most widely used methods for missing imputations. Inspired by \cite{bryzgalova2022missing}, we combine the tensor factor model with backward information by a cross-sectional regression. We define the covariates in the regression by
\begin{equation}
	\mathbf{X}_{i,t,\ell}^{{\textrm{B-TPCA}}} = \left( \sum_{r=1}^R\hat{\sigma}_r \hat{\mu}_{r,i} \hat{f}_{r,t} \hat{\lambda}_{r,\ell} \quad \mathbf{C}_{i,\ell}^{t-1} \quad \hat{\mathbf{e}}_{i,\ell}^{t-1} \right),
\end{equation}
where $\hat{\mathbf{e}}_{i,\ell}^{t} = \mathbf{C}_{i,t,\ell} - \hat{\mathbf{C}}_{i,t,\ell}^{\textrm{TPCA}}$ is the residual from TPCA, and we estimate the model
\begin{equation}
	\mathbf{C}_{i,t,\ell} = \left( \beta^{\ell,t,\textrm{B-TPCA}}\right)^\top \left( \sum_{r=1}^R\hat{\sigma}_r \hat{\mu}_{r,i} \hat{f}_{r,t} \hat{\lambda}_{r,\ell} \quad \mathbf{C}_{i,\ell}^{t-1} \quad \hat{\mathbf{e}}_{i,\ell}^{t-1} \right)
\end{equation}
with a cross-sectional regression on the partially observed data at time $t$ for characteristic $\ell$ as follows:
\begin{equation*}
	\hat{\beta}^{\ell,t,\textrm{B-TPCA}} = \left( \sum_{i=1}^{N_t} \mathbf{W}_i^{t,\ell} \mathbf{X}_{i,t,\ell}^{{\textrm{B-TPCA}}} (\mathbf{X}_{i,t,\ell}^{{\textrm{B-TPCA}}})^\top \right) \left( \sum_{i=1}^{N_t}\mathbf{W}_i^{t,\ell} \mathbf{X}_{i,t,\ell}^{{\textrm{B-TPCA}}} \mathbf{C}_{i,t,\ell} \right),
\end{equation*}
where $\mathbf{W}_i^{t,\ell} = 1$ if $\mathbf{X}_{i,t,\ell}^{{\textrm{B-TPCA}}}$ and $C_{i,t,\ell}$ are both observed and $0$ otherwise. The missing characteristics are predicted as
\begin{equation*}
	\hat{\mathbf{C}}_{i,t,\ell}^{\textrm{B-TPCA}} = (\hat{\beta}^{\ell,t,\textrm{B-TPCA}})^\top \mathbf{X}_{i,t,\ell}^{{\textrm{B-TPCA}}}.
\end{equation*}

\begin{table}
	\centering
	\caption{Different Imputation Methods}\label{tab:methods}
	{\footnotesize
		\begin{tabular}{lll}
			\hline
			&                       &              \\ [-3mm]
			Method & \multicolumn{1}{l|}{} &  Estimation \\
			&                       &              \\ [-3mm] \hline
			&                       &              \\ [-4mm] 
			Tucker (TPCA/ALS) & \multicolumn{1}{l|}{} & $\hat{\mathbf{C}}^{\textrm{TPCA}} = \hat{\mathbf{G}} \times_1 \hat{\mathbf{M}} \times_2 \hat{\mathbf{F}} \times_3 \hat{\mathbf{\Lambda}}$ \\ [2mm]
			Backward-Forward-TPCA (BF-TPCA) &     \multicolumn{1}{l|}{}       &     $ \hat{\mathbf{C}}_{i,t,\ell}^{\textrm{BF-TPCA}} = (\hat{\beta}^{\ell,t,\textrm{BF-TPCA}})^\top \left( \hat{\mathbf{C}}_{i,t,\ell}^{\textrm{TPCA}} \quad \mathbf{C}_{i,\ell}^{t-1} \quad \hat{\mathbf{e}}_{i,\ell}^{t-1} \quad \mathbf{C}_{i,\ell}^{t+1} \quad \hat{\mathbf{e}}_{i,\ell}^{t+1} \right)$        \\ [2mm]
			Backward-TPCA (B-TPCA) & \multicolumn{1}{l|}{} & $ \hat{\mathbf{C}}_{i,t,\ell}^{\textrm{B-TPCA}} = (\hat{\beta}^{\ell,t,\textrm{B-TPCA}})^\top \left( \hat{\mathbf{C}}_{i,t,\ell}^{\textrm{TPCA}} \quad \mathbf{C}_{i,\ell}^{t-1} \quad \hat{\mathbf{e}}_{i,\ell}^{t-1} \right)$ \\ [2mm]
			Cross-sectional (XS) & \multicolumn{1}{l|}{} & $\hat{\mathbf{C}}_{i,t,\ell}^{\textrm{XS}} = \sum_{r=1}^R\hat{\sigma}_{r,t} \hat{\mu}_{r,i,t}  \hat{\lambda}_{r,\ell,t}$ \\ [2mm]
			Backward-Forward-XS (BF-XS) & \multicolumn{1}{l|}{}  & $ \hat{\mathbf{C}}_{i,t,\ell}^{\textrm{BF-XS}} = (\hat{\beta}^{\ell,t,\textrm{BF-XS}})^\top \left( \hat{\mathbf{C}}_{i,t,\ell}^{\textrm{XS}} \quad \mathbf{C}_{i,\ell}^{t-1} \quad \hat{\mathbf{e}}_{i,\ell}^{t-1} \quad \mathbf{C}_{i,\ell}^{t+1} \quad \hat{\mathbf{e}}_{i,\ell}^{t+1} \right)$ \\ [2mm]
			Backward-XS (B-XS) & \multicolumn{1}{l|}{}  & $ \hat{\mathbf{C}}_{i,t,\ell}^{\textrm{B-XS}} = (\hat{\beta}^{\ell,t,\textrm{B-XS}})^\top \left( \hat{\mathbf{C}}_{i,t,\ell}^{\textrm{XS}} \quad \mathbf{C}_{i,\ell}^{t-1} \quad \hat{\mathbf{e}}_{i,\ell}^{t-1} \right)$ \\ [2mm]
			Autoregression (AR) & \multicolumn{1}{l|}{} & $ \hat{\mathbf{C}}_{i,t,\ell}^{\textrm{AR}} = (\hat{\beta}^{\ell,t,\textrm{AR}})^\top  \mathbf{C}_{i,\ell}^{t-1} $ \\ [2mm]
			Previous Value (PV) & \multicolumn{1}{l|}{} & $\hat{\mathbf{C}}_{i,t,\ell}^{\textrm{PV}} = \mathbf{C}_{i,\ell}^{t-1}$ \\ [2mm]
			Cross-sectional median (median) & \multicolumn{1}{l|}{} & $\hat{\mathbf{C}}_{i,t,\ell}^{\textrm{median}} = 0$ \\
			&                       &              \\ [-4mm] \hline
	\end{tabular}}
\end{table}

The B-TPCA model can be extended to include forward information at time $t+1$, where we define the covariates
\begin{equation}
	\mathbf{X}_{i,t,\ell}^{{\textrm{BF-TPCA}}} = \left( \sum_{r=1}^R\hat{\sigma}_r \hat{\mu}_{r,i} \hat{f}_{r,t} \hat{\lambda}_{r,\ell} \quad \mathbf{C}_{i,\ell}^{t-1} \quad \hat{\mathbf{e}}_{i,\ell}^{t-1} \quad \mathbf{C}_{i,\ell}^{t+1} \quad \hat{\mathbf{e}}_{i,\ell}^{t+1} \right),
\end{equation}
and estimate the model
\begin{equation}
	\mathbf{C}_{i,t,\ell} = \left( \beta^{\ell,t,\textrm{BF-TPCA}}\right)^\top \left( \sum_{r=1}^R\hat{\sigma}_r \hat{\mu}_{r,i} \hat{f}_{r,t} \hat{\lambda}_{r,\ell} \quad \mathbf{C}_{i,\ell}^{t-1} \quad \hat{\mathbf{e}}_{i,\ell}^{t-1} \quad \mathbf{C}_{i,\ell}^{t+1} \quad \hat{\mathbf{e}}_{i,\ell}^{t+1} \right).
\end{equation}
The model can also be estimated with a cross-sectional regression on the partially observed data at time $t$ for characteristic $\ell$ as follows:
\begin{equation}
	\hat{\beta}^{\ell,t,\textrm{BF-TPCA}} = \left( \sum_{i=1}^{N_t} \mathbf{W}_i^{t,\ell} \mathbf{X}_{i,t,\ell}^{{\textrm{BF-TPCA}}} (\mathbf{X}_{i,t,\ell}^{{\textrm{BF-TPCA}}})^\top \right) \left( \sum_{i=1}^{N_t}\mathbf{W}_i^{t,\ell} \mathbf{X}_{i,t,\ell}^{{\textrm{BF-TPCA}}} \mathbf{C}_{i,t,\ell} \right),
\end{equation}
where $\mathbf{W}_i^{t,\ell} = 1$ if $\mathbf{X}_{i,t,l}^{{\textrm{BF-TPCA}}}$ and $C_{i,t,l}$ are both observed and $0$ otherwise. The missing characteristics are predicted as
\begin{equation*}
	\hat{\mathbf{C}}_{i,t,\ell}^{\textrm{BF-TPCA}} = (\hat{\beta}^{\ell,t,\textrm{BF-TPCA}})^\top \mathbf{X}_{i,t,\ell}^{{\textrm{BF-TPCA}}}.
\end{equation*}

\subsection{Alternative Imputation Methods}

\cite{bryzgalova2022missing} consider various imputation methods for firm characteristics including pure cross-sectional factor model (XS), cross-sectional with backward information (B-XS), cross-sectional with backward and forward information (BF-XS). Essentially, these methods replace the TPCA estimation with XS estimation in the models discussed previously. Due to the fact that cross-sectional models fit each time period separately, these methods can easily result in overfitting. For this reason \cite{bryzgalova2022missing} consider a regularized version of the cross-sectional factor model, where they replace the regression in (\ref{eq:factorbryz}) with ridge regression:
\begin{equation}
	\label{eq:factorbryzridge}
	\hat{\mathbf{F}}_i^{t,\gamma} = \left( \frac{1}{L} \sum_{\ell = 1}^{L} W_{i,\ell}^t \hat{\mathbf{\Lambda}}^t_\ell (\hat{\mathbf{\Lambda}}^t_\ell)^\top + \gamma I_K \right)^{-1} \left( \frac{1}{L}\sum_{\ell=1}^{L} W_{i,\ell}^t \hat{\mathbf{\Lambda}}^t_\ell (\mathbf{C}^{t}_{\ell,i})^\top \right).
\end{equation}

In practice, the most commonly used methods for missing value imputation also include cross-sectional median, previous value, autoregression (AR) of order 1. All the methods are summarized in Table \ref{tab:methods}.

\newpage
\section{Supplementary Tables}\label{appsec:data}

This section details the firm characteristics included in the data set. The data cover 35 firm characteristics, which are described in Table \ref{tab:chars}. See \cite{freyberger2020dissecting} for detailed construction of each characteristic. The dataset is monthly and ranges from January 1966 to December 2020, and there are a total of 13,588 firms in the entire sample.
\begin{table}[h]
	\caption{Firm Characteristics}\label{tab:chars}
	\centering
	{\footnotesize	\begin{tabular}{lll}
			\hline
			& & \\ [-3mm]
			(1)  & A2ME         & Assets to market cap                                     \\
			(2)  & AT           & Total assets                                             \\
			(3)  & ATO          & Net sales over lagged net operating assets               \\
			(4)  & BEME         & Book to market ratio                                     \\
			(5)  & Beta         & CAPM beta                                                \\
			(6)  & C            & Ratio of cash and short-term investments to total assets \\
			(7)  & CTO          & Capital turnover                                         \\
			(8)  & D2A          & Capital intensity                                        \\
			(9)  & DPI2A        & Change in property, plants, and equipment                \\
			(10) & E2P          & Earnings to price                                        \\
			(11) & FC2Y         & Fixed costs to sales                                     \\
			(12) & Free CF      & Cash flow to book value of equity                        \\
			(13) & Idio vol     & Idiosyncratic volatility                                 \\
			(14) & Investment   & Percentage year-on-year growth rate in total assets      \\
			(15) & Lev          & Leverage                                                 \\
			(16) & LME          & Total market capitalization                              \\
			(17) & LTurnover    & Last month's volume over shares outstanding              \\
			(18) & NOA          & Net operating assets                                     \\
			(19) & OA           & Operating accruals                                       \\
			(20) & OL           & Operating leverage                                       \\
			(21) & PCM          & Price-to-cost margin                                     \\
			(22) & PM           & Profit margin                                            \\
			(23) & Prof         & profitability                                            \\
			(24) & Q            & Tobin's Q                                                \\
			(25) & Rel to High  & Closeness to 52-week high                                \\
			(26) & RNA          & Return on net operating assets                           \\
			(27) & ROA          & Return-on-assets                                         \\
			(28) & ROE          & Return-on-equity                                         \\
			(29) & ${\rm r}_{12-2}$  & Momentum                                                 \\
			(30) & ${\rm r}_{12-7}$  & Intermediate momentum                                    \\
			(31) & ${\rm r}_{2-1}$   & Short-term reversal                                      \\
			(32) & ${\rm r}_{36-13}$ & Long-term reversal                                       \\
			(33) & S2P          & Sales-to-price ratio                                     \\
			(34) & SGA2S        & SG\&A to sales                                           \\
			(35) & Spread       & Bid-ask spread                                          \\ [1mm] \hline
	\end{tabular}}
\end{table}

\clearpage

\section{Proofs}\label{appsec:proofs}
This appendix contains proofs of all results with several supplementary lemmas. The first lemma provides some useful properties of Kronecker products that will be used repeatedly in proofs.
\begin{lemma}\label{lemma:kronecker}
	If $\Lambda_j,M_j\in\mathbb{R}^{N_j\times R_j}$ for $j\leq d$, then
	\begin{enumerate}
		\item[(a)] $\left\|\bigotimes_{j=1}^d\Lambda_j\right\|_{\rm op} = \prod_{j=1}^d\|\Lambda_j\|_{\rm op}$.
		\item[(b)] $\left(\bigotimes_{j=1}^d\Lambda_j\right)^\top\left(\bigotimes_{j=1}^dM_j\right) = \bigotimes_{j=1}^d(\Lambda_j^\top M_j)$.
	\end{enumerate}
\end{lemma}
\begin{proof}
	The proof is based on the following three properties of Kronecker product:
	\begin{enumerate}
		\item $(A\otimes B)^\top = A^\top\otimes B^\top$.
		\item $(A_1\otimes A_2)(B_1\otimes B_2) = (A_1B_1)\otimes (A_2B_2)$.
		\item $(\alpha A_1+\beta A_2)\otimes B = \alpha A_1\otimes B + \beta A_2\otimes B$, where $\alpha,\beta\in\mathbb{R}$.
	\end{enumerate}
	These properties can be found in \cite{abadir2005matrix}, Exercises 10.3 and 10.7. Then (b) follows from
	\begin{equation*}
		\left(\bigotimes_{j=1}^d\Lambda_j\right)^\top\left(\bigotimes_{j=1}^dM_j\right)  = \left(\bigotimes_{j=1}^d\Lambda_j^\top\right)\left(\bigotimes_{j=1}^dM_j\right) = \bigotimes_{j=1}^d(\Lambda_j^\top M_j).
	\end{equation*}
	For (a), using the SVD decomposition of $\Lambda_j = \sum_{r_j=1}^{R_j}\sigma_{j,r_j}u_{j,r_j}v_{j,r_j}^\top$
	\begin{equation*}
		\begin{aligned}
			\bigotimes_{j=1}^d\Lambda_j & = \sum_{r_1,\dots,r_d}\sigma_{1,r_1}\dots\sigma_{d,r_d}\bigotimes_{j=1}^du_{j,r_j}v_{j,r_j}^\top \\
			& = \sum_{r_1,\dots,r_d}\sigma_{1,r_1}\dots\sigma_{d,r_d}\left(\bigotimes_{j=1}^du_{j,r_j}\right)\left(\bigotimes_{j=1}^dv_{j,r_j}\right)^\top,
		\end{aligned}
	\end{equation*}
	where we use properties 1-3. This shows that $\bigotimes_{j=1}^d\Lambda_j$ has singular values $\{\prod_{j=1}^d\sigma_{j,r_j}:\;r_j\leq R_j,j\leq d\}$, which implies (a):
	\begin{equation*}
		\left\|\bigotimes_{j=1}^d\Lambda_j\right\|_{\rm op} = \max_{r_j\leq R_j,j\leq d}\prod_{j=1}^d\sigma_{j,r_j} = \prod_{j=1}^d\max_{r_j\leq R_j}\sigma_{j,r_j} = \prod_{j=1}^d\|\Lambda_j\|_{\rm op}.
	\end{equation*}
\end{proof}

\begin{proof}[Proof of Proposition~\ref{prop:id}]
	By Lemma~\ref{lemma:kronecker}, 
	\begin{equation*}
		\left(\bigotimes_{k\ne j}\Lambda_k\right)^\top\left(\bigotimes_{k\ne j}\Lambda_k\right)  = \bigotimes_{k\ne j}(\Lambda_j^\top \Lambda_j) = \bigotimes_{k\ne j}I_{R_k}= I_{\prod_{k\ne j}R_k},
	\end{equation*}
	where the second equality follows under Assumption~\ref{as:orthogonal}, (i).
\end{proof}

\begin{proof}[Proof of Theorem~\ref{thm:rates}]
	Recall that $\hat \Lambda_j$ is $N_j\times R_j$ matrix of $R_j$ leading left singular vectors of
	\begin{equation*}
		\mathbf{Y}_{(j)} = \Lambda_j\mathbf{G}_{(j)}\left(\bigotimes_{k\ne j}\Lambda_k\right)^\top + \mathbf{U}_{(j)}.
	\end{equation*}
	
	In what follows, let $c,C>0$ be some numerical constant with values that can be different in different places. By \cite{cai2018rate}, equation (1.17) and the equation below it, we have
	\begin{equation*}
		\begin{aligned}
			\Pr\left(\|\sin\Theta(\hat\Lambda_j,\Lambda_j)\|_{\rm op}^2 >  C\frac{\left(\delta^2 + \prod_{l\ne j}N_l\right)\|\Lambda_j^\top\mathbf{Y}_{(j)}P_{\mathbf{Y}_{(j)}^\top\Lambda_j} \|^2_{\rm op}}{\delta^4}\right) \leq Ce^{-\frac{c\delta^4}{\delta^2 + \prod_{l\ne j}N_l}},
		\end{aligned}
	\end{equation*}
	where 
	\begin{equation*}
		\|\sin\Theta(\hat\Lambda_j,\Lambda_j)\|_{\rm op} = \max_{1\leq r\leq R_j}\cos^{-1}(\hat \sigma_{r,j})
	\end{equation*}
	and $\hat\sigma_{1,j}\geq \dots\geq \hat \sigma_{j,R_j}\geq 0$ are the singular values of $\Lambda_j^\top\hat\Lambda_j$.
	
	By \cite{cai2018rate}, equation (1.16), with $x=C\sqrt{N_j}$, we also have
	\begin{equation*}
		\begin{aligned}
			\Pr\left(\|\Lambda_j^\top\mathbf{Y}_{(j)}P_{\mathbf{Y}_{(j)}^\top\Lambda_j} \|_{\rm op} \geq C\sqrt{N_j}\right) & \leq Ce^{CN_j - c\min\{C^2N_j,C\sqrt{N_j}\sqrt{\delta^2 + \prod_{l\ne j}^dN_l} \}} + Ce^{-c(\delta^2 + \prod_{l\ne j}N_l)} \\
			& \lesssim e^{-cN_j} + e^{-c(\delta^2 + \prod_{l\ne j}N_l)} \lesssim Ce^{-cN_j},
		\end{aligned}
	\end{equation*}
	where we use $\delta^2\gtrsim \prod_{j=1}^d\sqrt{N_j}+\max_{1\leq j\leq d}N_j$, cf. Assumption~\ref{as:strong_factors}.
	
	Therefore,
	\begin{equation*}
		\|\hat\Lambda_jO_j - \Lambda_j\|_{\rm F}^2 \leq 2R_j\|\sin\Theta(\hat\Lambda_j,\Lambda_j)\|_{\rm op}^2 \lesssim \frac{R_j\left(\delta^2 + \prod_{l\ne j}N_l\right)N_j}{\delta^4}
	\end{equation*}
	with probability at least $1-Ce^{-cN_j}$.
\end{proof}

The following Lemma provides some useful properties of metricized mode-$j$ products. 
\begin{lemma}\label{lemma:tensor_matricization}
	For tensors $\mathbf{Y}\in\R^{N_1\times\dots\times N_d}$ and $\mathbf{G}\in\R^{R_1\times\dots\times R_d}$, and matrices $\Lambda_l\in\R^{N_l\times R_l}$ with $1\leq l\leq d$, we have
	\begin{enumerate}
		\item $\left(\mathbf{Y}\bigtimes_{l\ne j}\Lambda_l^\top \right)_{(j)} = \mathbf{Y}_{(j)}\left(\bigotimes_{l\ne j}\Lambda_l\right)$ for all $1\leq j\leq d$.
		\item $\left(\mathbf{G}\bigtimes_{l=1}^d\Lambda_l \right)_{(j)} =\Lambda_j \mathbf{G}_{(j)}\left(\bigotimes_{l\ne j}\Lambda_l\right)^\top$ for all $1\leq j\leq d$.
	\end{enumerate}
\end{lemma}
\begin{proof}
	These properties can be verified by looking at the entries of corresponding tensors; see also Proposition 3.7 in \cite{kolda2006multilinear}.
\end{proof}

\begin{lemma}\label{lemma:errors_tensor}
	Suppose that $\mathbf{U}\in\R^{N_1\times\dots\times N_d}$ is a tensor with i.i.d. mean-zero subgaussian entries. Then
	\begin{equation*}
		\left\|\mathbf{U}_{(j)}\bigotimes_{l\ne j} \Lambda_l\right\|_{\rm op} \lesssim \sqrt{N_j} +\sqrt{\prod_{l\ne j}R_l}
	\end{equation*}
	and
	\begin{equation*}
		\sup_{\|\Lambda_l\|_{\rm op}\leq 1,\forall l\ne j}\left\|\mathbf{U}_{(j)}\bigotimes_{l\ne j} \Lambda_l\right\|_{\rm op} \lesssim \sqrt{N_j} +\sqrt{\prod_{l\ne j}R_l} + \sum_{l\ne j}\sqrt{N_lR_l}
	\end{equation*}
	with probability at least $1-e^{-c\min_{1\leq j\leq d}N_j}$ for some $c>0$.
\end{lemma}
\begin{proof}
	See \cite{han2022exact}, Lemma 9.
\end{proof}

\begin{proof}[Proof of Theorem~\ref{thm:als}]
	By Lemma~\ref{lemma:tensor_matricization}, the mode-$j$ matricization of $\hat{\mathbf{Y}}^{(k)}$ at step $k$ computed in Algorithm~\ref{alg:hooi} is
	\begin{equation}\label{eq:projection}
		\begin{aligned}
			\hat{\mathbf{Y}}^{(k)}_{(j)} & = \left( \mathbf{Y}\bigtimes_{l\ne j}\hat \Lambda_l^{(k-1)\top}\right)_{(j)} \\
			& = \mathbf{Y}_{(j)}\bigotimes_{l\ne j}\hat \Lambda_l^{(k-1)} \\
			& = \Lambda_j\mathbf{G}_{(j)}\left(\bigotimes_{l\ne j}\Lambda_l\right)^\top\bigotimes_{l\ne j}\hat \Lambda_l^{(k-1)} + \mathbf{U}_{(j)}\bigotimes_{l\ne j}\hat \Lambda_l^{(k-1)} \\
			& = \Lambda_j\mathbf{G}_{(j)}\left(\bigotimes_{l\ne j}O_l\right)^\top + \Lambda_j\mathbf{G}_{(j)}\left(\bigotimes_{l\ne j}O_l\right)^\top\bigotimes_{l\ne j}\left\{O_l\Lambda_l^\top\hat\Lambda_l^{(k-1)}- I_{R_l}\right\} + \mathbf{U}_{(j)}\bigotimes_{l\ne j}\hat \Lambda_l^{(k-1)},
		\end{aligned}
	\end{equation}
	where the last line follows by Lemma~\ref{lemma:kronecker}. Consider the $\sin\Theta$ distance between the singular subspaces
	\begin{equation*}
		\|\sin\Theta(\hat\Lambda_j,\Lambda_j)\|_{\rm op} = \max_{1\leq r\leq R_j}\cos^{-1}(\hat \sigma_{r,j}),
	\end{equation*}
	where $\hat\sigma_{1,j}\geq \dots\geq \hat \sigma_{j,R_j}\geq 0$ are the singular values of $\Lambda_j^\top\hat\Lambda_j$.
	
	Since $\hat\Lambda_j^{(k)}$ and $\Lambda_j$ are the $R_j$ leading left-singular vectors of $\hat{\mathbf{Y}}_{(j)}^{(k)}$ and the first term in the last line of equation~(\ref{eq:projection}) respectively, by Wedin's perturbation theorem, see \cite{wedin1972perturbation},
	\begin{equation*}
		\begin{aligned}
			\|\sin\Theta(\hat\Lambda_j,\Lambda_j)\|_{\rm op}  & \lesssim \frac{1}{\delta}\left\|\Lambda_j\mathbf{G}_{(j)}\bigotimes_{l\ne j}\left\{\Lambda_l^\top\hat\Lambda_l^{(k-1)}-O_l^\top\right\} + \mathbf{U}_{(j)}\bigotimes_{l\ne j}\hat \Lambda_l^{(k-1)}\right\|_{\rm op} \\
			& \leq \frac{1}{\delta}\left\|\Lambda_j\mathbf{G}_{(j)}\bigotimes_{l\ne j}\left\{\Lambda_l^\top\hat\Lambda_l^{(k-1)}-O_l^\top\right\} \right\|_{\rm op} + \frac{1}{\delta}\left\|\mathbf{U}_{(j)}\bigotimes_{l\ne j}\hat \Lambda_l^{(k-1)}\right\|_{\rm op}.
		\end{aligned}
	\end{equation*}
	Under Assumption~\ref{as:data}, by Lemma~\ref{lemma:errors_tensor}
	\begin{equation*}
		\begin{aligned}
			\left\|\mathbf{U}_{(j)}\bigotimes_{l\ne j}\hat \Lambda_l^{(k-1)}\right\|_{\rm op} & \leq \sup_{\|\Lambda_l\|_{\rm op}\leq 1,\forall l\ne j}\left\|\mathbf{U}_{(j)}\bigotimes_{l\ne j} \Lambda_l\right\|_{\rm op} \\
			& \lesssim \sum_{j=1}^d\sqrt{N_j} + \sqrt{\prod_{l\ne j}R_l} + \sum_{l\ne j}\sqrt{N_lR_l} \\
			& \lesssim d\sqrt{N} + R^{(d-1)/2} + (d-1)\sqrt{NR},
		\end{aligned}	
	\end{equation*}	
	where we use $N_j\sim N$ and define $R=\max_{1\leq j\leq d}R_j$.

	Under Assumptions~\ref{as:orthogonal} and \ref{as:strong_factors}, by Lemma~\ref{lemma:kronecker}, we also have
	\begin{equation*}
		\begin{aligned}
			\frac{1}{\delta}\left\|\Lambda_j\mathbf{G}_{(j)}\bigotimes_{l\ne j}\left\{\Lambda_l^\top\hat\Lambda_l^{(k-1)}- O_l^\top\right\}\right\|_{\rm op} & \leq \frac{\sigma_{1,j}}{\delta}\left\|\bigotimes_{l\ne j}\left\{\Lambda_l^\top\hat\Lambda_l^{(k-1)}-O_l^\top\right\}\right\|_{\rm op} \\
			& \lesssim\prod_{l\ne j}\left\|\Lambda_l^\top (\hat \Lambda_l^{(k-1)}O_l - \Lambda_l) \right\|_{\rm op} \\
			& \leq \prod_{l\ne j}\left\|\hat\Lambda_l^{(k-1)}O_l - \Lambda_l \right\|_{\rm F} \\
		\end{aligned}
	\end{equation*}
	This shows that
	\begin{equation}\label{eq:als_sin_Theta}
		\|\sin\Theta(\hat\Lambda_j,\Lambda_j)\|_{\rm op} \lesssim  \prod_{l\ne j}\left\|\hat\Lambda_l^{(k-1)}O_l - \Lambda_l \right\|_{\rm F} + \frac{d\sqrt{N} + R^{(d-1)/2} + (d-1)\sqrt{RN}}{\delta}.
	\end{equation}
	From the proof of Theorem~\ref{thm:rates}, we know that the naive TPCA computed with Algorithm~\ref{alg:hosvd} satisfies
	\begin{equation*}
		E_0\lesssim \frac{\sqrt{N}}{\delta} + \frac{N^{d/2}}{\delta^2}
	\end{equation*}
	with probability at least $1-Ce^{-cN}$, provided that $N_j\sim N$ for all $1\leq j\leq d$. Since $\delta^{-4}N^d = o(1)$, we have $E_0=o(1)$ for a sufficiently large $N$. Let $E_k := \max_{1\leq j\leq d}\|\sin\Theta(\hat\Lambda_j,\Lambda_j)\|_{\rm op} $ be the largest error across all matricizations at iterations $k=0,1,\dots,\bar k$.  Then
	\begin{equation*}
		\begin{aligned}
			E_k & \lesssim E_{k-1}^{d-1} + \frac{d\sqrt{N} + \sqrt{R^{d-1}} + (d-1)\sqrt{RN}}{\delta}
		\end{aligned}
	\end{equation*}
	with probability at least $1-Ce^{-cN}$. If $E_{k-1}=o(1)$, then $E_k=o(1)$ since $N^{d/4}/\delta=o(1)$ and $R=O(1)$. This shows that $E_k=o(1),\forall k=0,1,\dots,\bar k$ because $E_0=o(1)$. Therefore, we can bound the ALS error for a sufficiently large $N$ as
	\begin{equation*}
		E_k \leq 0.5E_{k-1} + c\frac{\sqrt{N}}{\delta}.
	\end{equation*}
	Iterating this inequality, we obtain
	\begin{equation*}
		E_k \leq 2c\frac{\sqrt{N}}{\delta} + 0.5^kE_0 \lesssim \frac{\sqrt{N}}{\delta} + 0.5^k\left(\frac{\sqrt{N}}{\delta} + \frac{N^{d/2}}{\delta^2}\right) \lesssim \frac{\sqrt{N}}{\delta},
	\end{equation*}
	where the last inequality follows provided that $k\geq \log(N^{(d-1)/2}/\delta)/\log 2$. This shows that
	\begin{equation*}
		\|\hat\Lambda_j^{(\bar k)}O_j - \Lambda_j\|_{\rm F}^2 \leq 2R_j\|\sin\Theta(\hat\Lambda_j,\Lambda_j)\|_{\rm op}^2 \lesssim E_k^2 \lesssim \frac{N}{\delta^2}.
	\end{equation*}
\end{proof}

Next, we prove a central limit theorem for quadratic forms which is an important ingredient in the proof of Theorem~\ref{thm:clt}.
\begin{lemma}\label{lemma:clt_quadratic_form}
	Suppose that $\mathbf{U}=\{u_{i,j}:1\leq i\leq N_1,1\leq j\leq N_2 \}$ is a matrix of i.i.d.\ random variables such that $\E(u_{i,j})=0$, $\Var(u_{i,j})=\sigma^2$, and $\E|u_{i,j}|^{2+\delta}<\infty$ for some $\delta>0$. Let $(x_r)_{r=1}^R$ and $(y_r)_{r=1}^R$ be two orthonormal sets in $\R^{N_1}$ and $\R^{N_2}$ such that $\max_r\|x_r\|_\infty = o(1)$ or $\max_r\|y_r\|_\infty=o(1)$. Then
	\begin{equation*}
		(x_r^\top \mathbf{U} y_k)_{1\leq r,k\leq R} \xrightarrow{d} Z,
	\end{equation*}
	where $Z$ is an $R\times R$ matrix with i.i.d.\ $N(0,\sigma^2)$ entries.
\end{lemma}
\begin{proof}
	According to the Cram\'{e}r-Wold argument, it is enough to show that
	\begin{equation*}
		\sum_{r,k}t_{r,k}x_r^\top\mathbf{U}y_k \xrightarrow{d}\sum_{r,k}t_{r,k}Z_{r,k}
	\end{equation*}
	for every $(t_{r,k})_{1\leq r,k\leq R}\in\R^{R\times R}$. Note that since $\mathbf{U}$ is a matrix of i.i.d.\ mean zero random variables,
	\begin{equation*}
		\sum_{r,k}t_{r,k}x_r^\top\mathbf{U}y_k  = \sum_{i=1}^{N_1}\underbrace{\sum_{j=1}^{N_2}u_{i,j}\sum_{r,k}t_{r,k}x_{r,i}y_{k,j}}_{=: \xi_{N_2,i}}  = \sum_{j=1}^{N_2}\underbrace{\sum_{i=1}^{N_1}u_{i,j}\sum_{r,k}t_{r,k}x_{r,i}y_{k,j}}_{=: \xi_{N_1,j}}
	\end{equation*}
	is a partial sum of triangular arrays of independent centered random variables. Moreover,
	\begin{equation*}
		\begin{aligned}
			\Var\left(\sum_{r,k}t_{r,k}x_r^\top\mathbf{U}y_k \right) & = \sigma^2\sum_{r,k,r',k'}t_{r,k}t_{r',k'}\langle x_r,x_{r'}\rangle\langle y_k,y_{k'}\rangle \\
			& = \sigma^2\sum_{r,k}t_{r,k}^2 \\
			& = \Var\left(\sum_{r,k}t_{r,k}\xi_{r,k}\right),
		\end{aligned}
	\end{equation*}
	where  the second line follows under orthonormality. The result follows provided that one of the following Lyapunov's conditions
	\begin{equation*}
		\sum_{i=1}^{N_1}\E\left|\xi_{N_2,i}\right|^{2+\delta} = o(1)\qquad \text{or}\qquad \sum_{j=1}^{N_2}\E\left|\xi_{N_1,j}\right|^{2+\delta} = o(1)
	\end{equation*}
	holds for some $ \delta>0$; see \cite{bill1995_book}, Theorem 27.3. By Rosenthal's inequality, see \cite{petrov1995limit}, Theorem 2.9, there exists $c_\delta<\infty$ such that
	\begin{equation*}
		\begin{aligned}
			\E\left|\xi_{N_2,i}\right|^{2+\delta} & = \E\left|\sum_{j=1}^{N_2}u_{i,j}\sum_{r,k}t_{r,k}x_{r,i}y_{k,j}\right|^{2+\delta} \\
			& \leq c_\delta\left\{\sum_{j=1}^{N_2}\E\left|u_{i,j}\sum_{r,k}t_{r,k}x_{r,i}y_{k,j}\right|^{2+\delta} + \left(\sum_{j=1}^{N_2}\E\left|u_{i,j}\sum_{r,k}t_{r,k}x_{r,i}y_{k,j}\right|^2\right)^{1+\delta/2}\right\} \\
			& \lesssim \sum_{r,k}t_{r,k}^{2+\delta}\sum_{j=1}^{N_2}x_{r,i}^{2+\delta}y_{k,j}^{2+\delta} + \sum_{r,k}t_{r,k}^{2+\delta}\left(\sum_{j=1}^{N_2}x_{r,i}^2y_{k,j}^2\right)^{1+\delta/2},
		\end{aligned}
	\end{equation*}
	where the last line follows by Jensen's inequality and maintained assumptions. Therefore,
	\begin{equation*}
		\begin{aligned}
			\sum_{i=1}^{N_1}\E\left|\xi_{N_2,i}\right|^{2+\delta} & \lesssim \max_r\|x_r\|_{2+\delta}^{2+\delta}\max_k\left[\|y_k\|_{2+\delta}^{2+\delta} + \|y_k\|_2^{2+\delta} \right] \\
			& \leq 2\max_r\|x_r\|_\infty^\delta = o(1),
		\end{aligned}
	\end{equation*}
	where we use the fact that $\|y_k\|_{2+\delta}\leq \|y_k\|_2$ and orthonormality. Similarly, we could verify the second Lyapunov's condition when $\max_r\|y_r\|=o(1)$.
\end{proof}

\begin{proof}[Proof of Theorem~\ref{thm:clt}]
	Note that
	\begin{equation*}
		\begin{aligned}
			\mathbf{Y}_{(j)}\mathbf{Y}_{(j)}^\top - \Lambda_jD_j\Lambda_j^\top & = \Lambda_j\mathbf{G}_{(j)}\left(\bigotimes_{l\ne j}\Lambda_l\right)^\top\mathbf{U}_{(j)}^\top + \mathbf{U}_{(j)}\left(\bigotimes_{l\ne j}\Lambda_l\right)\mathbf{G}_{(j)}^\top\Lambda_j^\top + \mathbf{U}_{(j)}\mathbf{U}_{(j)}^\top \\
			& =: E_j.
		\end{aligned}
	\end{equation*}
	By the triangle inequality
	\begin{equation*}
		\begin{aligned}
			\|E_j\|_{\rm op} & \leq  \left\|\Lambda_j\mathbf{G}_{(j)}\left(\bigotimes_{l\ne j}\Lambda_l\right)^\top\mathbf{U}_{(j)}^\top\right\|_{\rm op} + \left\|\mathbf{U}_{(j)}\left(\bigotimes_{l\ne j}\Lambda_l\right)\mathbf{G}_{(j)}^\top\Lambda_j^\top\right\|_{\rm op} + \|\mathbf{U}_{(j)}\mathbf{U}_{(j)}^\top\|_{\rm op} \\
			& = 2\left\|\mathbf{U}_{(j)}\left(\bigotimes_{l\ne j}\Lambda_l\right)\mathbf{G}_{(j)}^\top\Lambda_j^\top\right\|_{\rm op} + \|\mathbf{U}_{(j)}\|_{\rm op}^2
		\end{aligned}	
	\end{equation*}
	Under Assumption~\ref{as:data} by \cite{latala2005some}
	\begin{equation*}
		\E\|\mathbf{U}_{(j)}\|_{\rm op} = O\left(\sqrt{N_j} + \prod_{l\ne j}\sqrt{N_l}\right).
	\end{equation*}
	Consider the SVD decomposition
	\begin{equation*}
		\left(\bigotimes_{l\ne j}\Lambda_l\right)\mathbf{G}_{(j)}^\top\Lambda_j^\top = V_jD_j^{1/2}\Lambda_j^\top = \sum_{r=1}^{R_j}\sigma_{j,r}v_{j,r}\lambda_{j,r}^\top.
	\end{equation*}
	Then
	\begin{equation*}
		\begin{aligned}
			\left\|\mathbf{U}_{(j)}\left(\bigotimes_{l\ne j}\Lambda_l\right)\mathbf{G}_{(j)}^\top\Lambda_j^\top\right\|_{\rm op} & \leq \sum_{r=1}^{R_j}\sigma_{j,r}\|\mathbf{U}_{(j)}v_{j,r}\lambda_{j,r}^\top\|_{\rm op} \\
			& = \sum_{r=1}^{R_j}\sigma_{j,r}\|\mathbf{U}_{(j)}v_{j,r}\| \\
			& \lesssim \prod_{j=1}^d\sqrt{N_j}\max_{1\leq r\leq R_j}\|\mathbf{U}_{(j)}v_{j,r}\|,
		\end{aligned}
	\end{equation*}
	where we use $\lambda_{j,r}^\top\lambda_{j,r}=1$ under Assumption~\ref{as:orthogonal} and $\sigma_{j,R_j}^2\sim \prod_{j=1}^dN_j$ under Assumption~\ref{as:clt} (i). Under Assumption~\ref{as:data}, we also have
	\begin{equation*}
		\E\|\mathbf{U}_{(j)}v_{j,r}\|^2 = v_{j,r}^\top\E[\mathbf{U}_{(j)}^\top \mathbf{U}_{(j)}]v_{j,r} = \sigma^2N_j,
	\end{equation*}
	since $v_{j,r}^\top v_{j,r}=1$. This shows that
	\begin{equation*}
		\|E_j\|_{\rm op}/\sigma^2_{j,R_j} = O_P\left(\frac{1}{\prod_{l\ne j}\sqrt{N_l}} + \frac{1}{N_j}\right) = o_P(1),\quad \text{as}\quad (N_1,\dots,N_d)\to \infty.
	\end{equation*}
	Therefore, $\|E_j\|_{\rm op}/\sigma_{j,R_j}^2<1/4$ with probability approaching one. Let $\Lambda_{j\perp}$ be the orthonormal basis of the orthogonal complement to the column space of $\Lambda_j$. By Theorem~\ref{thm:perturbation}
	\begin{equation*}
		\hat\Lambda_jO - \Lambda_j = \Lambda_{j\perp}\Lambda_{j\perp}^\top E_j\Lambda_jD_j^{-1} + r_j,
	\end{equation*}
	where the remainder term is
	\begin{equation*}
		\begin{aligned}
			\|r_j\|_{\rm op} & \leq 8\sigma_{j,1}^2\|E_j\|_{\rm op}^2/\sigma^6_{j,R_j} \\
			& = O_P\left(\frac{1}{\prod_{j=1}^dN_j} + \frac{1}{N_j^2} \right).
		\end{aligned}	
	\end{equation*}
	Therefore,
	\begin{equation*}
		\begin{aligned}
			\hat\Lambda_jO_j - \Lambda_j & = \Lambda_{j\perp}\Lambda_{j\perp}^\top E_j\Lambda_jD_j^{-1} + r_j\\
			& = \Lambda_{j\perp}\Lambda_{j\perp}^\top\left[\Lambda_j\mathbf{G}_{(j)}\left(\bigotimes_{l\ne j}\Lambda_l\right)^\top\mathbf{U}_{(j)}^\top + \mathbf{U}_{(j)}\left(\bigotimes_{l\ne j}\Lambda_l\right)\mathbf{G}_{(j)}^\top\Lambda_j^\top + \mathbf{U}_{(j)}\mathbf{U}_{(j)}^\top \right]\Lambda_jD_j^{-1}  + r_j \\
			& = \Lambda_{j\perp}\Lambda_{j\perp}^\top\mathbf{U}_{(j)}\left(\bigotimes_{l\ne j}\Lambda_l\right)\mathbf{G}_{(j)}^\top D_j^{-1} +  \Lambda_{j\perp}\Lambda_{j\perp}^\top\mathbf{U}_{(j)}\mathbf{U}_{(j)}^\top\Lambda_jD_j^{-1} + r_j \\
			& = \mathbf{U}_{(j)}V_jD_j^{-1/2} +  \Lambda_{j\perp}\Lambda_{j\perp}^\top\mathbf{U}_{(j)}\mathbf{U}_{(j)}^\top\Lambda_jD_j^{-1} - \Lambda_j\Lambda_j^\top\mathbf{U}_{(j)}V_jD_j^{-1/2}   + r_j \\
		\end{aligned}
	\end{equation*}
	where we use $\Lambda_j\Lambda_j^\top + \Lambda_{j\perp}\Lambda_{j\perp}^\top = I$ and $\left(\bigotimes_{l\ne j}\Lambda_l\right)\mathbf{G}_{(j)}^\top D_j^{-1} = V_jD_j^{-1/2}$. Under Assumption~\ref{as:clt} (ii)
	\begin{equation*}
		\|\Lambda_j\Lambda_j^\top \mathbf{U}_{(j)}V_j \|_{2,\infty} \lesssim  \sqrt{\frac{R_j}{N_j}}\|\mathbf{U}_{(j)}V_j\|_{2,\infty},
	\end{equation*}
	where
	\begin{equation*}
		\|\mathbf{U}_{(j)}V_j\|_{2,\infty} = \max_{1\leq i\leq N_j}\|\mathbf{u}_{(j),i} V_j\| \leq R_j\max_{1\leq r\leq R_j}\max_{1\leq i\leq N_j}|\mathbf{u}_{(j),i} v_{j,r}|
	\end{equation*}
	and we partition $\mathbf{U}_{(j)}^\top = [\mathbf{u}_{(j),1},\dots, \mathbf{u}_{(j),N}]$. Under Assumption~\ref{as:data}, by \cite{vershynin2018high}, Lemma 3.4.2, this is a maximum of sub-Gaussian random variables, so that
	\begin{equation*}
		\|\Lambda_j\Lambda_j^\top \mathbf{U}_{(j)}V_j \|_{2,\infty}  = O_P\left(R_j\sqrt{\frac{\log N_j}{N_j}} \right).
	\end{equation*}
	Putting $B_j = \Lambda_{j\perp}\Lambda_{j\perp}^\top\mathbf{U}_{(j)}\mathbf{U}_{(j)}^\top\Lambda_jD_j^{-1}$, we obtain
	\begin{equation*}
		(\hat\Lambda_jO_j - \Lambda_j - B_j)D_j^{1/2} = \mathbf{U}_{(j)}V_j + s_j
	\end{equation*}
	with
	\begin{equation*}
		\begin{aligned}
			\|s_j\|_{2,\infty} & = O_P\left(R_j\sqrt{\frac{\log N_j}{N_j}}+\frac{\sigma_{j,1}}{\prod_{j=1}^dN_j} + \frac{\sigma_{j,1}}{N_j^2}\right) \\
			& = O_P\left(R_j\sqrt{\frac{\log N_j}{N_j}}+\frac{1}{\prod_{j=1}^d\sqrt{N_j}} + \frac{\prod_{j=1}^d\sqrt{N_j}}{N_j^2}\right),
		\end{aligned}
	\end{equation*}
	where the second line follows under Assumption~\ref{as:clt}. By Theorem~\ref{lemma:clt_quadratic_form} the $(i,j)^{\rm th}$ element of $\mathbf{U}_{(j)}V_j$ is
	\begin{equation*}
		\mathbf{u}_{(j),i}^\top V_j \xrightarrow{d} N(0,\sigma^2I_{R_j}),
	\end{equation*}
	given that $\|V_j\|_{2,\infty}=o(1)$.
\end{proof}

Next, we provide a proof of the central limit theorem for scale components.
\begin{proof}[Proof of Theorem~\ref{thm:clt_scale}]
	Consider the SVD decomposition
	\begin{equation*}
		\Lambda_j\mathbf{G}_{(j)}\left(\bigotimes_{l\ne j}\Lambda_l\right)^\top = \Lambda_jD_j^{1/2}V_j^\top = \sum_{r=1}^{R_j}\sigma_{j,r}\lambda_{j,r}v_{j,r}^\top.
	\end{equation*}	
	By \cite{kneip2001inference}, Lemma A.1 (a),
		\begin{equation*}
			\begin{aligned}
				\hat\sigma^2_{j,r} - \sigma^2_{j,r} & = \mathrm{trace}\left\{\lambda_{j,r}\lambda_{j,r}^\top(\Lambda_jD_j^{1/2}V_j^\top\mathbf{U}_{(j)}^\top + \mathbf{U}_{(j)}V_jD_j^{1/2}\Lambda_j^\top + \mathbf{U}_{(j)}\mathbf{U}_{(j)}^\top)\right\} + \xi_{j,r}\\
				& = \lambda_{j,r}^\top\left(\Lambda_jD_j^{1/2}V_j^\top\mathbf{U}_{(j)}^\top + \mathbf{U}_{(j)}V_jD_j^{1/2}\Lambda_j^\top + \mathbf{U}_{(j)}\mathbf{U}_{(j)}^\top\right)\lambda_{j,r} + \xi_{j,r} \\
				& = 2\sigma_{j,r} \lambda_{j,r}^\top\mathbf{U}_{(j)}v_{j,r} + \lambda_{j,r}^\top\mathbf{U}_{(j)}\mathbf{U}_{(j)}^\top \lambda_{j,r} + \xi_{j,r},
			\end{aligned}
	\end{equation*}
	where
	\begin{equation*}
		\begin{aligned}
			|\xi_{j,r}| & \leq \frac{6}{\min_{s\ne r}|\sigma^2_{j,s}-\sigma^2_{j,r}|}\left\|\mathbf{Y}_{(j)}\mathbf{Y}_{(j)}^\top - \Lambda_jD_j^{1/2}\Lambda_j^\top\right\|_{\rm op}^2 \\
			& = O_P\left(N_j + \frac{N_j}{\prod_{l\ne j}N_l} + \frac{\prod_{l\ne j}N_l}{N_j}\right) = o_P(\sigma_{j,r}),
		\end{aligned}
	\end{equation*}
	which follows from $\sigma_{j,r}^2\sim\prod_{j=1}^d N_j$ and the proof of Theorem~\ref{thm:clt}, provided  that $N_j/\prod_{l\ne j}N_l=o(1)$ and $\prod_{l\ne j}N_l/N_j^3= o(1)$.  This shows that
	\begin{equation*}
		\frac{\hat \sigma_{j,r}^2 - \sigma_{j,r}^2}{\sigma_{j,r}} = 2\lambda_{j,r}^\top\mathbf{U}_{(j)}v_{j,r} + \sigma_{j,r}^{-1}\lambda_{j,r}^\top\mathbf{U}_{(j)}\mathbf{U}_{(j)}^\top \lambda_{j,r} + o_P(1).
	\end{equation*}
	
	Next, under Assumptions~\ref{as:data} and \ref{as:clt}
	\begin{equation*}
		\Var(\lambda_{j,r}^\top\mathbf{U}_{(j)}\mathbf{U}_{(j)}^\top \lambda_{j,r}) = \prod_{l\ne j}N_l\Var(\langle \lambda_{j,r},\mathbf{u}_{(j),i}\rangle^2) = \prod_{l\ne j}N_l\left(\E\langle \lambda_{j,r},\mathbf{u}_{(j),i}\rangle^4 - \sigma^4\right) = O\left(\prod_{l\ne j}N_l\right).
	\end{equation*}
	This shows that $\sigma^{-1}_{j,r}\lambda_{j,r}^\top\mathbf{U}_{(j)}\mathbf{U}_{(j)}^\top \lambda_{j,r} = o_P(\sigma_{j,r} )$, hence,
	\begin{equation*}
		\frac{\hat\sigma_{j,r} ^2 - \sigma_{j,r} ^2}{\sigma_{j,r} } = 2\lambda_{j,r}^\top\mathbf{U}_{(j)}v_{j,r} + o_P(1).
	\end{equation*}
	Finally, under Assumptions~\ref{as:orthogonal}, \ref{as:data}, and \ref{as:clt}, by Lemma~\ref{lemma:clt_quadratic_form}
	\begin{equation*}
		(\lambda_{j,r}^\top\mathbf{U}_{(j)}v_{j,r})_{1\leq r\leq R_j} \xrightarrow{d} N(0,\sigma^2I_{R_j}).
	\end{equation*}
\end{proof}

\begin{proof}[Proof of Theorem~\ref{thm:test}]
	Under Assumption~\ref{as:data}, if $u_{i_1,\dots,i_d}\sim N(0,\sigma^2)$, then $\Lambda_{j\perp}^\top\mathbf{u}_i^{(j)}\in\R^{N_j-R_j}$ is a vector of i.i.d.\ $N(0,\sigma^2)$ for every $\Lambda_{j\perp}$ such that $\Lambda_{j\perp}^\top \Lambda_{j\perp}=I$. Then under Assumptions~\ref{as:orthogonal}, \ref{as:data} (i), and \ref{as:strong_factors}, by Lemma~\ref{lemma:perturbation},
	\begin{equation*}
		\hat\sigma^2_{R+r,j} = \lambda_r\left(\mathbf{U}_{(j)}\mathbf{U}^\top_{(j)}\right) + O_P\left(N_j + \frac{\prod_{k\ne j}N_k}{N_j}\right),
	\end{equation*}
	where with some abuse of notation $\mathbf{U}_{(j)}$ is $(N_j-R)\times\prod_{k\ne j}N_k$ matrix of i.i.d.\ $N(0,\sigma^2)$. 
	
	By \cite{karoui2003largest},
	\begin{equation*}
		\left(\frac{\lambda_r\left(\mathbf{U}_{(j)}\mathbf{U}^\top_{(j)}\right) - \lambda_{r+1}\left(\mathbf{U}_{(j)}\mathbf{U}^\top_{(j)}\right)}{\tau}\right)_{r=k+1}^{K+1} \xrightarrow{d} (\xi_r-\xi_{r+1})_{r=1}^{K-k+1}.
	\end{equation*}
	Therefore,
	\begin{equation*}
		\left(\frac{\hat\sigma^2_{r,j} - \hat\sigma^2_{r+1,j}}{\tau}\right)_{r=k+1}^{K+1} \xrightarrow{d} (\xi_r-\xi_{r+1})_{r=1}^{K-k+1},
	\end{equation*}
	since $N_j/\tau + \prod_{k\ne j}N_k/(N_j\tau) = o(1)$. The result under $H_0$ follows by the continuous mapping theorem.
	
	For $H_1$, recall that
	\begin{equation*}
		\mathbf{Y}_{(j)} = \Lambda_j\mathbf{G}_{(j)}\left(\bigotimes_{k\ne j}\Lambda_k\right)^\top + \mathbf{U}_{(j)}.
	\end{equation*}
	By Weyl's inequality for singular values, see \cite{horn_joh_2013}, Eq. (7.3.15),
	\begin{equation*}
		\left|\frac{\hat\sigma_{j,R_j}}{\sqrt{\prod_{k\ne j}N_k}} - \frac{\sigma_{j,R_j}}{\sqrt{\prod_{k\ne j}N_k}}\right| \leq \frac{\|\mathbf{U}_{(j)}\|_{\rm op}}{\sqrt{\prod_{k\ne j}N_k}}.
	\end{equation*}
	The right-hand side of this equation is $O_P(1)$ by \cite{latala2005some}, while $\sigma_{j,R_j}/\sqrt{\prod_{k\ne j}N_k}\uparrow\infty$ under under the strong factor hypothesis. Therefore, $\hat\sigma_{j,R_j}^2/\prod_{k\ne j}N_k\uparrow\infty$. This implies that $(\hat\sigma^2_{j,R_j} - \hat\sigma^2_{j,R_j+1})/\tau\to\infty$, and whence $S_j\uparrow\infty$ under $H_1$.
\end{proof}

\section{Perturbation Theory}
In this section we formulate the required perturbation expansions for eigenvectors and eigenvalues following \cite{kato1996perturbation}. The perturbation theory has been successfully applied in statistics, see \cite{watson1983statistics}, \cite{kneip2001inference}, or \cite{koltchinskii2017normal}, but is not widely adopted in the factor literature in econometrics; see \cite{onatski2006formal} or \cite{lei2023estimating} for notable exceptions.

\begin{theorem}\label{thm:perturbation}
	Let $\Lambda=[\lambda_1,\dots,\lambda_R]\in\mathbb{R}^{N\times R}$ be the eigenvectors of a symmetric matrix $\Sigma\in\mathbb{R}^{N\times N}$ with $\mathrm{rank}(\Sigma)=R<N$ corresponding to non-zero eigenvalues $\sigma_1>\dots>\sigma_R>0$. Let $\hat\Lambda = [\hat\lambda_1,\dots\hat\lambda_R]\in\mathbb{R}^{N\times R}$ be the leading eigenvectors of its noisy measure $\hat\Sigma=\Sigma+E$, where the noise satisfies $\|E\|_{\rm op}<\sigma_R/4$. Then
	\begin{equation*}
		\hat\Lambda O - \Lambda = \Lambda_{\perp}\Lambda_{\perp}^\top E\Lambda D^{-1} + r_N,
	\end{equation*}
	where $O=\hat\Lambda^\top\Lambda$, $\Lambda_\perp = [\lambda_{R+1},\dots,\lambda_N]\in\mathbb{R}^{N\times (N-R)}$, and $r_N\leq 8\sigma_1\|E\|^2_{\rm op}/\sigma_R^3$.
\end{theorem}
\begin{proof}
	Let $\Sigma = \sum_{j=1}^N\sigma_j\lambda_j\lambda_j^\top$  be the eigendecomposition of $\Sigma$. If $\gamma$ is a contour in $\mathbb{C}$ that does not go through any eigenvalue of $\Sigma$, then
	\begin{equation*}
		\begin{aligned}
			\frac{1}{2\pi i}\oint_\gamma (\zeta I - \Sigma)^{-1}\dx\zeta & = \frac{1}{2\pi i}\oint_\gamma\sum_{j=1}^N\frac{1}{\zeta - \sigma_j}\lambda_j\lambda_j^\top\dx\zeta \\
			& = \sum_{j:\sigma_j\in\gamma}\lambda_j\lambda_j^\top\frac{1}{2\pi i}\oint_\gamma \frac{1}{\zeta - \sigma_j}\dx\zeta  + \sum_{j:\sigma_j\not\in\gamma}\lambda_j\lambda_j^\top\frac{1}{2\pi i}\oint_\gamma \frac{1}{\zeta - \sigma_j}\dx\zeta \\ 
			& = \sum_{j:\sigma_j\in\gamma}\lambda_j\lambda_j^\top,
		\end{aligned}	
	\end{equation*}
	where by Cauchy-Goursat theorem the integral is zero if an eigenvalue is outside of $\gamma$. The above identity is known as the Riesz formula for spectral projections; see also \cite{kato1996perturbation}, Problem I.5.9.  If $\gamma$ is a circle passing through $\sigma_R/2$ and $\sigma_1+\sigma_R/2$, then
	\begin{equation*}
		\sum_{r=1}^R\lambda_r\lambda_r^\top = \frac{1}{2\pi i}\oint_\gamma (\zeta I-\Sigma)^{-1}\dx \zeta.
	\end{equation*}
	By Weyl's inequality, see \cite{vershynin2018high}, Theorem 4.5.3
	\begin{equation*}
		\max_{1\leq j\leq N}|\hat\sigma_j - \sigma_j| \leq \|E\|_{\rm op} < \frac{\sigma_R}{2}.
	\end{equation*}
	This shows that the eigenvalues $\hat\sigma_1\geq \dots\geq \hat\sigma_R$ are inside $\gamma$ while the eigenvalues $\hat\sigma_{R+1}\geq \dots\geq\hat\sigma_N$ are outside. Therefore, the Riesz formula also gives us
	\begin{equation*}
		\sum_{r=1}^R\hat\lambda_r\hat \lambda_r^\top = \frac{1}{2\pi i}\oint_\gamma (\zeta I - \hat\Sigma)^{-1} \dx \zeta,
	\end{equation*}
	and whence
	\begin{equation*}
		\begin{aligned}
			\hat\Lambda\hat\Lambda^\top - \Lambda\Lambda^\top & = \sum_{r=1}^R\hat\lambda_r\hat \lambda_r - \sum_{r=1}^R\lambda_r\lambda_r^\top \\
			& = \frac{1}{2\pi i}\oint_\gamma (\zeta I - \hat\Sigma)^{-1} -  (\zeta I - \Sigma)^{-1} \dx \zeta.
		\end{aligned}	
	\end{equation*}
	Since $\hat\Sigma=\Sigma + E$, it is easy to see that
	\begin{equation*}
		\begin{aligned}
			(\zeta I - \hat\Sigma)^{-1} & = \left\{(I - E(\Sigma-\zeta I)^{-1})(\zeta I  - \Sigma)\right\}^{-1} \\
			& = (\zeta I - \Sigma)^{-1}[I-E(\zeta I-E)^{-1}]^{-1} \\
			& =  (\zeta I - \Sigma)^{-1}\sum_{j=0}^\infty E^j(\zeta I-\Sigma)^{-j} \\
			& =  (\zeta I - \Sigma)^{-1} + \sum_{j=1}^\infty (\zeta I - \Sigma)^{-1}E^j(\zeta I-\Sigma)^{-j}
		\end{aligned}	
	\end{equation*}
	where we use the Neumann series which converges since 
	\begin{equation*}
		\|E(\Sigma-\zeta I)^{-1}\|_{\rm op} \leq \|E\|_{\rm op}\|(\Sigma-\zeta I)^{-1}\|_{\rm op} \leq  2\|E\|_{\rm op}/\sigma_R<1,\qquad \forall \zeta\in\gamma.
	\end{equation*}
	Therefore,
	\begin{equation*}
		\begin{aligned}
			\hat\Lambda\hat\Lambda^\top - \Lambda\Lambda^\top & = \frac{1}{2\pi i}\oint_\gamma  \sum_{j=1}^\infty (\zeta I - \Sigma)^{-1}E^j(\zeta I-\Sigma)^{-j} \dx \zeta \\
			& = \frac{1}{2\pi i}\oint_\gamma(\zeta I - \Sigma)^{-1}E(\zeta I - \Sigma)^{-1}\dx\zeta + \frac{1}{2\pi i}\sum_{j=2}^\infty\oint_\gamma (\zeta I-\Sigma)^{-1}E^j(\Sigma - \zeta I)^{-j} \dx \zeta \\
			& =: L + r_N.
		\end{aligned}	
	\end{equation*}
	Using the eigendecomposition $\Sigma = \sum_{j=1}^N\sigma_j\lambda_j\lambda_j^\top$, we obtain
	\begin{equation*}
		L = \sum_{j=1}^N\sum_{k=1}^N\frac{1}{2\pi i}\oint_\gamma\frac{\dx \zeta}{(\zeta - \sigma_j)(\zeta - \sigma_k)}\lambda_j\lambda_j^\top E\lambda_k\lambda_k^\top.
	\end{equation*}

	Since the first $R$ eigenvalues of $\Sigma$ are inside $\gamma$ and the remaining eigenvalues are zero, by Cauchy's integral formula
	\begin{equation*}
		\frac{1}{2\pi i}\oint_\gamma\frac{\dx \zeta}{(\zeta - \sigma_j)(\zeta - \sigma_k)} = \begin{cases}
			\frac{1}{\sigma_j}, & j\leq R \text{ and } k>R, \\
			\frac{1}{\sigma_k}, & j> R \text{ and } k\leq R, \\
			0, & \text{otherwise}.
		\end{cases}
	\end{equation*}
		
	This shows that
	\begin{equation*}
		\begin{aligned}
			L & = \sum_{j=1}^R\frac{1}{\sigma_j}\lambda_j\lambda_j^\top E\sum_{k=R+1}^N\lambda_k\lambda_k^\top + \sum_{j=R+1}^N\lambda_j\lambda_j^\top E\sum_{k=1}^R\frac{1}{\sigma_k}\lambda_k\lambda_k^\top \\
			& = \Lambda D^{-1}\Lambda^\top E\Lambda_{\perp}\Lambda_{\perp}^\top + \Lambda_{\perp}\Lambda_{\perp}^\top E\Lambda D^{-1}\Lambda^\top.
		\end{aligned}	
	\end{equation*}
	The result follows since the eigenvectors are orthonormal and
	\begin{equation*}
		\begin{aligned}
			|r_N| & \leq \frac{1}{2\pi}\sum_{j=2}^\infty\pi\sigma_1\sup_{\zeta\in\gamma}\|(\zeta I - \Sigma)\|^{-(j+1)}_{\rm op}\|E\|^j_{\rm op} \\
			& \leq \frac{\sigma_1}{2}\|E\|^2_{\rm op}\left(2/\sigma_R\right)^3\sum_{j=0}^\infty\left(2\|E\|_{\rm op}/\sigma_R \right)^j \\
			& = \frac{2\sigma_1}{\sigma_R^2}\|E\|_{\rm op}\frac{2\|E\|_{\rm op}/\sigma_R}{1-2\|E\|_{\rm op}/\sigma_R} \\
			& \leq \frac{8\sigma_1}{\sigma_R}\frac{\|E\|_{\rm op}^2}{\sigma_R^2},
		\end{aligned}	
	\end{equation*}
	provided that $\|E\|_{\rm op}\leq \sigma_R/4$.
\end{proof}

We also need an asymptotic expansion for the smallest singular values. Let $\Lambda_{j\perp}=[\lambda_{j,R_j+1},\dots,\lambda_{j,N_j}]\in\mathbb{R}^{N_j\times (N_j-R_j)}$ be eigenvectors corresponding to zero eigenvalues of  $\Lambda_jD_j\Lambda_j^\top$. The columns of $\Lambda_{j\perp}$ constitute an orthonormal basis of the null space of $\Lambda_jD_j\Lambda_j^\top$.
\begin{lemma}\label{lemma:perturbation}
	Suppose that Assumptions~\ref{as:orthogonal} and \ref{as:data} are satisfied, and $\sigma_{j,R_j}^2\sim \prod_{j=1}^dN_j$. Then
	\begin{equation*}
		\hat\sigma^2_{j,R_j+r} =  \lambda_r(\Lambda_{j\perp}^\top\mathbf{U}_{(j)}\mathbf{U}_{(j)}^\top \Lambda_{j\perp}) + O_P\left(N_j + \frac{\prod_{l\ne j}N_l}{N_j}\right),\qquad 1\leq \forall r\leq N_j-R_j,
	\end{equation*}
	where $\lambda_r(B)$ is the $r^{\rm th}$ largest eigenvalue of a matrix $B$.
\end{lemma}
\begin{proof}
	Note that $\lambda_r(\mathbf{Y}_{(j)}\mathbf{Y}_{(j)}^\top) = \lambda_r(M\mathbf{Y}_{(j)}\mathbf{Y}_{(j)}^\top M^\top),r\geq 1$ for every orthogonal matrix $M$. Put
	\begin{equation*}
		M_j = \begin{bmatrix}
			\Lambda_j & \Lambda_{j\perp}
		\end{bmatrix} \qquad \text{and}\qquad \Sigma_j = \begin{bmatrix}
			D_j & 0 \\
			0 & 0
		\end{bmatrix}.
	\end{equation*}
	Then $M_j\Sigma_j M_j^\top$ is the spectral decomposition of $\Lambda_jD_j\Lambda_j^\top$ and we have
	\begin{equation*}
		\begin{aligned}
			\lambda_{R_j+r}\left(\frac{\mathbf{Y}_{(j)}\mathbf{Y}_{(j)}^\top}{\prod_{j=1}^dN_j}\right) & = \lambda_{R_j+r}\left(\frac{M^\top\mathbf{Y}_{(j)}\mathbf{Y}_{(j)}^\top M}{\prod_{j=1}^dN_j}\right) \\
			& = \lambda_{R_j+r}\left( \frac{1}{\prod_{j=1}^dN_j}\begin{bmatrix}
				\Lambda_j^\top \\
				\Lambda_{j\perp}^\top 
			\end{bmatrix}(\Lambda_jD_jV_j^\top + \mathbf{U}_{(j)})(\Lambda_jD_jV_j^\top + \mathbf{U}_{(j)})^\top \begin{bmatrix}
				\Lambda_j & \Lambda_{j\perp}
			\end{bmatrix}
			\right) \\
			& =: \lambda_{R_j+r}\left(A +\varkappa A^{(1)},
			\right) \\
		\end{aligned}
	\end{equation*}
	with $\varkappa = 1/N_j$, $A =  \begin{bmatrix}
		\frac{D_j}{\prod_{j=1}^dN_j} & 0 \\
		0 & 0
	\end{bmatrix}$ and
	\begin{equation*}
		A^{(1)} = \frac{1}{\prod_{l\ne j}N_l}\begin{bmatrix}
			\Lambda_j^\top\mathbf{U}_{(j)}V_jD_j + D_jV_j^\top\mathbf{U}_{(j)}^\top \Lambda_j + \Lambda_j^\top\mathbf{U}_{(j)}\mathbf{U}_{(j)}^\top \Lambda_j & D_jV_j^\top\mathbf{U}_{(j)}^\top \Lambda_{j\perp} + \Lambda_j^\top\mathbf{U}_{(j)}\mathbf{U}_{(j)}^\top \Lambda_{j\perp} \\
			\Lambda_{j\perp}^\top\mathbf{U}_{(j)}V_jD_j + \Lambda_{j\perp}^\top \mathbf{U}_{(j)}\mathbf{U}_{(j)}^\top \Lambda_j^\top & \Lambda_{j\perp}^\top\mathbf{U}_{(j)}\mathbf{U}_{(j)}^\top \Lambda_{j\perp}
		\end{bmatrix}.
	\end{equation*}
	By \cite{onatski_ecma_2009}, Lemma 6,
	\begin{equation*}
		\left|\lambda_{R_j+r}\left(A + \varkappa A^{(1)}\right) - \varkappa\lambda_r\left(\frac{\Lambda_{j\perp}^\top\mathbf{U}_{(j)}\mathbf{U}_{(j)}^\top \Lambda_{j\perp}}{\prod_{l\ne j}N_l}\right)\right|  \leq \frac{\varkappa^2\|A^{(1)}\|^2}{0.5\sigma^2_{j,R_j}/\prod_{j=1}^dN_j - \varkappa\|A^{(1)}\|},
	\end{equation*}
	provided that $\|A^{(1)}\|/N_j<0.5\sigma^2_{j,R_j}/\prod_{j=1}^dN_j$. To see that this condition holds, note that since $\sigma^2_{j,r}\sim \prod_{j=1}^dN_j$, we have $\sigma^2_{j,r}/\prod_{j=1}^dN_j\to d_{j,r}>0,\forall r\leq R_j$. Moreover,
	\begin{equation*}
		\|A^{(1)}\|/N_j \leq \frac{4}{\prod_{j=1}^dN_j}\left[\left\|\mathbf{U}_{(j)}V_jD_j \right\|_{\rm op} + \left\|\mathbf{U}_{(j)}\mathbf{U}_{(j)}^\top  \right\|_{\rm op}\right].
	\end{equation*}
	It follows from the proof of Theorem~\ref{thm:clt} under Assumption~\ref{as:data},
	\begin{equation*}
		\|\mathbf{U}_{(j)}V_jD_j\|_{\rm op} = O_P\left(\sqrt{N_j}\sum_{r=1}^{R_j}\sigma_{j,r}\right)\qquad \text{and}\qquad \|\mathbf{U}_{(j)}\mathbf{U}_{(j)}^\top\|_{\rm op} = O_P\left(N_j + \prod_{l\ne j}N_l\right).
	\end{equation*}
	Since $\sigma^2_{j,r}\sim\prod_{j=1}^dN_j$, this shows that
	\begin{equation*}
		\|A^{(1)}\|/N_j = O_P\left(\frac{1}{\prod_{l\ne j}\sqrt{N_l}} + \frac{1}{N_j}\right) = o_P(1).
	\end{equation*}
	Therefore,
	\begin{equation*}
		\lambda_{R+r}\left(\frac{\mathbf{Y}_{(j)}\mathbf{Y}_{(j)}^\top}{\prod_{j=1}^dN_j}\right) = \lambda_r\left(\frac{\Lambda_{j\perp}^\top\mathbf{U}_{(j)}\mathbf{U}_{(j)}^\top \Lambda_{j\perp}}{\prod_{j=1}^dN_j}\right) + O_P\left(\frac{1}{\prod_{l\ne j}N_l} + \frac{1}{N_j^2}\right).
	\end{equation*}
\end{proof}

\end{document}